\documentclass{article} % Aptara syntax
\usepackage{graphicx}

\usepackage{url}
\urldef{\mailsa}\path|{serena.cerrito, adavid}@ibisc.univ-evry.fr|
\urldef{\mailsb}\path|{vfgo}@dtu.dk|

\usepackage{xspace} 

\usepackage[utf8x]{inputenc}
\usepackage{pifont}
\usepackage{amssymb}
\usepackage{amsmath}
\usepackage[usenames,dvipsnames]{xcolor}
\usepackage{amsthm}	

\usepackage{tikz}
\usetikzlibrary{positioning,shadows,arrows,trees,shapes,fit}
\usepackage{graphics}

\usepackage{multicol}

\newcommand{\sr}{rule \textbf{(SR)}\xspace}
\newcommand{\rnext}{rule \textbf{(Next)}\xspace}

\newcommand{\eun}{Rule \textbf{(ER1)}\xspace}
\newcommand{\edeux}{Rule \textbf{(ER2)}\xspace}

\newcommand{\atl}{\textbf{\textit{ATL}}\xspace}

\newcommand{\dlangle}{\langle\!\langle}

\newcommand{\drangle}{\rangle\!\rangle}
\newcommand{\rond}{\bigcirc}
\newcommand{\diaA}{\dlangle A\drangle}
\newcommand{\ddiaA}{\dlangle A\drangle\!\bigcirc\!\dlangle A\drangle}
\newcommand{\diaAn}{\dlangle A\drangle\!\bigcirc\!}
\newcommand{\crochetA}{[\![A]\!]}
\newcommand{\dcrochetA}{[\![A]\!]\!\bigcirc\![\![A]\!]}
\newcommand{\crochetAn}{[\![A]\!]\!\bigcirc\!}
\newcommand{\diams}[1]{\dlangle #1\drangle}
\newcommand{\diamsn}[1]{\dlangle #1\drangle\!\bigcirc\!}
\newcommand{\crochet}[1]{[\![#1]\!]}
\newcommand{\crochetn}[1]{[\![#1]\!]\!\bigcirc\!}
\newcommand{\ddiams}[1]{\dlangle #1\drangle\!\bigcirc\!\dlangle #1\drangle}
\newcommand{\dcrochet}[1]{[\![#1]\!]\!\bigcirc\![\![#1]\!]}
\newcommand{\until}{\,\mathcal{U}}
\newcommand{\fonctA}[2]{Real(#1,#2)}
\newcommand{\fg}{\mathsf{dec}}
\newcommand{\g}[1]{\mathsf{dec}(#1)}
\newcommand{\initf}{\eta}
\newcommand{\pretab}{\mathcal{P}^{\initf}}

\newcommand{\tabb}[1]{\mathcal{T}^{\initf}_{#1}}
\newcommand{\taba}[2]{\mathcal{T}^{#1}_{#2}}
\newcommand{\tabfin}[1]{\mathcal{T}^{#1}}
\newcommand{\settab}[1]{S^{\initf}_{#1}}

\newcommand{\m}{\mathcal{M}}
\newcommand{\mr}{\mathcal{R}}

\newcommand{\hint}{\mathcal{H}_{\initf}}
\newcommand{\treeW}{\mathcal{W}}
\newcommand{\grid}{\mathcal{F}}

\newcommand{\vectorsatstate}[1]{\dd_\mathbb{A}(#1)}
\newcommand{\vectorsCoalatstate}[2]{\dd_{#1}(#2)}
\newcommand{\Acovectsatstate}[1]{\dd^c_A(#1)}
\newcommand{\ensCoVect}[2]{\dd^c_{#1}(#2)}
\newcommand{\covect}[1]{\sigma^c_{#1}}
\newcommand{\Hmodel}{\ensuremath{\mathcal{H}}\xspace}

\newtheorem{corol}{Corollary}

\newcommand{\agents}{\ensuremath{\mathbb{A}}\xspace}
\newcommand{\Actions}{\ensuremath{\mathsf{Act}}\xspace}

\newcommand{\out}{\ensuremath{\mathsf{out}\xspace}}
\newcommand{\outm}{\ensuremath{OutPath\xspace}}

\newcommand{\Out}{\ensuremath{\mathsf{Out}\xspace}}
\newcommand{\stat}{\ensuremath{\mathsf{St}\xspace}}
\newcommand{\move}[1]{\ensuremath{\sigma_{\hspace{-2.0pt}#1}}}
\newcommand{\moveAll}{\move{\agents}}

\newcommand{\strat}{\ensuremath{\sigma}\xspace}
\newcommand{\Plays}{\ensuremath{\mathsf{Plays}}\xspace}

\newcommand{\vect}{\sigma}

\newcommand{\coop}[2][]{\langle\!\langle{#2}\rangle\!\rangle_{_{\!\mathit{#1}}}}

\newcommand{\atlall}[1][]{\mathsf{A}}
\newcommand{\atlexists}[1][]{\mathsf{E}}
\newcommand{\Sometm}[1][]{\mathbf{F}}
\newcommand{\Always}[1][]{\mathbf{G}}
\newcommand{\Until}[1][]{\mathbf{U}}
\newcommand{\Next}[1][]{\mathbf{X}}

\newcommand{\atla}{\varphi}

\newcommand{\atlsa}{\Phi}

\newcommand{\atlsb}{\Psi}

\newcommand{\De}{\ensuremath{\Delta}}
\newcommand{\Ga}{\ensuremath{\Gamma}}
\newcommand{\Logicname}[1]{\ensuremath{\mathsf{#1}}}
\newcommand{\LTL}{\Logicname{LTL}\xspace}

\newcommand{\CTLs}{\Logicname{CTL^*}\xspace}

\newcommand{\ATL}{\Logicname{ATL}\xspace}
\newcommand{\ATLs}{\Logicname{ATL^*}\xspace}

\newcommand{\CTLp}{\Logicname{CTL^{+}}\xspace}
\newcommand{\ATLp}{\Logicname{ATL^{+}}\xspace}

\newcommand{\dd}{\ensuremath{\mathsf{act}\xspace}}
\newcommand{\agi}{\ensuremath{\mathsf{i}\xspace}}
\newcommand{\aga}{\ensuremath{\mathsf{a}\xspace}}

\newcommand{\acta}{\ensuremath{\alpha\xspace}}
\newcommand{\one}{\ensuremath{\mathsf{1}\xspace}}
\newcommand{\agk}{\ensuremath{\mathsf{k}\xspace}}
\newcommand{\cgs}{\ensuremath{\mathcal{S}\xspace}}
\newcommand{\cgm}{\ensuremath{\mathcal{M}\xspace}}

\newcommand{\keyterm}[1]{\emph{#1}}
\newcommand{\set}[1]{\{{#1}\}}
\newcommand{\powerset}[1]{\mathcal{P}({#1})}
\newcommand{\lab}{\mathsf{L}} 
\newcommand{\prop}{\ensuremath{\mathsf{Prop}}\xspace}   

\renewcommand{\vec}[1]{\mathbf{#1}}

\newcommand{\bigpaper}[1]{}

\newcommand{\exptime}{\ensuremath{\mathrm{EXPTIME}\xspace}}
\newcommand{\twoexptime}{\ensuremath{\mathrm{2EXPTIME}\xspace}}

\renewcommand{\atl}{\ATL}
\newcommand{\ctls}{\Logicname{CTL^{*}}\xspace}

\newcommand{\R}{\,\mathcal{R}}

\newcommand{\st}[1]{\ensuremath{\mathbf{states}(#1)}}

\newcommand{\plays}{\ensuremath{\mathsf{Plays}\xspace}}
\newcommand{\hists}{\ensuremath{\mathsf{Hist}\xspace}}
\newcommand{\strats}{\ensuremath{\mathsf{Strat}\xspace}}
\newcommand{\last}{\ensuremath{\mathrm{last}\xspace}}
\renewcommand{\strat}{\ensuremath{F}\xspace}
\renewcommand{\acta}{\ensuremath{\sigma\xspace}}

\newcommand{\ATLps}{\Logicname{ATL^{+}_{s}}\xspace}
\newcommand{\ATLpp}{\Logicname{ATL^{+}_{p}}\xspace}
\newcommand{\gcomp}{\gamma}

\newcommand{\co}{\textsf{co}}

  \theoremstyle{plain}
\newtheorem{theorem}{Theorem}[section]
\theoremstyle{plain}
\newtheorem{lemma}{Lemma}[section]
\theoremstyle{plain}
\theoremstyle{definition}
\newtheorem{definition}{Definition}[section]
\theoremstyle{remark}

\theoremstyle{remark}
\newtheorem{remark}{Remark}[section]
\theoremstyle{plain}
\newtheorem{corollary}{Corollary}[section]
\theoremstyle{plain}
\newtheorem{proposition}{Proposition}[section]
\theoremstyle{plain}

\theoremstyle{definition}
\newtheorem{example}{Example}[section]
 
\title{Optimal Tableau Method for Constructive Satisfiability Testing and Model Synthesis in the Alternating-time Temporal Logic \ATLp}

\author{SERENELLA CERRITO - AMELIE DAVID \\
	Laboratoire IBISC - Universit\'e Evry Val-d'Essonne, France \\
	VALENTIN GORANKO \\
	Department of Philosophy, Stockholm University \\ and Department of
	Mathematics, University of Johannesburg \\ (visiting professor)
}

\begin{document}
\maketitle
% Page heads
%\markboth{S. Cerrito et al.}{Optimal Tableau-based Decision Procedure for \ATLp}

% NOTE! Affiliations placed here should be for the institution where the
%       BULK of the research was done. If the author has gone to a new
%       institution, before publication, the (above) affiliation should NOT be changed.
%       The authors 'current' address may be given in the "Author's addresses:" block (below).
%       So for example, Mr. Abdelzaher, the bulk of the research was done at UIUC, and he is
%       currently affiliated with NASA.

\begin{abstract} 
We develop a sound, complete and practically implementable tableau-based decision method for constructive satisfiability testing and model synthesis for the fragment \ATLp of the full Alternating time
 temporal logic \ATLs. The method extends in an essential way a previously  developed tableau-based decision method for \ATL and works in 2EXPTIME, which is the optimal worst-case complexity of the satisfiability problem for \ATLp. We also discuss how suitable parameterizations and syntactic restrictions on the class of input \ATLp formulae can reduce the complexity of the satisfiability problem.
\end{abstract}

%\category{F.4.1}{Mathematical Logic and Formal Languages}{Mathematical Logic}[Computational logic \and Mechanical theorem proving]
%\category{F.4.3}{Mathematical Logic and Formal Languages}{Formal Languages}[Decision problems]

%\terms{Theory}

keywords: alternating-time temporal logics, ATL+, decision procedure, model synthesis, satisfiability, tableaux

%\acmformat{Serena Cerrito, Am\'elie David, and Valentin Goranko, 2015. 
%Optimal Tableau Method for Constructive Satisfiability Testing 
%and Model Synthesis in the Alternating-time Temporal Logic \ATLp.}
% At a minimum you need to supply the author names, year and a title.
% IMPORTANT:
% Full first names whenever they are known, surname last, followed by a period.
% In the case of two authors, 'and' is placed between them.
% In the case of three or more authors, the serial comma is used, that is, all author names
% except the last one but including the penultimate author's name are followed by a comma,
% and then 'and' is placed before the final author's name.
% If only first and middle initials are known, then each initial
% is followed by a period and they are separated by a space.
% The remaining information (journal title, volume, article number, date, etc.) is 'auto-generated'.

%\begin{bottomstuff}
%This paper is an essentially extended version of \cite{CDG-IJCAR2014}.\\
%\\
%\end{bottomstuff}

%%%%%%%%%%%%%%%%%%%%%%%%%%%%%%%%%%%%%%%%
%%%%%%%%%%%%%%%%%%%%%%%%%%%%%%%%%%%%%%%%
\section{Introduction}
\label{sec:intro}
%%%%%%%%%%%%%%%%%%%%%%%%%%%%%%%%%%%%%%%%
%%%%%%%%%%%%%%%%%%%%%%%%%%%%%%%%%%%%%%%%

The Alternating-time temporal logic \ATLs was introduced and studied in~\cite{AHK02} as a multi-agent extension of the branching time temporal logic \ctls, where the path quantifiers are generalized to ``strategic quantifiers'', indexed with coalitions of agents $A$ and ranging existentially over collective strategies of $A$ and then universally over all paths (computations) enabled by the selected collective strategy of $A$. \ATLs was proposed as
 logical framework for specification and verification of properties of open systems modelled as concurrent game models, in which all agents effect state transitions collectively, by taking simultaneous actions
 at each state. The language of \ATLs allows expressing statements of the type 
\textit{``Coalition $A$ has a collective strategy to guarantee the satisfaction of the objective $\atlsa$ on every play enabled by that strategy''}.
The syntactic fragment \atl of \ATLs allows only state formulae, where all occurrences of temporal operators must be immediately preceded by strategic quantifiers. The fragment \ATLp of \ATLs extends \atl by allowing any Boolean combinations of \ATL objectives in the scope of a strategic quantifier. It is considerably more expressive than \ATL, which is reflected in the high -- \twoexptime\ -- worst-case
 complexity lower bound of the satisfiability problem for \ATLp  (inherited from the lower bound for \CTLp, see \cite{Johannsen&Lange2003}) as opposed to the \exptime-completeness of the satisfiability problem for \ATL  \cite{vanDrimmelen03,WLWW06}. 
The matching 
\twoexptime\ upper bound is provided by the automata-based method for deciding satisfiability in the full \ATLs, developed in \cite{Schewe08}.

\medskip
The contribution of this paper is the development of a sound, complete and terminating tableau-based decision method for constructive satisfiability testing of \ATLp formulae. 
We also claim that our approach is intuitive and conceptually simple, as well as practically implementable and even manually usable, despite the inherently high worst-case complexity of the problem. 
The tableau method presented here is based on the general methodology going back to \cite{Pratt80} (for PDL), \cite{Wolper85} (for LTL) and \cite{BPM83,EmHal85} (for CTL), further adapted for \ATL in~\cite{GorankoShkatov09ToCL} to which the reader is referred for more details.
A recent implementation of such a method is reported in \cite{DBLP:conf/tableaux/David13}. 
The tableau method for \ATLp is an essential extension of the one for \ATL, as it has to deal with much more complex (and computationally expensive) path objectives that can be assigned to the agents. 
It is also rather different from the above mentioned automata-based method in \cite{Schewe08}.

\medskip
The paper is structured as follows. 
In Section \ref{sec:prelim} we offer brief technical preliminaries on concurrent game models, syntax and semantics of \ATLs and \ATLp. 
Section \ref{sec:decomp} develops the technical machinery needed for the presentation of the tableau method itself in Section \ref{sec:tab}. 
In Section \ref{sec:soundness} we prove the soundness of the tableau method, whereas in Section \ref{sec:compl} we prove its completeness and demonstrate with examples how satisfying models can be extracted from the final open tableau. 
We also estimate the worst-case complexity of the procedure. In Section \ref{sec:concluding} we offer a brief comparison with the automata-based method in \cite{Schewe08}.

%%%%%%%%%%%%%%%%%%%%%%%%%%%%%%%%%%%%%%%% 
\section{Preliminaries}
\label{sec:prelim}
%%%%%%%%%%%%%%%%%%%%%%%%%%%%%%%%%%%%%%%%

We assume that the reader has basic familiarity with the branching time logic \ctls, 
see e.g. \cite{Emerson90}. 
Also, basic knowledge on \ATLs ~\cite{AHK02} and the tableaux-based decision procedure for \ATL in \cite{GorankoShkatov09ToCL}, on which this paper builds, would be beneficial. 

%%%%%%%%%%%%%%%%%%%%%%%%%%%%%%%%%%%%%%%% 
\subsection{Concurrent game models, strategies and co-strategies}
\label{subsec:CGM}
%%%%%%%%%%%%%%%%%%%%%%%%%%%%%%%%%%%%%%%%

For technical reasons that will become clearer later in the soundness and the completeness proofs, we define a more general, \emph{non-deterministic} version of the concurrent game structure with respect to~\cite{AHK02}. For the moment, we can say that the basic idea is avoiding several definitions of the notion of Realization Witness Trees for very similar structures (models, tableaux and Hintikka structures). Note that the very notion of tableau will be defined as a non-deterministic labelled CGS (see the beginning of Section \ref{sec:tab}). 

\medskip
Notation: given a set $X$, we denote the power set of $X$ by $\powerset{X}$.

\begin{definition}
A \keyterm{(non-deterministic) concurrent game structure} (CGS) is a tuple 
\[\cgs = (\agents, \stat, \{\Actions_\aga\}_{\aga\in\agents}, \{\dd_{\aga}\}_{\aga\in\agents}, \out)\]
comprising: 

%%%%
\begin{itemize}
\itemsep = 0pt
\item a finite, non-empty set of  \keyterm{players (agents)} $\agents=\{1,\dots, k\}$ 
\item a non-empty set of \keyterm{states} $\stat$,
\item a set of actions $\Actions_{\aga}\neq \emptyset$ for each $\aga\in\agents$. 

For any $A\subseteq\agents$ we denote $\Actions_A:=\prod_{\aga\in A}\Actions_{\aga}$ and use $\move{A}$ to denote a tuple from $\Actions_A$.
In particular,  $\Actions_\agents$ is the set of all possible  \keyterm{action profiles} in 
$\cgs$.
\item for each $\aga\in\agents$, a map $\dd_{\aga} : \stat \rightarrow \powerset{\Actions_{\aga}} \setminus \{ \emptyset \}$ defining for each state $s$ the actions available to $\aga$ at $s$, 
\item
a \keyterm{transition relation} $\out \subseteq \stat\times\Actions_\agents \times \stat $. 

Whenever $\langle s, \sigma_\agents, s' \rangle \in \out$, for $\moveAll = \langle\acta_{\one}, \dots, \acta_{\agk}\rangle$, then $\acta_{\aga} \in \dd_{\aga}(s)$ for every $\aga \in \agents$. 
Given a pair  $\langle s, \moveAll \rangle $, 
the set of states $s' \in \stat$ such that  $\langle s, \sigma_\agents, s' \rangle \in \out$ is denoted $\out(s,\vec{\acta}_{\agents})$ and called
the set of \keyterm{successor (outcome) states of 
$\vec{\acta}_{\agents }$ at $s$}. 

When $\out(s,\vec{\acta}_{\agents})$ is a singleton, the CGS is said to be  \keyterm{deterministic}. 
In such cases, by a slight abuse of notation we will use $\out(s,\vec{\acta}_{\agents})$ to denote a state $s'$ rather than the singleton $\{s' \}$. 
\end{itemize}
\end{definition}

\begin{definition}
\label{def:CGS}
\begin{enumerate}
\item Given a set of formulae (of some language) $\Theta$, a CGS $\cgs$ with a state space $\stat$ is \keyterm{state-labelled by $\Theta$}  if there is a mapping $l: \stat \to  \powerset{\Theta}$ assigning to every state in $\cgs$ a set of formulae from $\Theta$, called the \keyterm{label} of that state. 

\item  A \keyterm{concurrent game model}  (CGM) 
is a \emph{deterministic} CGS state-labelled by a fixed set of atomic propositions $\prop$, i.e., a tuple \\$\cgm = (\agents, \stat, \{\Actions_\aga\}_{\aga\in\agents}, \{\dd_{\aga}\}_{\aga\in\agents}, \out,\prop,\lab)$ where
\begin{itemize}
\item $(\agents, \stat, \{\Actions_\aga\}_{\aga\in\agents}, \{\dd_{\aga}\}_{\aga\in\agents}, \out)$ is a deterministic CGS, 
\item  $\prop$ is a set of \emph{atomic propositions}, and
\item $\lab:\stat\rightarrow \powerset{\prop}$ is a \emph{(state-)labelling function}.
\end{itemize}
\end{enumerate}
\end{definition}

Concurrent game models represent multi-agent discrete transition systems that function as follows. At any moment the system is in a given state, where each agent selects an action from those available to him at that state. 
All agents execute their actions synchronously and the combination of these actions together with the current state determines a transition to a unique successor state in the model. 
A \emph{play} in a CGM is an infinite sequence of subsequent successor states, i.e., an infinite sequence $s_0 s_1 ... \in \stat^\omega$ of states such that for each $i \ge 0$ there exists an action profile $\vec{\acta}_{\agents} = \langle\acta_{\one}, \dots, \acta_{\agk}\rangle$  such that $\out(s_i,\vec{\acta}_{\agents}) = s_{i+1}$. A \emph{history} is a finite prefix of a play.  
We denote by $\plays_\cgm$ and $\hists_\cgm$ respectively the set of plays and set of histories in a CGM $\cgm$. 
For a state $s \in \stat$ we define $\plays_\cgm(s)$ and $\hists_\cgm(s)$ as the set of plays and set of histories with initial state $s$. 
Given a sequence of states  $\lambda$, we denote by $\lambda_0$ its initial state, by $\lambda_i$ its $(i+1)$th state, by $\lambda_{\leq i}$ the prefix $\lambda_0 ... \lambda_i$ of $\lambda$ and by $\lambda_{\ge i}$ the suffix $\lambda_i \lambda_{i+1} ...$ of $\lambda$. When $\lambda = \lambda_0 ... \lambda_\ell$ is finite, we say that it has length $\ell$ and write $|\lambda| = \ell$. 
Further, we put $\last(\lambda) = \lambda_\ell$. 

For any coalition $A \subseteq \agents$, a given CGM $\cgm$ and state $s\in \stat$, an $A$-\emph{co-action} 
 at $s$ in $\cgm$  is a mapping $ \Actions^{c}_{A}:  \Actions_{A} \to \Actions_{\agents\setminus A}$ that assigns to every collective action of $A$ at the state $s$ a collective action at $s$ for the
 complementary coalition $\agents\setminus A$. 

We use  $\vectorsCoalatstate{A}{s}$ 
to denote the set of all A-actions that can be played by the coalition $A$ at state $s$, i.e. $\vectorsCoalatstate{A}{s} = \Pi_{\aga \in A} \vectorsCoalatstate{a}{s}$. 
We also use $\ensCoVect{A}{s}$ to denote the set of all $A$-co-actions available at state $s$ 
and $\covect{A}$ for an element of this set.

A \emph{(perfect recall) strategy} for an agent $\aga$ in $\cgm$ is a mapping $\strat_\aga: \hists_\cgm \rightarrow \Actions_{\aga}$ such that for all $h \in \hists_\cgm$ we have $\strat_\aga(h) \in \dd_{\aga}(\last(h))$. 
Intuitively, it assigns an admissible action for agent $\aga$ after any history $h$ of the game. We denote by $\strats_\cgm(\aga)$ the set of all strategies of agent $\aga$. 
A (collective) strategy of a set (\emph{coalition}) of agents $A \subseteq \agents$ is a  tuple 
$(\strat_\aga)_{\aga \in A}$ of strategies,
 one for each agent in $A$. When $A = \agents$ this is called a \emph{strategy profile}. We denote by $\strats_\cgm(A)$ the set of collective strategies of coalition $A$. 
A play $\lambda \in \plays_\cgm$ is \emph{consistent with a collective strategy $\strat_A \in \strats_\cgm(A)$} if for every $i \ge 0$ there exists an action profile 
$\vec{\acta}_{\agents} = \langle\acta_{\one}, \dots, \acta_{\agk}\rangle$ such that $\out(\lambda_i,\vec{\acta}_{\agents}) = \lambda_{i+1}$ and $\acta_\aga = \strat_\aga(\lambda_{\leq i})$ for all $\aga \in A$. 
The set of plays with initial state $s$ that are consistent with $\strat_A$ is denoted $\plays_\cgm(s,\strat_A)$.

  Likewise, a \emph{(perfect-recall)} $A$-\emph{co-strategy} in $\cgm$ for a coalition of agents $A$ (possibly reduced to just one agent $\aga$) is a mapping $\strat_{\agents\setminus A}:\hists_\cgm \times \strats_\cgm(A) \to \Actions_{\agents\setminus A}$  that assigns to each $h \in \hists_\cgm$ and every collective strategy $\strat_A \in \strats_\cgm(A)$ an $A$-co-action $\strat_{\agents\setminus A}(h,\strat_A) \in \ensCoVect{A}{last(h)}$.

%%%%%%%%%%%%%%%%%%%%%%%%%%%%%%%%%%%%%%%% 
\subsection{The logic ATL* and fragments}
\label{subsec:ATL}
%%%%%%%%%%%%%%%%%%%%%%%%%%%%%%%%%%%%%%%%
  
The logic \ATLs is a multi-agent extension of \CTLs with \emph{strategic quantifiers} $\coop{A}$ indexed with coalitions $A$ of agents. There are two types of formulae in \ATLs: \emph{state formulae}, that
 are evaluated at states, and \emph{path formulae}, that are evaluated on plays. To simplify the presentation we will work with formulae in negation normal form over a fixed set  $\prop$ of atomic
 propositions and primitive temporal operators \textit{Always} $ \Box$ and \textit{Until} $ \until$. The syntax of the full language \ATLs and its fragments \ATLp and \ATL can then be defined as follows, where 
$ l \in 
\prop \cup \{ \lnot p \mid p \in \prop\}$ is a literal,   
$\agents$ is a fixed finite set of agents and $A\subseteq \agents$:
\begin{eqnarray}
\keyterm{State formulae:} \ \varphi := & l  \mid &
 (\varphi \vee \varphi) \mid  (\varphi \wedge \varphi) \mid \diaA\Phi \mid \crochetA\Phi\\
\keyterm{\ATLs-path formulae:} \ \Phi := & \varphi \mid & \rond\Phi \mid \Box\Phi \mid 
(\Phi \until \Phi) \mid (\Phi\vee\Phi) \mid  (\Phi \wedge \Phi) \\
\keyterm{\ATLp-path formulae:}\ \Phi := & \varphi \mid & \rond\varphi \mid \Box\varphi \mid 
(\varphi \until \varphi)  \mid (\Phi\vee\Phi) \mid  (\Phi \wedge \Phi) \\ 
\keyterm{\ATL-path formulae:} \ \Phi := & & \rond\varphi \mid \Box\varphi \mid 
(\varphi \until \varphi) 
\end{eqnarray}

Note that the state formulae have the same definition but define different sets in all 3 cases. 
To keep the notation lighter, we will list the members of the set $A$  in $\coop{A}$  without using $\{ \}$. When the length of a formula is measured, $A$ will be assumed given by a bit vector.  Parentheses will be omitted whenever safe, but they will be important when conjunctions and disjunctions are composed.

Hereafter, we use $\varphi$, $\psi$, $\initf$ to denote arbitrary state formulae and $\Phi$, $\Psi$ to denote path formulae. By an \ATLp formula we will mean by default a  \emph{state} formula of \ATLp; likewise for \atl.  
We define $\top := p \lor \lnot p$, $\bot :=\lnot \top$ and the temporal operators \textit{Sometime} $\Diamond$ by 
$\Diamond\varphi := \top\until\varphi$ and \textit{Release} $\R$ by $\psi \R \varphi := \Box \varphi \lor \varphi \until (\varphi \land \psi)$. Note, that 
$\diaA\psi \R \varphi$ and $\crochetA\psi \R \varphi$ are  \ATLp state formulae. 

\CTLs can be regarded as the fragment of \ATLs where $\coop{\emptyset}$ represents the path quantifier $\forall$ 
and $\coop{\agents}$ represents $\exists$.
The semantics of \ATLs (inherited by \ATLp) is defined in a given CGM $\m$, state $s \in \m$ and a path  $\lambda$ in $\m$ just like the semantics of \ctls, with the added clauses for the strategic quantifiers: 
 \begin{itemize}
 \item $\m,s \models p$, for any proposition $p \in \prop$, iff $p \in L(s)$.
 \item $\m,s \models \neg p$ iff $\m,s \not\models p$.
 \item $\m,s \models \varphi \wedge \psi$ iff $\m,s \models \varphi$ and $\m,s \models \psi$.
 \item $\m,s \models \varphi \vee \psi$ iff $\m,s \models \varphi$ or $\m,s \models \psi$.
 \item $\m,s \models \diaA \Phi$ iff there exists an $A$-strategy $F_A$ such that, for all computations $\lambda$ consistent with $\strat_A$, $\m, \lambda \models \Phi$.
 \item $\m,s \models \crochetA \Phi$ iff there exists an $A$-co-strategy $\strat^c_A$ such that, for all computations $\lambda$ consistent with $\strat^c_A$,  $\m, \lambda \models \Phi$
 \item $\m, \lambda \models \varphi$ iff $\m, \lambda_0 \models \varphi$. 
 \item $\m, \lambda \models \rond \varphi$ iff $\m, \lambda_{\geq 1} \models \varphi$. 
 \item $\m, \lambda \models \Box\varphi$ iff for all positions $i \geqslant 0$, $\m, \lambda_{\geq i} \models \varphi$.
 \item $\m, \lambda \models \varphi \until \psi$ iff there exists a position $i \geqslant 0$ where $\m, \lambda_{\geq i} \models \psi$ and for all positions $0 \leqslant j < i$, $\m, \lambda_{\geq j} \models \varphi$.
 \item $\m, \lambda \models \Phi \wedge \Psi$ iff $\m, \lambda \models \Phi$ and $\m, \lambda \models \Psi$.
 \item $\m, \lambda \models \Phi \vee \Psi$ iff $\m, \lambda \models \Phi$ or $\m, \lambda \models \Psi$.
\end{itemize}

Valid, satisfiable and equivalent formulae in \ATLs are defined as usual. 
Here are some important equivalences in \LTL \cite{Emerson90} and in \ATLs \cite{AHK02,GorDrim06}, used further:  
\begin{itemize}
\item\label{eqLTL} $ \Box \atlsb    \equiv \atlsb \land \rond \Box \atlsb$; \ 
$\atlsa \until  \atlsb \equiv \atlsb \lor ( \atlsa \land \rond (\atlsa \until  \atlsb))$;  

\item\label{eqATL-<>} $\diaA \Box \atlsb    \equiv    \atlsb \land \ddiaA\Box \atlsb$; \ \label{eqATL-U} $\diaA \atlsa \until  \atlsb    \equiv    \atlsb \lor ( \atlsa \land \ddiaA\atlsa \until  \atlsb)$; \ 

\item\label{eqATL-[]} 
$\crochetA \Box \atlsb    \equiv    \atlsb \land \dcrochetA\Box \atlsb$; 
$\crochetA \atlsa \until  \atlsb    \equiv    \atlsb \lor ( \atlsa \land \dcrochetA \atlsa \until  \atlsb)$;

\item\label{eqATLs} 
$\crochetn{\agents}\varphi \equiv 
\neg \diamsn{\agents}\neg \varphi 
 \equiv
\diamsn{\emptyset}\varphi$; \ 
$\diams{A} \diams{B}  \atlsa \equiv  \diams{B}  \atlsa$; 

\item\label{eqATL-st} For every state formula $\atla$: $\diams{A} (\atla \land \atlsb) \equiv  \atla \land \diams{A} \atlsb$,   $\diams{A} (\atla \lor \atlsb) \equiv  \atla \lor \diams{A} \atlsb$. 
\end{itemize}

\begin{remark} It is known \cite{AHK02} that, when restricted to  \atl formulae,  the semantics above (based on perfect-recall strategies) is equivalent to the semantics based on positional (or memoryless)
 strategies, where the prescribed actions only depend on the current state, not on the whole history. This is no longer the case for \ATLp. For example, the formula  $\diams{1}\Diamond (p \wedge \diams{1} \Diamond q) \to \diams{1}(\Diamond p \wedge \Diamond q)$ in a 2-agents language is valid in the semantics with perfect-recall strategies (which can be freely composed) but not in the semantics with positional strategies (which cannot be freely composed). 
Indeed, in the concurrent game model of Figure \ref{CGS_strategy}, the antecedent of the above implication, namely $\diams{1}\Diamond (p \wedge \diams{1} \Diamond q)$,  is true at state $s_0$ no matter what strategy --  perfect-recall or positional -- is considered, whereas the consequent, namely  $\diams{1}(\Diamond p \wedge \Diamond q)$,
is true at  $s_0$ only with respect to perfect-recall strategies. 
To be more precise, with respect to the state $S_0$ only two cases of memoryless 
strategy $F$ for player 1 are possible: $F(S_0)= a$ and $F(S_0)= b$. Since the strategy is positional, these actions would be applied every time the play reaches $S_0$, and neither of them guarantees that the play will eventually visit both a state satisfying $p$ and a state satisfying $q$. On the other hand, a perfect-recall strategy  $F$  such that $F(S_0) = a$ and $F(S_0S_1S_0)=b$ guarantees the satisfaction of the objective $\Diamond p \wedge \Diamond q$.

\begin{figure}[ht]
 \centering

\noindent\scalebox{0.8}{
\begin{tikzpicture}[node distance=3cm]

%Définition des styles
  \tikzstyle{etat}=[text badly centered,circle,draw]
	\tikzstyle{lalabel}=[node distance=0.7cm]
  \tikzstyle{fleche}=[->,>=latex]

%Description des noeuds
  \node [xshift=-2cm] (M) {$\m$};
  \node[etat] (S0) {$S_0$}; \node  [lalabel] (l0) [above of=S0] {$\emptyset$} ;
  \node[etat] (S1) [right of=S0] {$S_1$} ; \node [lalabel] (l1) [above of=S1] {$\{p\}$};,
	\node[etat] (S2) [below of=S1] {$S_2$} ; \node [lalabel] (l2) [right of=S2] {$\{p\}$};
	\node[etat] (S3) [below of=S0] {$S_3$} ; \node [lalabel] (l3) [left of=S3] {$\{q\}$};
	
% Description des liens entre états

\path[fleche]
    (S0) edge [bend left] node [above] {$a,a$} (S1)
	       edge node [left,text width=0.5cm,yshift=0.2cm] {$b,a$ $b,b$}(S3)
				edge node [right] {$a,b$} (S2)
		(S1) edge [bend left] node [below] {$a,a$} (S0)
		(S2) edge node [above] {$a,a$} (S3)
		(S3) edge [loop below] node [below] {$a,a$} (S3);
\end{tikzpicture}
}
 \caption{A CGM}
 \label{CGS_strategy}
\end{figure}
\end{remark}
Here we assume that the semantics is based on  perfect-recall strategies.

The \emph{(constructive) satisfiability decision problem} for \ATLp is defined as follows:  

\textsf{Given a state formula $\atla$ in $\ATLp$, does there exist a CGM $\cgm$ and a state $s$ in $\cgm$ such that $\cgm, s \models \atla$? If so, construct such a satisfying pair $(\cgm,s)$.} 

\begin{remark} There are three variants of the satisfiability problem: \emph{tight}, where it is assumed that all agents in the model are mentioned in the formula, \emph{loose} where just one additional agent, not mentioned in the formula is allowed in the model, and \emph{general}, where any number of additional agents, not mentioned in the formula, are allowed in the model.  
These variants are really different, but the general satisfiability is immediately reducible to the loose satisfiability, by adding just one extra agent $\aga$ to the language.
 Furthermore, this extra agent can be easily added superfluously to the formula, e.g., by adding a conjunct $\diamsn{\aga}\top$, thus reducing loose to tight satisfiability. So, hereafter we only consider the tight satisfiability version. 
For further details and discussion on this issue, see e.g., \cite{WLWW06,GorankoShkatov09ToCL}.
\end{remark}

%%%%%%%%%%%%%%%%%%%%%%%%%%%
\section{Decomposition and closure of \ATLp formulae}
\label{sec:decomp}

We partition the set of \ATLp formulae into \keyterm{primitive} and \keyterm{non-primitive} formulae. The primitive formulae are $\top,\bot$, the literals
and all  \ATLp \emph{successor formulae}, of the form $\diaA\rond\psi$ or $\crochet{A'}\rond\psi$, where $A \subseteq\agents$ and $A'\subset \agents$, 
each with \emph{successor component} $\psi$. 
The non-primitive formulae are classified as $\alpha$-, $\beta$- and $\gamma$-formulae.
An $\alpha$-formula in our syntax is a conjunction $\varphi \land \psi$ with (conjunctive) \emph{$\alpha$-components} $\varphi$ and $\psi$, plus the formulae of the form $\crochetn{\agents}\psi$ whose $\alpha$-components are both $\diamsn{\emptyset}\psi$; 
a $\beta$-formula is a disjunction $\varphi \lor \psi$ with (disjunctive) \emph{$\beta$-components} $\varphi$ and $\psi$. 
The rest of the non-primitive formulae are classified as $\gamma$-formulae. 
That is,  a $\gamma$-formula is one of the form $\crochet{A} \Phi$ or 
$\diaA \Phi$, where $\Phi$ is an \ATLp path formula whose main operator is not $\rond$ and $A\neq\agents$.

The need of introducing the new category of $\gamma$-formulae, w.r.t. the partition of non-primitive formulae into $\alpha$- and  $\beta$- classes done in  
\cite{GorankoShkatov09ToCL} is the following. In ATL each strategic quantifier is necessarily followed by a temporal operator, and, for instance $\diaA \Box \varphi$ can be seen as an $\alpha$-formula while $\diaA \varphi \until \psi$ can be  seen as  a $\beta$-formula. However, typical state formulae in \ATLp have the form $\diaA (\Phi_1 \vee \Phi_2)$, $\diaA (\Phi_1 \wedge \Phi_2)$, $\crochetA (\Phi_1 \wedge \Phi_2)$,$\crochetA (\Phi_1 \vee \Phi_2)$. Now, these four types of formulae cannot reasonably be classified as $\alpha$- or  $\beta$- formulae. Note, in particular, that the strategic quantifier $\diaA$ in general distributes neither on $\vee$ nor on $\wedge$ and the same applies to $ \crochetA$.   Thus, a new
category of $\gamma$-formulae is created, containing also $\diaA \Box \varphi$ and  $\diaA \varphi \until \psi$ as special cases, and needing a special analysis.

Thus $\alpha$- and $\beta$-formulae will be decomposed in the tableau as usual, while the case of $\gamma$-formulae $\diaA \Phi$ and $\crochetA \Phi$ is special and needs extra work, because their tableau decomposition will depend on the structure of $\Phi$. 

%%%%%%%%%
\subsection{$\gamma$-decomposition and $\gamma$-components of $\gamma$-formulae}

We denote the set of \ATLp state formulae by \ATLps and the set of \ATLp path formulae by \ATLpp. We will define a \emph{$\gamma$-decomposition function}  
$\fg: \ATLpp \to {\cal P}(\ATLps \times \ATLpp)$ with the following intuitive meaning: 
for any $\Phi\in  \ATLpp$ and pair $\langle \psi, \Psi \rangle \in \fg(\Phi)$, $\psi$ is a state formula true at the current state and $\Psi$ is a path formula expressing what must be true at the next state of a possible play starting at the current state.
Thus, the set $\fg(\Phi)$ is interpreted as a disjunction describing all possible `types of paths' starting from the current state and satisfying $\Phi$.

We emphasize that, although the domain of $\fg$ is the whole set
$\ATLpp$, $\fg$ will only be used to analyse $\Phi$ in the contexts $\diaA \Phi$ and $\crochetA \Phi$ (where $\Phi$ does not have $\rond$ as main connective), and, as we will see, its role is just auxiliary to the rewriting of the (always quantified) $\gamma$-formulae in a special form useful to obtain a key 
ty of our tableau calculus (see further Lemma \ref{lem:gamma}). 

\smallskip
\noindent \textit{Base cases:}

$\star$ \ $\g{\varphi} = \{\langle \varphi, \top \rangle\}$, $\g{\rond\varphi} = \{\langle \top,\varphi \rangle\}$ for any \ATLp state formula $\varphi$. 

The other base cases derive from the well-known LTL equivalences listed in \ref{eqLTL}: 

$\star$ \ $\g{\Box\varphi} = \{\langle \varphi,\Box\varphi\rangle \}$ 

$\star$ \ $\g{\varphi\until\psi} = \{\langle \varphi,\varphi\until\psi\rangle,\langle \psi,\top \rangle \}$.

\smallskip
\noindent\textit{Recursive steps:}

$\star$ \ 
 $\fg(\Phi_1 \wedge\Phi_2)= \fg(\Phi_1) \otimes \fg(\Phi_2)$, where \\
$\fg(\Phi_1) \otimes \fg(\Phi_2)$ :=
$\{ \langle \psi_i \wedge \psi_ j, \Psi_i \wedge \Psi_j \rangle  \; \mid \; \langle \psi_i, \Psi_i \rangle \in \fg(\Phi_1),  \langle \psi_j, \Psi_j \rangle \in \fg(\Phi_2)
\}$. 

\smallskip
$\star$ \ 
 $\fg(\Phi_1 \vee \Phi_2) = \fg(\Phi_1) \cup  \fg(\Phi_2) \cup (\fg(\Phi_1) \oplus \fg(\Phi_2))$, where \\
$\fg(\Phi_1) \oplus \fg(\Phi_2) :=$ \\
$\{ \langle \psi_i \wedge \psi_ j, \Psi_i \vee \Psi_j \rangle  \; \mid \; \langle \psi_i, \Psi_i \rangle \in \fg(\Phi_1), \;  \langle \psi_j, \Psi_j \rangle \in \fg(\Phi_2 ), \; \Psi_i \not = \top,  \Psi_j \not = \top 
\}$.

Note that the operations $\otimes$ and $\oplus$ are associative, up to logical equivalence.

\smallskip
The conjunctive case should be clear: every path satisfying $\Phi_1 \wedge\Phi_2$ combines  
a type of path satisfying $\Phi_1$ with a type of path satisfying $\Phi_2$.  
To understand the disjunctive case, note that, as it will be seen in Section  \ref{sec:tab}, the construction of the tableau is step-by-step.
Therefore, for a given prestate under construction, when we have a formula of the form $\diaA(\Phi_1\lor\Phi_2)$, where, for instance $\Phi_1=\Box\varphi_1$ and $\Phi_2 =  \Box\varphi_2$, we do not know in advance which of $\Box\varphi_1$ or $\Box\Phi_2$ would be completed; so it is important to keep both possibilities at the current state, if possible. 
This idea is expressed by the use of $\fg(\Phi_1) \oplus \fg(\Phi_2)$ in the above union, where we keep both disjuncts true at the present state and delay the choice. 
This is why the state formulae $\psi_i$ and $\psi_j$ are connected by $\wedge$ but the path formulae $\Psi_i$ and $\Psi_j$ are connected by $\vee$.  
 Moreover, the $\oplus$ operation avoids the construction of 
a pair $ \langle \psi_i \wedge \psi_ j, \Psi_i \vee \Psi_j \rangle$
where either $\Psi_i$ or $\Psi_j$ is $\top$, because that case  would already be included in $\g{\Phi_1}$ or in $\g{\Phi_2}$.
The three cases for paths satisfying the disjunction $\Phi_1 \vee \Phi_2$ can be illustrated by the picture in Figure \ref{fig:disj}.

\begin{figure}
\centering
\begin{tikzpicture}
  \begin{scope}
    \draw (1.8,1.6) node{\underline{$\fg(\Phi_1)$}};
  \draw (0,0) node {$\bullet$} node [below] {$\varphi_1$};
  \draw [dashed] (0,0) -- (1.5,1); \draw (1.5,1) node [right] {$\Phi_1$}; \draw [dashed] (2.1,1) -- (3,1);
  \draw [dashed] (0,0) -- (1.5,0.5); \draw (1.5,0.5) node [right] {$\Phi_1$}; \draw [dashed] (2.1,0.5) -- (3,0.5);
  \draw [dashed] (0,0) -- (1.5,0); \draw (1.5,0) node [right] {$\Phi_1$}; \draw [dashed] (2.1,0) -- (3,0);
  \draw [dashed] (0,0) -- (1.5,-0.5); \draw (1.5,-0.5) node [right] {$\Phi_1$}; \draw [dashed] (2.1,-0.5) -- (3,-0.5);
  \draw [dashed] (0,0) -- (1.5,-1); \draw (1.5,-1) node [right] {$\Phi_1$}; \draw [dashed] (2.1,-1) -- (3,-1);
  \end{scope}

  \begin{scope}[xshift=4cm]
    \draw (1.8,1.6) node{\underline{$\fg(\Phi_2)$}};
  \draw (0,0) node {$\bullet$} node [below] {$\varphi_2$};
  \draw [dashed] (0,0) -- (1.5,1); \draw (1.5,1) node [right] {$\Phi_2$}; \draw [dashed] (2.1,1) -- (3,1);
  \draw [dashed] (0,0) -- (1.5,0.5); \draw (1.5,0.5) node [right] {$\Phi_2$}; \draw [dashed] (2.1,0.5) -- (3,0.5);
  \draw [dashed] (0,0) -- (1.5,0); \draw (1.5,0) node [right] {$\Phi_2$}; \draw [dashed] (2.1,0) -- (3,0);
  \draw [dashed] (0,0) -- (1.5,-0.5); \draw (1.5,-0.5) node [right] {$\Phi_2$}; \draw [dashed] (2.1,-0.5) -- (3,-0.5);
  \draw [dashed] (0,0) -- (1.5,-1); \draw (1.5,-1) node [right] {$\Phi_2$}; \draw [dashed] (2.1,-1) -- (3,-1);
  \end{scope}

  \begin{scope}[xshift=8cm]
    \draw (1.8,1.6) node{\underline{$\fg(\Phi_1) \oplus \fg(\Phi_2)$}};
  \draw (0,0) node {$\bullet$} node [below] {$\varphi_2$} node [above] {$\varphi_1$};
  \draw [dashed] (0,0) -- (1.5,1); \draw (1.5,1) node [right] {$\Phi_1$}; \draw [dashed] (2.1,1) -- (3,1);
  \draw [dashed] (0,0) -- (1.5,0.5); \draw (1.5,0.5) node [right] {$\Phi_2$}; \draw [dashed] (2.1,0.5) -- (3,0.5);
  \draw [dashed] (0,0) -- (1.5,0); \draw (1.5,0) node [right] {$\Phi_2$}; \draw [dashed] (2.1,0) -- (3,0);
  \draw [dashed] (0,0) -- (1.5,-0.5); \draw (1.5,-0.5) node [right] {$\Phi_1$}; \draw [dashed] (2.1,-0.5) -- (3,-0.5);
  \draw [dashed] (0,0) -- (1.5,-1); \draw (1.5,-1) node [right] {$\Phi_1$}; \draw [dashed] (2.1,-1) -- (3,-1);
  \end{scope}
 \end{tikzpicture}
\caption{The three cases for disjunctive path objectives in a $\gamma$-formula.}
\label{fig:disj} 
\vspace{-2mm}
\end{figure}
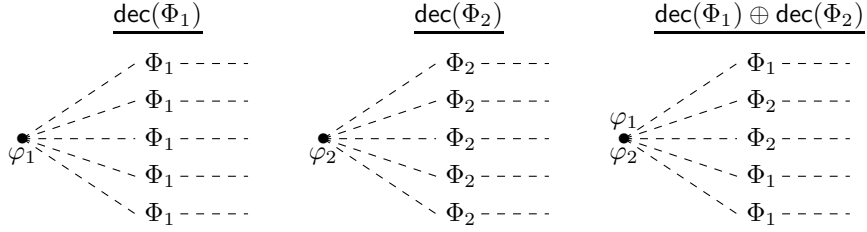

Now, let $\zeta= \diaA \Phi$ or $\zeta= \crochetA \Phi$ be a $\gamma$-formula to be decomposed.
Each pair  $\langle \psi, \Psi\rangle \in \fg(\Phi)$ is then converted to a \emph{$\gamma$-component} $\gcomp(\psi, \Psi)$ as follows: 
\begin{eqnarray}                          
    \gcomp(\psi, \Psi) = \psi & &  \text{ if } \Psi = \top \\                                                                                   
    \gcomp(\psi, \Psi) = \psi \wedge \ddiaA\Psi  &  &\text{ if } \zeta \text{ is of the form } \diaA\Phi,    \\                                                                                                                                                                                                   
  \gcomp(\psi, \Psi) = \psi \wedge \dcrochetA\Psi  & &\text{ if } \zeta \text{ is of the form } \crochetA\Phi                                                                                           
\end{eqnarray}

Thus, the role of $\fg$ is to associate with any $\gamma$-formula $\zeta$
 a set of formulae that are simpler in some precise sense, viz. its \emph{$\gamma$-components}, so that $\zeta$ is equivalent to the disjunction of its $\gamma$-components.
This key property is item 3 of the next lemma (the first two items being just auxiliary claims), and it is the core distinction between the proposed calculus for \ATLp in this work and the tableau calculus for \ATL in \cite{GorankoShkatov09ToCL}.

\begin{lemma}
\label{lem:gamma}
For any $\gamma$-formula $\Theta = \diaA \Phi$ or $\Theta = \crochetA \Phi$ of \ATLp, the following properties hold: 

\begin{enumerate}
\item $ \Phi \equiv \bigvee \{ \psi \land \rond \Psi \mid \langle \psi,\Psi \rangle \in \g{\Phi} \}$.

\item $\diaA\Phi \equiv \bigvee \{\diaA (\psi \land \rond \Psi) \mid \langle \psi,\Psi \rangle \in \g{\Phi} \}$, and respectively, 

 $\crochetA\Phi \equiv \bigvee \{\crochetA (\psi \land \rond \Psi) \mid \langle \psi,\Psi \rangle \in \g{\Phi} \}$.

\item $\Theta \equiv \bigvee \{\gamma(\psi,\Psi) \mid \langle \psi,\Psi \rangle \in \g{\Phi} \}$.\end{enumerate}

\end{lemma}

\begin{proof}
Claim 1. We will prove the claim by induction on the path formula $\Phi$. It is equivalent to 
the following property $P(\Phi)$:

\begin{quote}
For every CGM $\m$ and a play $\lambda$ in it, $ \m,\lambda \models \Phi$ iff there exists $ \langle \psi,\Psi \rangle \in \g{\Phi}$ such that  $\m,\lambda_0 \models \psi$ and $ \m,\lambda_{\geq 1} \models \Psi$. 
\end{quote}

The base cases are $\Phi = \varphi$, $\Phi = \rond\varphi$, $\Phi = \Box\varphi$ and 
 $\Phi = \varphi \until \psi$.  For each of these the property $P(\Phi)$ follows immediately from the definitions of $\fg$ and $\gamma$-components and -- for the latter two cases -- the well-known fixed point \LTL equivalences for the temporal operators, listed at the end of Section \ref{subsec:ATL}. 

\medskip
For the inductive steps there are two cases to consider: 

\smallskip
Case 1:
$\Phi = \Phi_1 \wedge \Phi_2$.  \ We have that: \\ 
$\m,\lambda\models\Phi$ iff \\
$\m,\lambda\models \Phi_1$ and $\m,\lambda\models\Phi_2$, iff (by the induction hypothesis):  
\begin{enumerate}
\item[(i)] 
 there is $\langle \psi_1,\Psi_1 \rangle \in \g{\Phi_1}$, such that $\m,\lambda_0\models \psi_1$ and $\m, \lambda_{\geq 1} \models \Psi_1$, \\
 and 
\item[(ii)] 
 there is $\langle \psi_2,\Psi_2 \rangle \in \g{\Phi_2}$ such that $\m,\lambda_0\models \psi_2$ and $\m, \lambda_{\geq 1} \models \Psi_2$.
\end{enumerate}
These two are the case iff \\ 
$\m,\lambda_0\models \psi_1 \wedge \psi_2$ and $\m, \lambda_{\geq 1} \models \Psi_1 \wedge \Psi_2$, iff \\
$\m,\lambda_0 \models \psi$ and $\m, \lambda_{\geq 1} \models \Psi$ where $\psi = \psi_1\wedge\psi_2$, $\Psi = \Psi_1 \wedge \Psi_2$ and $\langle \psi,\Psi \rangle \in \g{\Phi}$. This completes the proof of $P(\Phi)$ for $\Phi = \Phi_1 \wedge \Phi_2$.

\smallskip
Case 2:  $\Phi = \Phi_1 \vee \Phi_2$. We have that $\m,\lambda \models\Phi$ iff $\m,\lambda\models \Phi_1$ or $\m,\lambda\models\Phi_2$. By inductive hypotheses for  $\Phi_1$ and $\Phi_2$ and from the fact that $\fg(\Phi_1) \cup \fg(\Phi_2) \subseteq \fg(\Phi)$, we obtain the direction from left to right in property $P(\Phi)$. For the converse direction, we only need to consider the case that does not follow directly from the inductive hypotheses for  $\Phi_1$ and $\Phi_2$, viz. when there exists $ \langle \psi,\Psi \rangle \in  (\fg(\Phi_1) \oplus \fg(\Phi_2))$ such that  $\m,\lambda_0 \models \psi$ and $ \m, \lambda_{\geq 1} \models \Psi$. In this case, $\psi = \psi_1 \wedge \psi_ 2$ and $\Psi = \Psi_1 \vee \Psi_ 2$ for some $\langle \psi_1, \Psi_1 \rangle \in \fg(\Phi_1)$ and 
$\langle \psi_2, \Psi_2 \rangle \in \fg(\Phi_2 )$ such that $\Psi_1 \not = \top,  \Psi_2 \not = \top$. Suppose 
$ \m, \lambda_{\geq 1} \models \Psi_{1}$.  Since we also have $\m,\lambda_0 \models \psi_{1}$, by the inductive hypothesis for $\Phi_1$, it follows that $\m,\lambda\models \Phi_1$, hence $\m,\lambda\models \Phi$. Likewise, when $ \m, \lambda_{\geq 1} \models \Psi_{2}$.

\medskip
Claim 2. We will consider the case of $\Theta = \diaA \Phi$;  the case of $\crochetA\Phi$ is analogous. The implication from right to left of the claimed equivalence follows from Claim 1 and the monotonicity of $\diaA$ (in sense that if $\Psi \models \Phi$ then $\diaA \Psi \models \diaA \Phi$). 
For the converse direction, first recall that every \ATLp path formula  $\Xi$ is a positive Boolean combination of sub-formulae of the types $\varphi, \rond\varphi, \Box\varphi, \varphi\until\psi$ where $\varphi, \psi$ are \ATLp state formulae.  
Let the set of these sub-formulae of $\Xi$ be $S(\Xi)$. Now, we introduce some ad hoc notation for special sets of formulae in $S(\Xi)$ and their sub-formulae:  

\begin{itemize}
\item $L(\Xi)$ is the set of all state formulae in $S(\Xi)$; 
\item $N(\Xi) := \{\varphi \mid \rond\varphi \in S(\Xi)\}$; 
\item $B(\Xi) := \{\varphi \mid \Box\varphi \in S(\Xi)\}$; 
\item $U(\Xi) := \{\varphi\until\psi \mid \varphi\until\psi \in S(\Xi)\}$; 
\item $U_{1}(\Xi) := \{\varphi \mid \varphi\until\psi \in S(\Xi)\}$; 
\item $U_{2}(\Xi) := \{\psi \mid \varphi\until\psi \in S(\Xi)\}$; 
\end{itemize}

Without loss of generality we can assume that $\Phi$ is in a DNF over the set of formulae in $S(\Phi)$, i.e. $\Phi = \Phi_{1} \lor \ldots \lor \Phi_{m}$, where each  $\Phi_{i}$ is a conjunction of formulae from $S(\Phi)$.  

Now, to prove the implication from left to right, take any CGM $\m$ and state $s$ in it, such that $\m,s\models \diaA \Phi$.
Take and fix any collective strategy $\strat_A$ of $A$ such that  $\m,\lambda\models\Phi$ for every play $\lambda$ starting at $s$ and consistent with $\strat_A$. We denote that set of plays by  $\outm(s, \strat_A)$. 
Then for every play $\lambda\in \out_{\m}(s,\sigma_{A})$ we have that $\m,\lambda\models \Phi_{i}$ for some $i=1,\ldots,m$. Without restriction of generality we can assume that the set of $\Phi_{i}$'s for which there is a $\lambda\in \out_{\m}(s,\sigma_{A})$ such that 
 $\m,\lambda\models \Phi_{i}$ is $\{\Phi_{1}, \ldots, \Phi_{n}\}$ for some $n \leq m$. 
 
 Let $\Phi_{i}$ be any of these. 
 We will associate with it a pair $\langle \psi_{i},\Psi_{i} \rangle \in \g{\Phi_{i}}$ as follows. First, note that all formulae from $L(\Phi_{i})$ and $B(\Phi_{i})$ are true at $s$. Further, let $E_{i}(s)$ be the subset of those formulae from 
 $U_{2}(\Phi_{i})$ which are true at $s$ in $\m$. Thus, for every play $\lambda\in \outm(s,\strat_{A})$ satisfying $\Phi_{i}$ the following hold: 

\begin{itemize}
\label{list}
\item[i)]  $\m,\lambda\models\varphi$ for each $\varphi \in L(\Phi_{i})$. 
\item[ii)]  $\m,\lambda\models\rond \varphi$ for each $\rond \varphi \in S(\Phi_{i})$. 
\item[iii)]  $\m,\lambda\models \varphi \land \rond \Box \varphi$ for each $\Box \varphi \in S(\Phi)$. 
\item[iv)]  $\m,\lambda\models \psi$ for each $\psi \in E_{i}(s)$. 
\item[v)]  $\m,\lambda\models \varphi \land \rond \varphi\until\psi$ for each $\psi \in U_{2}(\Phi_{i}) - E_{i}(s)$. 
\end{itemize}

Now, suppose $\Phi_{i} = \Psi_{i1} \land \ldots  \land \Psi_{ik}$ for some 
$\Psi_{i1} \land \ldots  \land \Psi_{ik} \in S(\Phi)$. Then 
$\fg(\Phi_{i}) = \fg(\Psi_{i1}) \otimes \ldots \otimes \fg(\Psi_{ik})$. (Recall that the operations $\otimes$ and $\oplus$ are associative, up to logical equivalence, so there is no need to put parentheses.)  Thus, for every $\langle \psi,\Psi \rangle \in \g{\Phi_{i}}$,  $\psi$ is a conjunction of all formulae from $L(\Phi_{i}) \cup  B(\Phi_{i})$ and, for every conjunct of $\Phi_{i}$ of the type $\varphi\until\psi$, at least one of the respective    
formulae coming from $U_{1}(\Phi_{i})$ and $U_{2}(\Phi_{i})$. We now select $\langle \psi_{i},\Psi_{i} \rangle \in \g{\Phi_{i}}$ to be the one where the conjuncts taken from $U_{2}(\Phi_{i})$ are exactly those in $E_{i}(s)$. Then we claim that  for every play $\lambda\in \outm(s,\strat_{A})$ satisfying $\Phi_{i}$, it is the case that 
$\m,\lambda\models \psi_{i} \land \rond \Psi_{i}$. Indeed, this follows from the list of properties  (i - v) above and from the definition of $\fg(\Psi_{i1}) \otimes \ldots \otimes \fg(\Psi_{ik})$. 
Note further, that if $\Psi_{i}$ above is $\top$, then 
$\m,\lambda\models \psi_{i} \land \rond \Psi_{i}$ for all paths $\lambda$ starting at $s$, so we can assume without affecting what follows that no $\Psi_{i}$ above is $\top$. 

After having selected such a pair $\langle \psi_{i},\Psi_{i} \rangle \in \g{\Phi_{i}}$ for each 
$\Phi_{i} \in \{\Phi_{1}, \ldots, \Phi_{n}\}$, we use these $n$ pairs (or, those of them for which $\Psi_{i} \neq \top$) to construct the pair  $\langle \psi,\Psi \rangle \in \g{\Phi_{1}} \oplus \ldots \oplus \g{\Phi_{n}}$ such that 
$ \psi = \psi_{1}\land \ldots \land \psi_{n}$ and $ \Psi =  \Psi_{1}\lor \ldots \lor \Psi_{n}$.  

Finally, we claim that, by virtue of the construction, $\m,\lambda\models \psi \land \rond \Psi$ 
 for every play $\lambda\in \outm(s,\strat_{A})$ satisfying $\Phi$. Therefore, 
the strategy $\sigma_{A}$ is a witness of the truth of  
$\m,s\models \diaA (\psi \land \rond \Psi)$, hence   
$\m,s\models \bigvee \{\diaA (\psi \land \rond \Psi) \mid \langle \psi,\Psi \rangle \in \g{\Phi}\}$.
This completes the proof of the implication left-to-right of Claim 2. 

\medskip
Claim 3. This claim follows easily from Claim 2 by noting that: 

\begin{itemize}
\item 
$\diaA(\psi \land \rond \Box \Psi) \equiv \psi \land \diaAn \Box \Psi \equiv \psi \land \ddiaA \Box \Psi$, because $\psi$ is a state formula. Note that the second equivalence is due to the fact that the semantics of $\coop{}$ is based on perfect recall strategies, that can be composed. More precisely, it essentially assumes that any strategy at $s$ ensuring that every successor satisfies $\diaA \Box \Psi$ can be composed with the family of strategies, one for every such successor $s'$ 
witnessing the truth of $\diaA \Box \Psi$ on all plays starting at $s'$, into one perfect recall strategy that guarantees the truth of $\rond \Box \Psi$ on all plays starting at $s$.  (This, in general, cannot be done if only positional strategies are considered, as those applied at the different successors of $s$ may interfere with each other.)

\item Likewise, $\crochetA(\psi \land \rond \Box \Psi) \equiv \psi \land \crochetAn \Box \Psi \equiv \psi \land \dcrochetA\Box \Psi$. 
\end{itemize}
Therefore, for each $\langle \psi,\Psi \rangle \in \g{\Phi}$ the $\gamma$-component $\gamma(\psi,\Psi)$ is equivalent to its respective disjunct on the right hand side of Claim 2. 
\qed
\end{proof}

\begin{example}
\label{run-ex}
We will use two syntactically similar, yet different, running examples: 
\[\theta=\diams{1}(p \until q \vee \Box q) \wedge \diams{2}(\Diamond p \wedge \Box \neg q)\] 
and 
\[\vartheta = \diams{1}(p \until q \vee \Box q) \wedge \crochet{2}(\Diamond p \wedge \Box \neg q).\]

First, we consider $\theta$. It is an $\alpha$-formula with conjunctive components \\
$\theta_1 = \diams{1}(p \until q \vee \Box q)$ and $\theta_2 = \diams{2}(\Diamond p \wedge \Box \neg q)$. 

Further, $\theta_1$ is a $\gamma$-formula of the form $\diaA\Phi$ where the main connective of $\Phi$ is $\vee$. 
So, $\g{\theta_1} = \g{p \until q} \cup \g{\Box q} \cup (\g{p \until q} \oplus \g{\Box q}) $, where 
$\g{p \until q} = \{\langle p,p \until q\rangle , \langle q, \top \rangle\}$ and 
$\g{\Box q} = \{\langle q, \Box q\rangle\}$. 

Thus, $\g{\theta_1} = \{\langle p,p \until q\rangle, \langle q, \top \rangle, \langle q,\Box q  \rangle, \langle p \wedge q, p \until q \vee \Box q\rangle\}$, hence \\ 
$\theta_1 \equiv (p\wedge\ddiams{1}p \until q) \vee q \vee (q \wedge \ddiams{1}\Box q) \vee (p\wedge q \wedge \ddiams{1}(p \until q \vee \Box q))$.

\smallskip
Likewise, $\theta_2$ is a $\gamma$-formula of the form $\diaA\Phi$ and the main connective of $\Phi$ is $\wedge$. 
So $\g{\theta_2} = \g{\Diamond p} \otimes \g{\Box\neg q}$, with 
$\g{\Diamond p} = \{\langle T, \Diamond p \rangle, \langle p,T \rangle\}$ and
$\g{\Box \neg q} = \{\langle \neg q, \Box\neg q\rangle\}$. 

Thus, 
$\g{\theta_2} = \{\langle \top \wedge \neg q, \Diamond p \wedge \Box\neg q\rangle,\langle p \wedge \neg q, \top \wedge \Box\neg q \rangle\} $

\hspace{2.2cm}$= \{\langle \neg q, \Diamond p \wedge \Box\neg q\rangle,\langle p \wedge \neg q, \Box\neg q \rangle\}$ and
 
$\theta_2 \equiv (\neg q \wedge \ddiams{2}(\Diamond p\wedge\Box\neg q)) \vee (p \wedge \neg q\wedge\ddiams{2}\Box\neg q)$.

For $\vartheta$, the $\gamma$-decomposition is similar, we only replace $\diams{2}$ by $\crochet{2}$. Thus, we obtain \\
$\vartheta_1 \equiv (p\wedge\ddiams{1}p \until q) \vee (q) \vee (q \wedge \ddiams{1}\Box q) \vee (p\wedge q \wedge \ddiams{1}(p \until q \vee \Box q))$ \\and\\ $\vartheta_2 \equiv (\neg q \wedge \dcrochet{2}(\Diamond p\wedge\Box\neg q)) \vee (p \wedge \neg q\wedge\dcrochet{2}\Box\neg q)$.
\end{example}

\bigskip
The \keyterm{closure} $cl(\psi)$ of an \ATLp state formula $\psi$ is the least set of \ATLp formulae such that $\psi,\top,\bot \in cl(\psi)$ and $cl(\psi)$ is closed under taking of successor-, $\alpha$-, $\beta$- and $\gamma$-components. For any set of state formulae $\Ga$ we define 
\[cl(\Ga) := \bigcup \{ cl(\psi) \mid \psi \in \Ga \}.\]
We denote by $|\psi |$ the length of $\psi$ and by $\|\Ga \|$ the  cardinality of $\Ga$.

\begin{example} 
\label{run-ex2}
The construction of the closure of the formula $\theta$ from Example \ref{run-ex} is given in Figure \ref{closureFig}. Each node of the tree represents an element of the closure. 
Children of an interior node are respective components of the parent formula,  according to the definition of closure.

\begin{figure}[h]
\scalebox{0.6}{
\begin{tikzpicture}[->,>=stealth, thick,level 1/.style={sibling distance=10cm},level 2/.style={sibling distance=3cm}, level 3/.style={sibling distance=1.7cm}] 
\node {$\theta$}
    child{ node {$\theta_1 = \diams{1}(p \until q \vee \Box q)$} child{ node {$q$}}		
				child[xshift=-0.5cm]{ node {$p \land \ddiams{1}p\until q$}
						child{ node {$p$}}
						child{ node {$\ddiams{1}p \until q$}
							child {node {$\diams{1}p \until q$}}
						}
				}
				child [xshift=-0.2cm]{ node {$q \land \ddiams{1}\Box q$}
					child{ node {$\ddiams{1}\Box q$}
						child{ node {$\diams{1}\Box q$}}
					}
				}
				child[xshift=0.8cm]{ node {$p \land q \land \ddiams{1}(p \until q \lor \Box q)$}
					child {node {$q \land \ddiams{1}(p \until q \lor \Box q)$}
						child { node {$\ddiams{1}(p \until q \lor \Box q)$}}
					}
				}  
		}
    child{ node {$\theta_2 = \diams{2}(\Diamond p \wedge \Box \neg q)$}
            child[xshift=2cm]{ node  {$\neg q \land \ddiams{2}(\Diamond p \land \Box \neg q)$} 
							child[xshift=-1cm]{ node {$\neg q$}}
							child{ node {$\ddiams{2}(\Diamond p \land \Box \neg q)$}}
            }
            child[xshift=3.5cm]{ node {$p \land \neg q \land \ddiams{2}\Box\neg q$}
							child{ node {$\neg q \land \ddiams{2}\Box\neg q$}
							  child{ node {$\ddiams{2}\Box\neg q$}
									child{ node {$\diams{2}\Box\neg q$}}
								}
							}
						}
		}
; 
\end{tikzpicture}
}
\caption{Closure of the formula $\theta = \diams{1}(p \until q \vee \Box q) \wedge \diams{2}(\Diamond p \wedge \Box \neg q)$}
\label{closureFig}
\end{figure}

The closure of $\vartheta$ is similar to the one of $\theta$ except that every $\crochet{2}$ is replaced by $\diams{2}$.
\end{example}

\begin{lemma}
\label{lem:closure}
For any \ATLp state formula $\varphi$,  
$\| cl(\varphi) \| < 2^{|\varphi|^{2}}$. 
\end{lemma}

\begin{proof}
Every formula in $cl(\varphi)$ has length less than $2|\varphi|$ and is built from symbols in $\varphi$, so there can be at most 
$|\varphi|^{2|\varphi|} = 2^{2|\varphi| \log_{2} |\varphi|} < 2^{|\varphi|^{2}}$ such formulae. \qed
\end{proof}

The estimate above is rather crude, but 
$\| cl(\varphi) \|$ \emph{can} reach size  exponential in $|\varphi|$. Indeed, consider the 
formulae 
$\phi_{k} = \diams{1}(p_{1}\until q_{1} \land (p_{2}\until q_{2} \land (\ldots \land p_{k}\until q_{k})\ldots)$ for $k=1,2,\ldots$ and distinct 
$p_{1}, q_{1}, \ldots, p_{k},q_{k}, \ldots \in \prop$. Then $|\phi_{k}| = O(k)$, while the number of different $\gamma$-components of $\phi_{k}$ is $2^{k}$, hence $\| cl(\phi_{k}) \| > 2^{k}$.

%%%%%%%
\subsection{Full expansions of sets of \ATLp formulae}

As part of the tableau construction we will need a procedure that, for any given finite set of \ATLp state formulae $\Ga$, produces all ``full expansions'' (called in \cite{GorankoShkatov09ToCL} ``downward saturated extensions''; see Remark \ref{rem1}) defined below.  
\begin{definition}
\label{def:fullexp}
 Let $\Ga$, $\Delta$ be sets of \ATLp state formulae and $\Ga \subseteq \Delta \subseteq cl(\Ga)$. 
 \begin{enumerate}
\item  $\Delta$ is \emph{patently inconsistent} if it contains $\bot$ or a pair of formulae $\varphi$ and $\neg \varphi$.
\item $\Delta$ is a \emph{full expansion of $\Ga$} if it is not patently inconsistent and satisfies the following closure conditions: 
\begin{itemize}
 \item if $\varphi \land \psi \in \Delta$ then $\varphi \in \Delta$ and $\psi \in \Delta$;
  \item if $\varphi \lor \psi \in \Delta$ then $\varphi \in \Delta$ or $\psi \in \Delta$;
  \item if $\varphi \in \Delta$ is a $\gamma$-formula, then at least one $\gamma$-component 
of $\varphi$ is in $\Delta$ and exactly one of these $\gamma$-components  in $\Delta$, denoted $\gamma(\varphi, \Delta)$, is designated as \emph{the  $\gamma$-component in $\Delta$ linked to the $\gamma$-formula  $\varphi$}, as explained below. 
\end{itemize}
\end{enumerate}
\end{definition}

The family of all full expansions of $\Ga$ will be denoted by $FE(\Ga)$.  
%\vgadd{
It can be constructed by a simple iterative procedure that starts with $\{\Ga\}$ and repeatedly, until saturation, takes a set $X$ from the currently constructed family, selects a formula $\varphi \in X$ and applies the closure rule above corresponding to its type.  
Clearly, this procedure terminates on every finite input set of formulae $\Ga$ and produces a family of at most $2^{\| cl(\Ga) \|}$ sets. Furthermore, due to Lemma \ref{lem:gamma}, we have the following: 
%}

\begin{proposition}
\label{prop:FullExp}
For any finite set of \ATLp state formulae $\Ga$: 
\[\bigwedge \Ga \equiv \bigvee \left\{\bigwedge \De \mid \De \in FE(\Ga)\right\}.\]
\end{proposition}

\begin{proof}
Lemma \ref{lem:gamma} implies that every set extension step, described above, 
applied to a family of sets $\mathcal{F}$ preserves the formula $\bigvee \left\{\bigwedge \De \mid \De \in \mathcal{F}\right\}$ up to logical equivalence. At the beginning, that formula is $\bigwedge \Ga$. \qed
\end{proof}
 
 \begin{remark} 
 \label{rem1}
 Instead of full expansions, the tableau construction in \cite{GorankoShkatov09ToCL} uses 'minimal downward saturated extensions', where 
 'downward saturated extension' corresponds to 'full expansion'. 
The minimality condition means that if one full expansion is contained in another one, then it is omitted. 
This could be problematic, as sometimes non-minimal full expansions may be needed. 
For instance, if $ \Ga  = \{\diams{1}(p \until q), p \land \ddiams{1} (p \until q)\}$ then 
\\
$FE(\Ga ) = \big\{ \{\diams{1}(p \until q), p \land \ddiams{1}(p \until q), p, \ddiams{1}(p \until q)\}, \{q, \diams{1}(p \until q), p \land \ddiams{1}(p \until q), p, \ddiams{1}(p \until q)\}\big\}$. 

Although 
the second full expansion contains the first one, we might have to consider both alternatives in the tableau where $\Ga$ is only part of the label of a state, for the sake of satisfying an eventuality of the type $\varphi \until \psi$. However, we have no concrete example showing that such situation may occur indeed. 
\end{remark}

%%%%%%%%%%%%%%%%%%%%%%%%%%%% 
\section{Tableau-based decision procedure for \ATLp}
\label{sec:tab}
%%%%%%%%%%%%%%%%%%%%%%%%%%%%

The tableau procedure consists of three major phases: 
\emph{pretableau construction}, \emph{prestate elimination} and \emph{state elimination}. 
Given an input formula $\initf$, it essentially constructs a (non-deterministic) CGS which is state-labelled by the closure set of the input formula $cl(\initf)$, i.e.,  a directed graph $\tabb{}$ (called a \emph{tableau}) where each node is labelled by a subset of $cl(\initf)$ (see Def. \ref{def:CGS}(1)), and directed edges between nodes relating them to successor nodes. 

 The pretableau construction phase produces the so-called \emph{pretableau} $\pretab{}$ for the input formula $\initf$, with two kinds of nodes: \emph{states}
and \emph{prestates}. States are fully expanded sets,  
meant to represent  states of a  CGM, while prestates can be any finite sets of formulae from $cl(\initf)$ and only play a temporary role in the construction of $\pretab{}$.  
States and prestates are labelled uniquely, so they can be identified with their labels. 
The prestate elimination phase creates a smaller graph
$\tabb{0}$ out of $\pretab$, called the
\emph{initial tableau for $\initf$}, by eliminating all the prestates
from $\pretab$ and accordingly redirecting its edges.
Finally, the state elimination phase removes, step-by-step,  
all the states (if any) that cannot be satisfied in a CGM, because they lack necessary successors or because they contain unrealized eventualities. 
Eventually, 
the elimination procedure produces a (possibly empty) subgraph $\tabb{}$ of $\tabb{0}$, called the \emph{final tableau for $\initf$}. If some state $\Delta$ of
$\tabb{}$ contains $\initf$, the tableau procedure declares $\initf$
satisfiable and a partly defined CGM (called \emph{Hintikka game structure}) satisfying  $\initf$ can be extracted from it by another procedure described in Section \ref{subsec:completeness}; otherwise it declares $\initf$ unsatisfiable. 

\subsection{Pretableau construction phase}

The pretableau construction phase for an input formula $\initf$ starts with an initial prestate (with label) $\{\initf\}$ and consists of alternating application of two construction rules, until saturation: \textbf{(SR)}, expanding prestates into states, and \textbf{(Next)}, creating successor prestates from states. 
This phase closely resembles the corresponding one for the \ATL tableaux in \cite{GorankoShkatov09ToCL}, with the only essential difference being the $\gamma$-decomposition of $\gamma$-formulae used here by the \sr, which causes, as we will see, a possibly exponential blow-up of the size of the tableaux, and eventually of the entire worst-case time complexity, as compared to the \ATL tableaux. 
Another (minor) difference with respect to \cite{GorankoShkatov09ToCL} is in the formulation of both rules, because here we work with formulae in negation normal form.

\medskip
\textbf{Rule (SR)} 
Given a prestate $\Gamma$, do the following:
\begin{enumerate}
 \item For each full expansion $\Delta$ of $\Gamma$ add to the pretableau
 a state with label $\Delta$. 
 \item For each of the added states $\Delta$, if $\Delta$ does not contain any formulae of the form $\diaAn\varphi$ or $\crochetAn\varphi$, add the formula $\diamsn{\agents}\top$ to it;
 \item For each state $\Delta$ obtained at steps 1 and 2, link $\Gamma$ to $\Delta$ via a $ \Longrightarrow$ edge;
 \item If, however, the pretableau already contains a state $\Delta'$ with label $\Delta$, do not create another copy of it but  only link $\Gamma$ to $\Delta'$ via a $ \Longrightarrow$ edge.
\end{enumerate}

\begin{example}
\label{run-ex3}
 For the formula $\theta$ from Example \ref{run-ex}  the initial prestate is 
 \[\Gamma^{\theta}_0 = \{\diams{1}(p \until q \vee \Box q) \wedge \diams{2}(\Diamond p \wedge \Box \neg q) \}.\] 
 It has two full expansions: 
 
$\Delta_1 = \{\theta,\theta_1,\theta_2$, $p,\neg q,\ddiams{1}p\until q,\ddiams{2}(\Diamond p \wedge \Box\neg q)\}$, and 

$\Delta_2 = \{\theta,\theta_1,\theta_2,p,p\wedge \neg q,\neg q,\ddiams{1}p\until q,\ddiams{2}\Box\neg q\}$.

\bigskip
\begin{figure}[h]
\centering
\noindent\scalebox{0.65}{
\begin{tikzpicture}[node distance=3cm]

%Définition des styles
  \tikzstyle{pre_etat}=[text badly centered]
  \tikzstyle{etat}=[text badly centered]
  \tikzstyle{dfleche}=[->,>=latex, double]
  \tikzstyle{fleche}=[->,>=latex]

%Description des noeuds
  \node[pre_etat] (G0) {$\Gamma_0^\theta: \diams{1}(p \until q \vee \Box q) \wedge \diams{2}(\Diamond p \wedge \Box \neg q)$} ;
  \node (p1) [below of=G0, node distance=2cm] {};

  \node[etat, text width=5cm] (D1) [left of=p1] {$\Delta_1:\theta,\theta_1,\theta_2,p,\neg q$, $\ddiams{1}p\until q$, $\ddiams{2}(\Diamond p \wedge \Box\neg q)$} ;
  \node[etat, text width=5cm] (D2) [right of=p1] {$\Delta_2: \theta,\theta_1,\theta_2,p,p\wedge \neg q,\neg q$, $\ddiams{1}p\until q,\ddiams{2}\Box\neg q$} ;

% Description des liens entre pré-états et états
\path[dfleche]
    (G0) edge [double] (D1)
	 edge [double] (D2);
\end{tikzpicture}}
\caption{Application of the rule \textbf{(SR)} on the prestate $\Gamma^{\theta}_0 = \{\diams{1}(p \until q \vee \Box q) \wedge \diams{2}(\Diamond p \wedge \Box \neg q) \}$  }
\end{figure}
\bigskip

Likewise, for the formula $\vartheta$  the initial prestate is 
\[\Gamma^{\vartheta}_0 = \{\diams{1}(p \until q \vee \Box q) \wedge \crochet{2}(\Diamond p \wedge \Box \neg q)\}\] 
and it has 2 full expansions: 

$\Delta_1 = \{\vartheta,\vartheta_1,\vartheta_2,p,\neg q,\ddiams{1}p\until q,\dcrochet{2}(\Diamond p \wedge \Box\neg q)\}$, and 

$\Delta_2 = \{\vartheta,\vartheta_1,\vartheta_2,p,p\wedge \neg q,\neg q,\ddiams{1}p\until q$, $\dcrochet{2}\Box\neg q\}$.

\end{example}

\bigskip

In the following, by \emph{enforceable successor formula} we mean a formula of the form $\diaAn \psi$ and  by \emph{unavoidable successor formula} -- one of the form $\crochetn{A'}\psi$. 

\medskip
\noindent \textbf{Rule (Next)} Given a state $\Delta$, do the following, where $\vect$ 
is a shorthand for $\vec{\acta}_{\agents}$: 

\begin{enumerate}
 \item List all primitive successor formulae of $\Delta$ in such a way that all enforceable successor formulae precede all unavoidable ones where $A\neq\agents$; let the result be the list 
\[\mathbb{L} = \diamsn{A_0}\varphi_0,\dots,\diamsn{A_{m-1}}\varphi_{m-1},\crochetn{A'_0}\psi_0,\dots,\crochetn{A'_{l-1}}\psi_{l-1}\]

Let $r_{\Delta}=m+l$; denote by $D(\Delta)$ the set $\{0,\dots,r_{\Delta}-1\}^{|\agents|}$. 
Then, for every $\vect \in D(\Delta)$, denote $N(\vect):= \{i \mid \vect_i \geqslant m\}$, where $\vect_i$ is the $i$th component of the tuple $\vect$, and let $\co(\vect):= [\Sigma_{i \in N(\vect)}(\vect_i-m)]\mod l$.

  \item For each $\vect \in D(\Delta)$ 
   create a prestate:
    \begin{eqnarray*}
         \Gamma_{\vect} & = & \{\varphi_p \mid \diamsn{A_p}\varphi_p \in \Delta \text{ and } 
         \vect_a = p \text{ for all } a \in A_p\} \\
                         & \cup & \{\psi_q \mid \crochetn{A'_q}\psi_q \in \Delta, \co(\vect) = q, \text{ and } 
                          \agents - A'_q \subseteq N(\vect)\}
    \end{eqnarray*}
    If $\Gamma_{\vect}$ is empty, add $\top$ to it. Then 
   connect $\Delta$ to $\Gamma_{\vect}$ with $\stackrel{\vect}{\longrightarrow}$. 
   
If, however, $\Gamma_{\vect} = \Gamma$ for some prestate $\Gamma$ that has already been added to the pretableau, only connect $\Delta$ to $\Gamma$ with $\stackrel{\vect}{\longrightarrow}$.

\end{enumerate}

\begin{remark}
\label{rmk:rnext}
Rule \textbf{(Next)}  ensures that every  prestate $\Ga$ of 
$\De$, that is every element of the finite set of prestates that are targets of 
$\longrightarrow$ edges outgoing from $\Delta$, satisfies the following:
  \begin{itemize}
  \item if $\set{\diamsn{A_i} \varphi_i, \diamsn{A_j} \varphi_j} \subseteq
    \De$ and $\set{\varphi_i, \varphi_j} \subseteq \Ga$, then $A_i
    \cap A_j = \emptyset$;

    \item $\Ga$ contains at most one formula of the form
      $ \psi$ such that $\crochetn{A} \psi \in \De$, since the
      number $\co(\move{})$ is uniquely determined for every
      $\vect \in D(\Delta)$;

    \item if $\set{\diamsn{A_i} \varphi_i, \crochetn{A'} \psi}
      \subseteq \De$ and $\set{\varphi_i, \psi} \subseteq
      \Ga$, then $A_i \subseteq A'$.
  \end{itemize}
\end{remark}

Here is some intuition on the rule \textbf{(Next)} 
(see also \cite{GorankoShkatov09ToCL}). 
This rule must ensure that for each 
$\diamsn{A} \varphi$ from $\mathbb{L} $ there is a respective $A$-action at $\Delta$ that guarantees $\varphi$  in the label of every successor and that for every $\crochetn{A'} \psi$ from $\mathbb{L} $ there is a  $A'$-co-action at $\Delta$ that ensures $\psi$ in the label of the respective successors. 

Now, the actions at $\Delta$ are defined so that every agent's action represents a choice of that agent of a formula from 
$\mathbb{L} $ for the satisfaction of which the agent chooses to act. When all agents in some $A_{p}$ choose action $p$, then they act together for satisfying $\diamsn{A_{p}} \varphi_{p}$, so this is the required $A_{p}$-action.  
As for the co-actions, the idea is that for any fixed $\crochetn{A'_{q}} \psi_{q}$ in $\mathbb{L}$, all agents in the complement of $A'_{q}$ may choose to act in favour of some $\crochet{.}$-formula by simply selecting an action of the type $\sigma_{i} \geq m$. 
Then, for every $A'_{q}$-action the agents in $\agents - A'_{q}$ can synchronise their actions to ensure that the resulting action profile $\vect$ satisfies $\co(\vect) = q$, thereby ensuring  $\psi_{q}$ in the successor state. 
In fact, any agent who chooses to act \emph{co-strategically}, i.e., in favour of a $\crochet{.}$-formula,  can always synchronise her action with all other agents acting co-strategically to ensure that the resulting action profile $\vect$ satisfies $\co(\vect) = j$, for any value $j = 0, \ldots l-1$. So, every such agent $\agi$ is able, once all other agents have chosen their actions, to unilaterally enforce in the successor state any $\psi_{q}$ such that $\crochetn{A'_{q}} \psi_{q}$ in $\mathbb{L}$ and $\agi \notin A'_{q}$. 
   
The rules \textbf{(SR)} and \textbf{(Next)} are applied alternatively until saturation, which is bound to occur because every label is a subset of $cl(\initf)$. 
Then the construction phase is over. 
The  graph built in this phase is called \keyterm{pretableau} for the input formula $\initf$ and denoted by $\pretab$.
Given a pretableau, if $\Gamma$ is a prestate, we denote by $\mathbf{states}(\Gamma)$ the finite set of states that are targets of 
$\Longrightarrow$ edges outgoing from $\Gamma$ and 
if  $\Delta$ is a state we denote by $\mathbf{prestates}(\Delta)$ the finite set of prestates that are targets of 
$\longrightarrow$ edges outgoing from $\Delta$.

Before providing an example of how \rnext works, we give an example for the computation of the function $\co$. 

\begin{example}
Let the input formula, containing two agents, 1 and 2, be such that at some step of the pretableau construction, there is a state containing the next four primitive formulae: $\{ \diamsn{1}\varphi_1, \diamsn{1,2}\varphi_2, \crochetn{2}\varphi_3, \crochetn{1}\varphi_4 \}$. 
The computation of the functions $N$ and $\co$ and the successor prestate $ \Gamma_{\sigma}$ for each action profile gives:

\bigskip

\begin{center}
$
\arraycolsep=1.4pt\def\arraystretch{1.5}
\begin{array}{|c|c|c|c||c|c|c|c|}
\hline
	\sigma & N(\sigma)  & \co(\sigma) & \Gamma_{\sigma} & \sigma & N(\sigma) & \co(\sigma)  & \Gamma_{\sigma} \\\hline
	0,0 & \emptyset & 0 & \set{\varphi_1} & 2,0 & \set{1} & 0 & \set{\varphi_3} \\\hline
	0,1 & \emptyset & 0 &  \set{\varphi_1}  & 2,1 & \set{1} & 0 & \set{\top} \\\hline
	0,2 & \set{2} & 0 &  \set{\varphi_1}  & 2,2 & \set{1,2}& 0 & \set{\varphi_3} \\\hline
    0,3 & \set{2} & 1 & \set{\varphi_1, \varphi_4} & 2,3 & \set{1,2 }& 1 & \set{\varphi_4} \\\hline
	1,0 & \emptyset & 0 & \set{\top} & 3,0  & \set{1} & 1 & \set{\top} \\\hline
	1,1 & \emptyset & 0 & \set{\varphi_2} & 3,1 & \set{1} & 1 & \set{\top} \\\hline
	1,2 & \set{2} & 0 & \set{\top} & 3,2 & \set{1,2} & 1 & \set{\varphi_4} \\\hline
	1,3 & \set{2} & 1 & \set{\varphi_4} & 3,3 & \set{1,2} & 0 & \set{\varphi_3} \\\hline
\end{array}
$
\end{center}
\end{example}

\bigskip

\begin{example}
\label{run-ex4}
 (\textit{Continuation of Example \ref{run-ex3} for $\theta$}) 
 For $\Delta_1$, the list of successor formulae is 
 \begin{equation*}
 	\mathbb{L} = \ddiams{1}p\until q, \ddiams{2}(\Diamond p \wedge \Box \neg q)
 \end{equation*}
So $m = 2$, $l=0$ and $r_{\Delta_1} = 2$. 
    
    \medskip
    As there are no unavoidable successor formulae, we do not need to compute $N(\sigma)$ and $\co(\sigma)$. 
    Then, 
\begin{align*}
   \Gamma_{(0,0)} &= \{\diams{1}p\until q\} = \Gamma_1 & \Gamma_{(1,0)} &= \{\top\} = \Gamma_3\\
 \Gamma_{(0,1)} &= \{\diams{1}p\until q,\diams{2}(\Diamond p\wedge \Box \neg q)\} = \Gamma_2 & \Gamma_{(1,1)} &= \{\diams{2}(\Diamond p \wedge \Box \neg q)\} = \Gamma_4. 
\end{align*}

For $\Delta_2$, the list of successor formulae is 
\begin{equation*}
\mathbb{L} = \ddiams{1}p\until q,\ddiams{2}\Box\neg q
\end{equation*}
So $m = 2$, $l=0$ and $r_{\Delta_2} = 2$. 

\medskip
Here again, we do not compute $N(\sigma)$ and $\co(\sigma)$. Then
\begin{align*}
\Gamma_{(0,0)} &= \{\diams{1}p\until q\} = \Gamma_1 & \Gamma_{(1,0)} &= \{\top\} = \Gamma_3\\ 
\Gamma_{(0,1)} &= \{\diams{1}p\until q,\diams{2}\Box \neg q\} = \Gamma_5 & \Gamma_{(1,1)} &= \{\diams{2}\Box \neg q\} = \Gamma_6.
\end{align*}
Applying \sr to the so-obtained prestates, we have: 
\\
$\st{\Gamma_1} = \{\Delta_3:\{\diams{1}p\until q, p, \ddiams{1}p \until q\},\Delta_4:\{\diams{1}p\until q, q,\diamsn{1,2}\top\}\}$, 
\\
$\st{\Gamma_2} = \{\Delta_5:\{\diams{1}p\until q,\diams{2}(\Diamond p\wedge \Box \neg q),p,\neg q,\ddiams{1}p\until q,\ddiams{2}(\Diamond p \wedge \Box\neg q)\},
\Delta_6: \{\diams{1}p\until q,\diams{2}(\Diamond p\wedge \Box \neg q),p,p\wedge \neg q,\neg q,\ddiams{1}p\until q,\ddiams{2}\Box\neg q\}\}$; 
\\
$\st{\Gamma_3} = \{\Delta_7: \{\top,\diamsn{1,2}\top\}\}$; \\
$\st{\Gamma_4} = \{\Delta_8:\{\diams{2}(\Diamond p \wedge \Box\neg q),\neg q, \ddiams{2}(\Diamond p \wedge \Box\neg q)\}, 
\Delta_9: \{\diams{2}(\Diamond p \wedge \Box\neg q), p\wedge\neg q, \neg q, \ddiams{2}\Box\neg q\}\}$; 
\\
$\st{\Gamma_5} = \{\Delta_{10}:\{\diams{1}p\until q,\diams{2}\Box \neg q,p,\neg q$, $\ddiams{1}p\until q,\ddiams{2}\Box\neg q\}\}$; 
\\
$\st{\Gamma_6} = \{\Delta_{11}:\{\diams{2}\Box \neg q,\neg q,\ddiams{2}\Box\neg q\}\}$.
\end{example}

\medskip
The pretableau for $\theta$ is given in Figure \ref{fig:pretab:theta}.

\begin{figure}
\centering
\noindent\scalebox{0.65}{
\begin{tikzpicture}[node distance=3cm]

%Définition des styles
  \tikzstyle{pre_etat}=[text badly centered]
  \tikzstyle{etat}=[text badly centered]
  \tikzstyle{dfleche}=[->,>=latex, double]
  \tikzstyle{fleche}=[->,>=latex]

%Description des noeuds
  \node[pre_etat] (G0) {$\Gamma_0: \diams{1}(p \until q \vee \Box q) \wedge \diams{2}(\Diamond p \wedge \Box \neg q)$} ;
  \node (p1) [below of=G0, node distance=2cm] {};

  \node[etat, text width=5cm] (D1) [left of=p1] {$\Delta_1:\theta,\theta_1,\theta_2,p,\neg q$, $\ddiams{1}p\until q$, $\ddiams{2}(\Diamond p \wedge \Box\neg q)$} ;
  \node[etat, text width=5cm] (D2) [right of=p1] {$\Delta_2: \theta,\theta_1,\theta_2,p,p\wedge \neg q,\neg q$, $\ddiams{1}p\until q,\ddiams{2}\Box\neg q$} ;

  \node[pre_etat, text width=2.5cm,xshift=0cm] (G1) [below right of=D1] {$\Gamma_1: \diams{1}p\until q$} ;
  \node[pre_etat, text width=2.5cm,xshift=0cm] (G4) [left of=G1] {$\Gamma_4: \diams{2}(\Diamond p\wedge \Box \neg q)$} ;
  \node[pre_etat, text width=2.5cm,xshift=0cm] (G2) [left of=G4] {$\Gamma_2: \diams{1}p\until q$,\\$\diams{2}(\Diamond p\wedge \Box \neg q) $} ;
  \node[pre_etat, text width=2.5cm,xshift=0cm] (G3) [right of=G1] {$\Gamma_3: \top$} ;
  \node[pre_etat, text width=2.5cm,xshift=0cm] (G5) [right of=G3] {$\Gamma_5: \diams{1}p\until q,\diams{2}\Box \neg q$} ;
  \node[pre_etat, text width=2.5cm,xshift=0cm] (G6) [right of=G5] {$\Gamma_6: \diams{2}\Box \neg q$} ;

  \node[etat] (D7) [below of=G3,node distance=2cm,xshift=0.5cm] {$\Delta_7$} ;
  \node[etat] (D10) [below of=G5,node distance=2cm] {$\Delta_{10}$} ;
  \node[etat] (D11) [below of=G6,node distance=2cm] {$\Delta_{11}$} ;
  \node[etat] (D4) [left of=D7,node distance=1.7cm] {$\Delta_4$} ;
  \node[etat] (D3) [left of=D4,node distance=1.7cm] {$\Delta_3$} ;
  \node[etat] (D9) [left of=D3,node distance=1.7cm] {$\Delta_9$} ;
  \node[etat] (D8) [left of=D9,node distance=1.7cm] {$\Delta_8$} ;
  \node[etat] (D6) [left of=D8,node distance=1.7cm] {$\Delta_6$} ;
  \node[etat] (D5) [left of=D6,node distance=1.7cm] {$\Delta_5$} ;

  \node(p2) [below left of=D5,node distance=1.5cm] {}; \node(p3) [below right of=D5,node distance=1.5cm] {};
  \node(p4) [below left of=D6,node distance=1.5cm] {}; \node(p5) [below right of=D6,node distance=1.5cm] {};
  \node(p6) [below left of=D10,node distance=1.5cm] {}; \node(p7) [below right of=D10,node distance=1.5cm] {};

% Description des liens entre pré-états et états
\path[dfleche]
    (G0) edge [double] (D1)
	 edge [double] (D2)
    (G1) edge [double] (D3)
         edge [double] (D4)
    (G2) edge [double] (D5)
         edge [double] (D6)
    (G3) edge [double] (D7)
    (G4) edge [double] (D8)
         edge [double] (D9)
    (G5) edge [double] (D10)
    (G6) edge [double] (D11);

% Description des liens entre états et pré-états
\path[fleche]
    (D1) edge node [left] {$0,0$} (G1)
	 edge node [left] {$0,1$} (G2)
	 edge node [yshift=0.2cm] {$1,0$} (G3)
	 edge node [right] {$1,1$} (G4)
    (D2) edge node [yshift=0.2cm] {$0,0$} (G1)
	 edge node [right] {$0,1$} (G5)
	 edge node [right] {$1,0$} (G3)
	 edge node [right] {$1,1$} (G6)
     (D3) edge [bend left] node [left] {$0,0$} (G1)
     (D4) edge [bend right] node [left] {$0,0$} (G3)
    (D7) edge [bend right] node [right] {$0,0$} (G3)
    (D8) edge [bend left] node [left] {$0,0$} (G4)
    (D9) edge [bend right=70] node [below] {$0,0$} (G6)
    (D11) edge [bend right] node [right] {$0,0$} (G6);

\path[dashed]
   (D5) edge (p2)
   (p2) edge (p3)
   (p3) edge (D5)
   (D6) edge (p4)
   (p4) edge (p5)
   (p5) edge (D6)
   (D10) edge (p6)
   (p6) edge (p7)
   (p7) edge (D10);
\end{tikzpicture}
}

\begin{center}
\scalebox{0.65}{
\begin{tikzpicture}[node distance=3cm]

%Définition des styles
  \tikzstyle{pre_etat}=[text badly centered]
  \tikzstyle{etat}=[text badly centered]
  \tikzstyle{dfleche}=[->,>=latex, double]
  \tikzstyle{fleche}=[->,>=latex]

%Description des noeuds
  \node[etat, text width=5cm] (D5) [left of=D6,node distance=1.7cm] {$\Delta_5$} ;
  \node(p1) [below of=D5,node distance=2cm] {};
  \node[pre_etat] (G2) [left of=p1, node distance=0.8cm] {$\Gamma_2$} ;
  \node[pre_etat] (G1) [left of=G2, node distance=1.6cm] {$\Gamma_1$} ;
  \node[pre_etat] (G3) [right of=p1, node distance=0.8cm] {$\Gamma_3$} ;
  \node[pre_etat] (G4) [right of=G3, node distance=1.6cm] {$\Gamma_4$} ;

% Description des liens entre états et pré-états
\path[fleche]
    (D5) edge node [left] {$0,0$} (G1)
	 edge node [left] {$0,1$} (G2)
	 edge node [right] {$1,0$} (G3)
	 edge node [right] {$1,1$} (G4);
\end{tikzpicture}
}
\scalebox{0.65}{
\begin{tikzpicture}[node distance=3cm]

%Définition des styles
  \tikzstyle{pre_etat}=[text badly centered]
  \tikzstyle{etat}=[text badly centered]
  \tikzstyle{dfleche}=[->,>=latex, double]
  \tikzstyle{fleche}=[->,>=latex]

  \node[etat, text width=5cm] (D6) [left of=D6,node distance=1.7cm] {$\Delta_6/\Delta_{10}$} ;
  \node(p2) [below of=D6,node distance=2cm] {};
  \node[pre_etat] (G5) [left of=p2, node distance=0.8cm] {$\Gamma_5$} ;
  \node[pre_etat] (G1b) [left of=G5, node distance=1.6cm] {$\Gamma_1$} ;
  \node[pre_etat] (G3b) [right of=p2, node distance=0.8cm] {$\Gamma_3$} ;
  \node[pre_etat] (G6) [right of=G3b, node distance=1.6cm] {$\Gamma_6$} ;

% Description des liens entre états et pré-états
\path[fleche]
    (D6) edge node [left] {$0,0$} (G1b)
	 edge node [left] {$0,1$} (G5)
	 edge node [right] {$1,0$} (G3b)
	 edge node [right] {$1,1$} (G6);

\end{tikzpicture}
}
\end{center}
\caption{The pretableau for $\theta$} 
\label{fig:pretab:theta}
\end{figure}
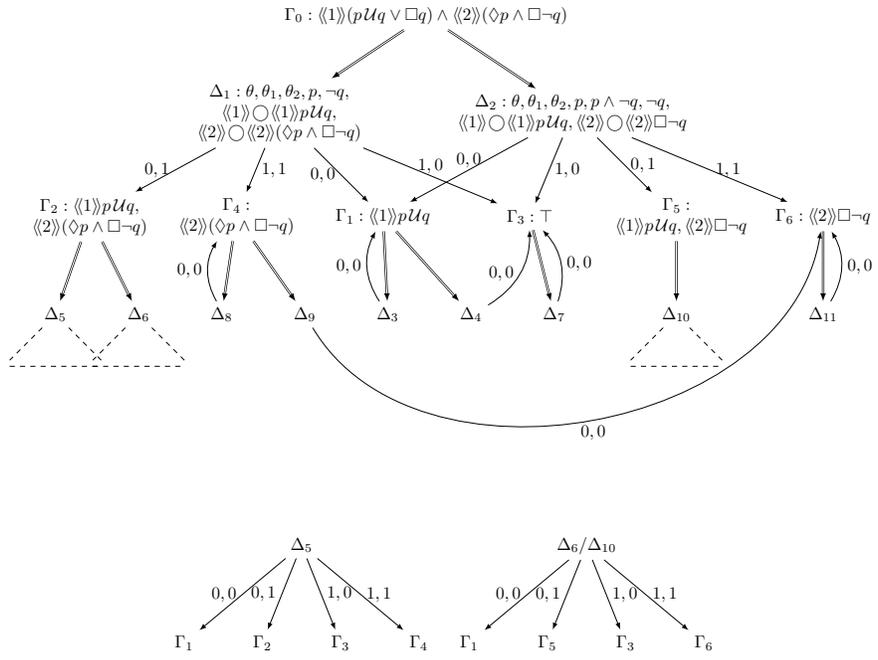

\begin{example}
\label{run-ex5}
 (\textit{Continuation of Example \ref{run-ex3} for $\vartheta$}) 
  For $\Delta_1$, the list of successor formulae is 
  \begin{equation*}
  \mathbb{L} = \ddiams{1}p\until q, \dcrochet{2}(\Diamond p \wedge \Box \neg q)
  \end{equation*}
  So $m = 1$, $l=1$ and $r_{\Delta_1} = 2$. Therefore, 
  \begin{align*}
  N(0,0) &= \emptyset &  N(1,0) &= \{1\} \\
  N(0,1) &= \{2\} & N(1,1) &= \{1,2\}
\end{align*}
and also $	\co(0,0) = \co(0,1) = \co(1,0) = \co(1,1)=0$.
 Then, 
\begin{equation*}
\Gamma_{(0,0)} =\Gamma_{(0,1)} = \{\diams{1}p\until q\} = \Gamma_1 \text{, and } \Gamma_{(1,0)} = \Gamma_{(1,1)} = \{\crochet{2}(\Diamond p\wedge \Box \neg q)\} = \Gamma_2.
\end{equation*}

\medskip
For $\Delta_2$, the list of successor formulae is 
\begin{equation*}
\mathbb{L} = \ddiams{1}p\until q,\dcrochet{2}\Box\neg q
\end{equation*}
So $m = 1$, $l=1$ and $r_{\Delta_2} = 2$. Here also 
\begin{align*}
N(0,0) &= \emptyset & N(1,0) &= \{1\}\\
N(0,1) &= \{2\} & N(1,1) &= \{1,2\}
\end{align*}
and  $\co(0,0) = \co(0,1)= \co(1,0) = \co(1,1)=0$. 
Then, 
\begin{equation*}
\Gamma_{(0,0)} =\Gamma_{(0,1)} = \{\diams{1}p\until q\} = \Gamma_1 \text{, and } \Gamma_{(1,0)} = \Gamma_{(1,1)} = \{\crochet{2}\Box \neg q\} = \Gamma_3.
\end{equation*}
In the same way, we obtain: 
\\
$\st{\Gamma_1} = \{\Delta_3:\{\diams{1}p\until q,p,\ddiams{1}p\until q\}, \Delta_4:\{\diams{1}p\until q, q, \diamsn{1,2}\top\}\}$;
\\
$\st{\Gamma_2}  = \{\Delta_5: \{\crochet{2}(\Diamond p \wedge \Box\neg q),\neg q,\dcrochet{2}(\Diamond p \wedge \Box\neg q)\},\Delta_6:\{\crochet{2}(\Diamond p \wedge \Box\neg q), p\wedge \neg q, p,\neg q, \dcrochet{2}\Box\neg q\}\}; \\
\st{\Gamma_3} = \{\Delta_7:\{\crochet{2}\Box\neg q,\neg q, \dcrochet{2}\Box\neg q\}\} $;\\
$\st{\Gamma_4} = \{\Delta_8:\{\top,\diamsn{1,2}\top\}\}$.

\end{example}

%%%%%%%%%%%%
The pretableau for $\vartheta$ is given in Figure \ref{fig:pretab:vartheta}.

\begin{figure}
\centering
\noindent\scalebox{0.65}{
\begin{tikzpicture}[node distance=3cm]

%Définition des styles
  \tikzstyle{pre_etat}=[text badly centered]
  \tikzstyle{etat}=[text badly centered]
  \tikzstyle{dfleche}=[->,>=latex, double]
  \tikzstyle{fleche}=[->,>=latex]

%Description des noeuds
  \node[pre_etat] (G0) {$\Gamma_0: \diams{1}(p \until q \vee \Box q) \wedge \diams{2}(\Diamond p \wedge \Box \neg q)$} ;
  \node (p1) [below of=G0, node distance=2cm] {};

  \node[etat, text width=5cm] (D1) [left of=p1] {$\Delta_1:\vartheta,\vartheta_1,\vartheta_2,p,\neg q,\ddiams{1}p\until q$, $\dcrochet{2}(\Diamond p \wedge \Box\neg q)$} ;
  \node[etat, text width=5cm] (D2) [right of=p1] {$\Delta_2: \vartheta,\vartheta_1,\vartheta_2,p,p\wedge \neg q,\neg q$, $\ddiams{1}p\until q$, $\dcrochet{2}\Box\neg q$} ;

  \node[pre_etat,xshift=0cm] (G1) [below right of=D1] {$\Gamma_1: \diams{1}p\until q$} ;
  \node[pre_etat,xshift=0cm] (G2) [left of=G1,node distance=4cm] {$\Gamma_2: \crochet{2}(\Diamond p\wedge \Box \neg q) $} ;
  \node[pre_etat,xshift=0cm] (G3) [right of=G1,node distance=8cm] {$\Gamma_3: \crochet{2}\Box \neg q$} ;

  \node[etat] (D3) [below left of=G1,node distance=2cm] {$\Delta_3$} ;
  \node[etat] (D4) [below right of=G1,node distance=2cm] {$\Delta_4$} ;
  \node[etat] (D7) [below of=G3,node distance=1.7cm, xshift=0.5cm] {$\Delta_7$} ;
  \node[etat] (D5) [below left of=G2,node distance=2cm] {$\Delta_5$} ;
  \node[etat] (D6) [below right of=G2,node distance=2cm] {$\Delta_6$} ;

  \node[pre_etat, xshift=0cm] (G4) [right of=D4,node distance=2cm] {$\Gamma_4: \top$} ;
  \node[etat] (D8) [right of=G4,node distance=2.5cm, yshift=0.7cm] {$\Delta_8: \top,\diamsn{1,2}\top$} ;

% Description des liens entre pré-états et états
\path[dfleche]
    (G0) edge [double] (D1)
	 edge [double] (D2)
    (G1) edge [double] (D3)
         edge [double] (D4)
    (G2) edge [double] (D5)
         edge [double] (D6)
    (G3) edge [double] (D7)
     (G4) edge [double] (D8);

% Description des liens entre états et pré-états
\path[fleche]
    (D1) edge node [right,text width=0.5cm,xshift=0.2cm] {$0,0$ $0,1$} (G1)
	 edge node [left,text width=0.5cm,xshift=-0.1cm] {$1,0$ $1,1$} (G2)
    (D2) edge node [left,text width=0.5cm,xshift=-0.2cm,yshift=0.2cm] {$0,0$ $0,1$} (G1)
	 edge node [right,text width=0.5cm,xshift=0.3cm,yshift=0.2cm] {$1,0$ $1,1$} (G3)
    (D3) edge [bend left] node [left] {$0,0$} (G1)
    (D4) edge node [above] {$0,0$} (G4)
    (D5) edge [bend left] node [left] {$0,0$} (G2)
    (D6) edge [bend right=55] node [above] {$0,0$} (G3)
    (D7) edge [bend right] node [right] {$0,0$} (G3);

\end{tikzpicture}
}
\caption{The pretableau for $\vartheta$} 
\label{fig:pretab:vartheta}
\end{figure}
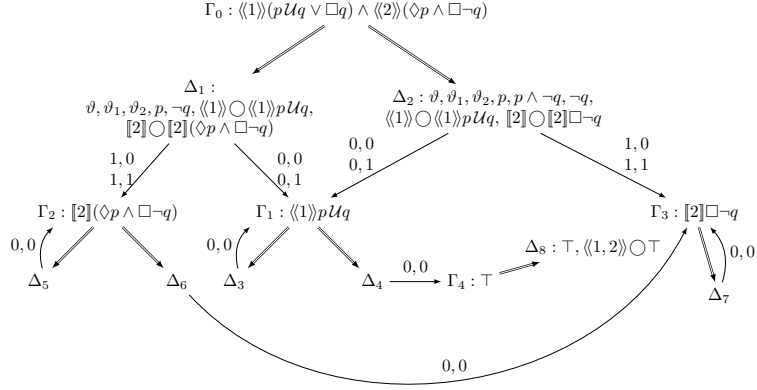
%%%%%%%%%%

\subsection{The prestate and state elimination phases. Eventualities}
First, we remove from $\pretab$ all the prestates and the $\Longrightarrow$ edges, as follows.  For every prestate $\Gamma$ in $\pretab$ put 
$\Delta\stackrel{\vect}{\longrightarrow} \Delta'$ 
for all states $\Delta$ in $\pretab$ with 
$\Delta\stackrel{\vect}{\longrightarrow} \Gamma$ 
and all $\Delta' \in \mathbf{states}(\Gamma)$; 
then,  remove $\Gamma$ from $\pretab$. 
The graph obtained after eliminating all prestates is called the \keyterm{initial tableau},  denoted by $\tabb{0}$. The initial tableau for the formula $\theta$ in our running example is given on Figure \ref{fig:init-tab:theta} and the initial tableau for $\vartheta$ is given on Figure \ref{fig:init-tab:vartheta}.

%% Initial Tableau Exemple 1 : theta
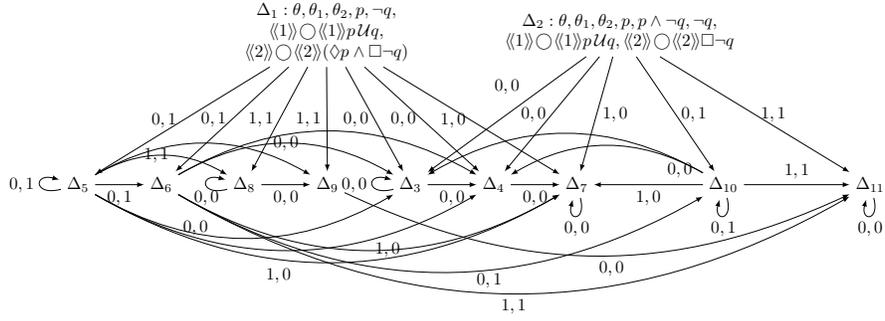
\begin{figure}
\centering
\noindent\scalebox{0.65}{
\begin{tikzpicture}[node distance=3cm]

%Définition des styles
  \tikzstyle{pre_etat}=[text badly centered]
  \tikzstyle{etat}=[text badly centered]
  \tikzstyle{dfleche}=[->,>=latex, double]
  \tikzstyle{fleche}=[->,>=latex]

%Description des noeuds
  \node (p1) {};

  \node[etat, text width=5cm] (D1) [left of=p1] {$\Delta_1:\theta,\theta_1,\theta_2,p,\neg q$, $\ddiams{1}p\until q,$ $\ddiams{2}(\Diamond p \wedge \Box\neg q)$} ;
  \node[etat, text width=5cm] (D2) [right of=p1] {$\Delta_2: \theta,\theta_1,\theta_2,p,p\wedge \neg q,\neg q$, $\ddiams{1}p\until q,\ddiams{2}\Box\neg q$} ;

  \node[pre_etat, text width=2.5cm,xshift=0cm] (G1) [below right of=D1] {} ;
  \node[pre_etat, text width=2.5cm,xshift=0cm] (G4) [left of=G1] {} ;
  \node[pre_etat, text width=2.5cm,xshift=0cm] (G2) [left of=G4] {} ;
  \node[pre_etat, text width=2.5cm,xshift=0cm] (G3) [right of=G1] {} ;
  \node[pre_etat, text width=2.5cm,xshift=0cm] (G5) [right of=G3] {} ;
  \node[pre_etat, text width=2.5cm,xshift=0cm] (G6) [right of=G5] {} ;

  \node[etat,yshift=1cm] (D7) [below of=G3,node distance=2cm] {$\Delta_7$} ;
  \node[etat,yshift=1cm] (D10) [below of=G5,node distance=2cm] {$\Delta_{10}$} ;
  \node[etat,yshift=1cm] (D11) [below of=G6,node distance=2cm] {$\Delta_{11}$} ;
  \node[etat] (D4) [left of=D7,node distance=1.7cm] {$\Delta_4$} ;
  \node[etat] (D3) [left of=D4,node distance=1.7cm] {$\Delta_3$} ;
  \node[etat] (D9) [left of=D3,node distance=1.7cm] {$\Delta_9$} ;
  \node[etat] (D8) [left of=D9,node distance=1.7cm] {$\Delta_8$} ;
  \node[etat] (D6) [left of=D8,node distance=1.7cm] {$\Delta_6$} ;
  \node[etat] (D5) [left of=D6,node distance=1.7cm] {$\Delta_5$} ;

  \node(p2) [below left of=D5,node distance=1.5cm] {}; \node(p3) [below right of=D5,node distance=1.5cm] {};
  \node(p4) [below left of=D6,node distance=1.5cm] {}; \node(p5) [below right of=D6,node distance=1.5cm] {};
  \node(p6) [below left of=D10,node distance=1.5cm] {}; \node(p7) [below right of=D10,node distance=1.5cm] {};

% Description des liens entre états et pré-états
\path[fleche]
    (D1) edge node [left] {$0,0$} (D3)
         edge node [left] {$0,0$} (D4)
	 edge node [left] {$0,1$} (D5)
	 edge node [left] {$0,1$} (D6)
	 edge node [left] {$1,0$} (D7)
	 edge node [left] {$1,1$} (D8)
	 edge node [left] {$1,1$} (D9)
    (D2) edge node [above,yshift=0.3cm] {$0,0$} (D3)
         edge node [left] {$0,0$} (D4)
	 edge node [right] {$0,1$} (D10)
	 edge node [right] {$1,0$} (D7)
	 edge node [right] {$1,1$} (D11)
     (D3) edge [loop left] node [left] {$0,0$} (D3)
          edge node [below] {$0,0$} (D4)
     (D5) edge [bend right] node [yshift=0.2cm,xshift=-1cm] {$0,0$} (D3)
	  edge [bend right] node [left] {} (D4)
	  edge [loop left] node [left] {$0,1$} (D5)
	  edge  node [below] {$0,1$} (D6)
	 edge [bend right] node [below,xshift=-1cm] {$1,0$} (D7)
	 edge [bend left] node [xshift=-0.1cm] {$1,1$} (D8)
	 edge [bend left] node [right] {} (D9)
     (D4) edge  node [below] {$0,0$} (D7)
   (D6) edge [bend left] node [] {$0,0$} (D3)
        edge [bend left] node {} (D4)
	 edge [bend right] node [below,xshift=1cm,yshift=0.1cm] {$0,1$} (D10)
	 edge [bend right] node [right] {$1,0$} (D7)
	 edge [bend right] node [below] {$1,1$} (D11)
     (D7) edge [loop below] node [below] {$0,0$} (D7)
    (D8) edge [loop left] node [below] {$0,0$} (D8)
	 edge node [below] {$0,0$} (D9)
    (D9) edge [bend right=25] node [below,xshift=0.3cm] {$0,0$} (D11)
   (D10) edge [bend right] node [left] {} (D3)
        edge [bend right] node [xshift=1.5cm,yshift=-0.5cm] {$0,0$} (D4)
	 edge [loop below] node [below] {$0,1$} (D10)
	 edge  node [below] {$1,0$} (D7)
	 edge node [above] {$1,1$} (D11)
    (D11) edge [loop below] node [below] {$0,0$} (D11);

\end{tikzpicture}
}
\caption{The initial tableau for $\theta$} 
\label{fig:init-tab:theta}
\end{figure}

%% Initial tableau : vartheta
\begin{figure}
\centering
\noindent\scalebox{0.65}{
\begin{tikzpicture}[node distance=3cm]

%Définition des styles
  \tikzstyle{pre_etat}=[text badly centered]
  \tikzstyle{etat}=[text badly centered]
  \tikzstyle{dfleche}=[->,>=latex, double]
  \tikzstyle{fleche}=[->,>=latex]

%Description des noeuds
   \node (p1) {};

  \node[etat, text width=5cm] (D1) [left of=p1] {$\Delta_1:\vartheta,\vartheta_1,\vartheta_2,p,\neg q,\ddiams{1}p\until q,$ $\dcrochet{2}(\Diamond p \wedge \Box\neg q)$} ;
  \node[etat, text width=5cm] (D2) [right of=p1] {$\Delta_2: \vartheta,\vartheta_1,\vartheta_2,p,p\wedge \neg q,\neg q,$ $\ddiams{1}p\until q$, $\dcrochet{2}\Box\neg q$} ;

  \node[pre_etat,xshift=0cm] (G1) [below right of=D1] {} ;
  \node[pre_etat,xshift=0cm] (G2) [left of=G1,node distance=4cm] {} ;
  \node[pre_etat,xshift=0cm] (G3) [right of=G1,node distance=8cm] {} ;

  \node[etat,yshift=1cm] (D3) [below left of=G1,node distance=2cm] {$\Delta_3$} ;
  \node[etat,yshift=1cm] (D4) [below right of=G1,node distance=2cm] {$\Delta_4$} ;
  \node[etat,yshift=1cm] (D7) [below of=G3,node distance=1.7cm] {$\Delta_7$} ;
  \node[etat,yshift=1cm] (D5) [below left of=G2,node distance=2cm] {$\Delta_5$} ;
  \node[etat,yshift=1cm,xshift=-0.5cm] (D6) [below right of=G2,node distance=2cm] {$\Delta_6$} ;

  \node[pre_etat, xshift=0cm] (G4) [right of=D4,node distance=2cm] {} ;
  \node[etat] (D8) [right of=G4,node distance=2.5cm] {$\Delta_8: \top,\diamsn{1,2}\top$} ;

% Description des liens entre états et pré-états
\path[fleche]
    (D1) edge node [right,text width=1cm] {$0,0$ $0,1$} (D3)
	  edge node [right,yshift=0.3cm,text width=1cm] {$0,0$ $0,1$} (D4)
	 edge node [left, text width=1cm] {$1,0$ $1,1$} (D5)
	 edge node [left,text width=0.5cm] {$1,0$ $1,1$} (D6)
    (D2) edge node [yshift=0.4cm, text width=0.6cm] {$0,0$ $0,1$}(D3)
	  edge node [right,text width=1cm,xshift=0.3cm] {$0,0$ $0,1$}(D4)
	 edge node [right,text width=0.5cm,yshift=0.3cm] {$1,0$ $1,1$} (D7)
    (D3) edge [loop left] node [below] {$0,0$} (D3)
         edge [bend right] node [below] {$0,0$} (D4)
    (D4) edge node [above] {$0,0$} (D8)
    (D5) edge [loop left] node [left] {$0,0$} (D5)
         edge [bend right] node [below] {$0,0$} (D6)
    (D6) edge [bend right=30] node [above] {$0,0$} (D7)
    (D7) edge [loop right] node [right] {$0,0$} (D7)
    (D8) edge [loop below] node [below] {$0,0$} (D8);
\end{tikzpicture}}
\caption{The initial and final tableau for $\vartheta$} 
\label{fig:init-tab:vartheta}
\end{figure}
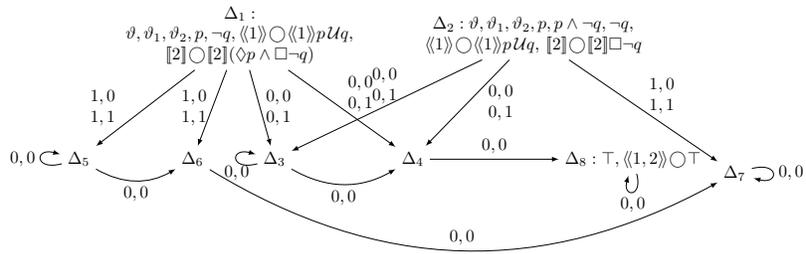
%%%%%%%%

The elimination phase starts with $\tabb{0}$ and goes through stages. At stage $n+1$ we remove exactly one state from the tableau $\tabb{n}$ obtained at the previous stage, by applying one of the elimination rules described below, thus obtaining the tableau $\tabb{n+1}$. The set of states of  $\tabb{m}$ is noted $\settab{m}$.

The first elimination rule \textbf{(ER1)}, defined below, is used to eliminate all states with missing successors for some action vectors determined by the \rnext. 
If, due to a previous state elimination, any state has an outgoing action vector for which the  corresponding successor state is missing, we delete the state.
The reason is clear: if $\Delta$ is to be satisfiable, then for each $\vect \in D(\Delta)$ there should exist a satisfiable $\Delta'$ that $\Delta$ reaches  via $\sigma$.
Formally, the rule is stated as follows, where $ D(\Delta)$ is defined in the \rnext: 

\smallskip
\textbf{\eun}: 
If, for some $\sigma \in D(\Delta)$, 
all states $\Delta'$ with $\Delta \stackrel{\sigma}{\longrightarrow} \Delta'$ have been eliminated at earlier stages, then obtain $\tabb{n+1}$ by eliminating $\Delta$  (together with its adjacent edges) from $\tabb{n}$.

\smallskip

The aim of the next elimination rule 
is to make sure that there are no \emph{unrealized eventualities}. In \ATL there are only two kinds of eventualities : $\diaA \varphi \until \psi$ and $\crochetA \varphi \until \psi$ . 
The situation is more complex in \ATLp. 
For instance, should the formula $\diaA(\Box\varphi \vee \psi_1\until \psi_2)$ be considered an eventuality?
Our solution for \ATLp is to consider all $\gamma$-formulae as \emph{potential eventualities}.  
In order to properly define the notion of \emph{realization of a potential eventuality} we first define 
a Boolean-valued function $Real$ that takes as arguments two elements: an \ATLp path-formula $\Phi$ and a set $\Theta$ of \ATLp state-formulae. This function allows us to check the realization of a potential eventuality of the form $\diaA\Phi$ and $\crochetA\Phi$ (where $\Phi$ is the first argument of $Real$) at a given state labelled by $\Theta$ (where $\Theta$ is the second argument of $Real$).
 \begin{itemize}
    \item $\fonctA{\Phi \wedge \Psi}{\Theta} = \fonctA{\Phi}{\Theta} \wedge \fonctA{\Psi}{\Theta}$
    \item $\fonctA{\Phi \vee \Psi}{\Theta} = \fonctA{\Phi}{\Theta} \vee \fonctA{\Psi}{\Theta}$
      \item $\fonctA{\varphi}{\Theta} = true$ iff $\varphi \in \Theta$
    \item $\fonctA{\rond\varphi}{\Theta} = false$ 
    \item $\fonctA{\Box\varphi}{\Theta} = true$ iff $\varphi \in \Theta$
    \item $\fonctA{\varphi\until\psi}{\Theta}= true$ iff $\psi \in \Theta$
  \end{itemize}

\begin{definition}[Descendant potential eventualities]
\label{def:descendant}
Let $\xi \in \Delta$ be a potential eventuality of the form $\diaA\Phi$ or $\crochetA\Phi$. 
Suppose the $\gamma$-component $\gamma(\xi, \Delta)$ in $\Delta$ linked to $\xi$ is, respectively, of the form 
 $\psi\wedge\ddiaA\Psi$ or $\psi\wedge\dcrochetA\Psi$. Then the \emph{successor potential eventuality} of $\xi$ w.r.t. $\gamma(\xi, \Delta)$ is the $\gamma$-formula $\diaA\Psi$ (resp. $\crochetA\Psi$) and it will be denoted by $\xi^{1}_{\Delta}$. 
 The notion of \emph{descendant potential eventuality of  $\xi$ of degree $d$}, for $d>1$,
is defined inductively as follows:

- any successor eventuality  of $\xi$  (w.r.t. some $\gamma$-component of $\xi$) is a descendant eventuality of $\xi$ of degree 1;  

- any successor eventuality of a descendant eventuality $\xi^n$ of $\xi$ of degree $n$ is a descendant eventuality of $\xi$ of degree $n+1$. 

We will also consider $\xi$
to be a descendant eventuality of itself of degree 0. 
\end{definition}

\begin{example}
Let $\xi = \diaA(\Box p \lor q \until r)$ be a potential eventuality such that $\xi \in \Theta$, where $\Theta$ is the labelling of a state $\Delta$. Let us see the different cases that can occur and the corresponding result of the function $Real$.
\begin{enumerate}
\item $p \in \Theta$ and $r \in \Theta$. In that case, $\fonctA{\Box p}{\Theta} = true$ and $\fonctA{q \until r}{\Theta} = true$, so $\fonctA{\Box p \lor q \until r}{\Theta} = true$. This is indeed correct since $q \until r$ is immediately realized. 
\item $p \not\in \Theta$ and $r \in \Theta$. This case is similar to the previous one even if $\fonctA{\Box p}{\Theta} = false$, indeed $\fonctA{\Box p \lor q \until r}{\Theta} = true$.
\item 
$p \not\in \Theta$ and $r \not\in \Theta$. Here $\fonctA{\Box p \lor q \until r}{\Theta} = false$ and the potential eventuality $\xi$ is not immediately realized. The \sr guarantees that $q \in \Theta$. This case means that the part $\Box p$ of $\xi$ is skipped and the part $q \until r$ will be continued. Therefore the next potential eventuality to be realized is $\diaA q \until r$. The immediate realization of this new potential eventuality will be checked again at next states.
\item $p \in \Theta$ and $r \not\in \Theta$. The potential eventuality $\xi$ is immediately realized since $\fonctA{\Box p}{\Theta} = true$, but two sub-cases can be distinguished to explain why this is correct:
\begin{enumerate}
\item $q \in \Theta$. Here both possibilities to do either $\Box p$ or $q \until r$ are kept. So the successor potential eventuality is $\diaA(\Box p \lor q \until r)$ and its immediate realization will be checked again at next states. 
\item $q \not\in \Theta$. This means that only the part $\Box p$ will be kept and the successor potential eventuality is $\diaA\Box p$. This case can be declared immediately realized since the construction rules of the tableau guarantees that $\diaA\Box p$ is correctly treated.
\end{enumerate}
\end{enumerate}
\end{example}

\begin{example}
\label{run-ex6}
(\textit{Continuation of Example \ref{run-ex4}}) \  
In $\Delta_5$ we have $\xi = \diams{1}(p\until q\vee\Box q)$ with $\fonctA{p\until q\vee\Box q}{\Delta_1} = \fonctA{p\until q}{\Delta_1} \vee \fonctA{\Box q}{\Delta_1} = false \vee false = false$, since $q \not\in \Delta_1$, and $\xi' = \diams{2}(\Diamond p \wedge \Box \neg q)$  with $\fonctA{\Diamond p \wedge \Box\neg q}{\Delta_1} = \fonctA{\Diamond p}{\Delta_1} \wedge \fonctA{\Box\neg q}{\Delta_1} = true \wedge true = true$ since $p,\neg q \in \Delta_1$.

The successor eventuality of $\xi=\diams{1}(p\until q \vee \Box q)$ w.r.t $\gamma(\xi,\Delta_1)$ is $\xi^1_{\Delta_1}=\diams{1}p\until q$ in $\Delta_3,\Delta_4,\Delta_5,\Delta_6$. For each $n>1$, the descendant eventuality of degree $n$ of $\xi$ w.r.t $\gamma(\xi,\Delta_1)$ is $\xi^n_{\Delta_1} = \xi^1_{\Delta_1}$ in $\Delta_3,\Delta_4\Delta_5,\Delta_6,\Delta_{10}$. 
The successor eventuality of $\xi'= \diams{2}(\Diamond p \wedge \Box\neg q)$ w.r.t $\gamma(\xi',\Delta_5)$ is $\xi'^1_{\Delta_5} = \diams{2}(\Diamond p \wedge \Box\neg q)$ in $\Delta_5$, $\Delta_6$, $\Delta_8$. For each $n>1$, the descendant eventualities of degree $n$ of $\xi'$ w.r.t $\gamma(\xi', \Delta_5)$ are 
$\xi'^n_{\Delta_5} = \xi'^1_{\Delta_5}$ in $\Delta_5$, $\Delta_6$, $\Delta_8$ and $\Delta_9$; and $\xi'^n_{\Delta_5} = \diams{2}\Box\neg q$ in $\Delta_{10}$ and $\Delta_{11}$.
\end{example}

Now, let $\mathbb{L} = \diamsn{A_0}\varphi_0,\dots,\diamsn{A_{m-1}}\varphi_{m-1},\crochetn{A'_0}\psi_0,\dots,\crochetn{A'_{l-1}}\psi_{l-1}$
be the list of all primitive successor formulae of $\Delta \in \settab{0}$, induced as part of an application of \textbf{(Next)}. We will use the following notation:\\
 $D(\Delta, \diamsn{A_p}\varphi_p) := \{\sigma \in D(\Delta) \mid \sigma_a = p \text{ for every } a \in A_p \}$\\
 $D(\Delta, \crochetn{A'_q}\psi_q) := \{\sigma \in D(\Delta) \mid \co(\sigma) = q \text{ and } \agents - A'_q \subseteq N(\sigma)\}$

 Next, we will  define recursively what it means for an eventuality $\xi$ to be realized at a state $\Delta$ of a tableau $\tabb{n}$, followed by our second elimination rule. 

\begin{definition}[Realization of potential eventualities]
\label{def:realEvent}
Let $\Delta \in \settab{n}$ and $\xi \in \Delta$ be a potential eventuality of the form $\diaA\Phi$ or $\crochetA\Phi$. Then: 
 \begin{enumerate}
  \item If $\fonctA{\Phi}{\Delta} = true$ then $\xi$ is realized at $\Delta$ in $\tabb{n}$.
  \item Else, let $\xi^1_{\Delta}$ be the successor potential eventuality of $\xi$ w.r.t. $\gamma(\xi, \Delta)$. If for every $\sigma \in D(\Delta, \diaAn\xi^1_\Delta)$ (resp. $\sigma \in D(\Delta, \crochetAn\xi^1_\Delta)$), 
  there exists $\Delta' \in \tabb{n}$ with
$\Delta \stackrel{\sigma}{\longrightarrow} \Delta'$ and $\xi^1_\Delta$ is realized at $\Delta'$ in $\tabb{n}$, then $\xi$ is realized at $\Delta$ in $\tabb{n}$. 
 \end{enumerate}
\end{definition}

\textbf{\edeux}: If $\Delta \in \settab{n}$ contains a potential eventuality that is not realized at $\Delta \in \tabb{n}$, 
then obtain $\tabb{n+1}$ by removing $\Delta$ (together with its adjacent edges) from $\settab{n}$.

\begin{example}
\label{run-ex7}
(\textit{Continuation of Example \ref{run-ex6}}) 
 The potential eventuality $\xi''=\diams{1}(p\until q)$ is not realized at $\Delta_5$, so by \edeux we remove the state $\Delta_5$ from $\taba{\theta}{0}$ and obtain the tableau $\taba{\theta} 
{1}$. The same applies to $\Delta_6$ for $\xi''$, so we also remove $\Delta_6$ from $\taba{\theta}{1}$ and obtain $\taba{\theta}{2}$ with \edeux. In $\taba{\theta}{2}$ there is no more move vector $(0,1)$ for the state $\Delta_1$, so by \eun we remove $\Delta_1$ from $\taba{\theta}{2}$ and obtain $\taba{\theta}{3}$. In the same way, $\Delta_{10}$ is removed by \edeux and $\Delta_2$ by \eun.
 
For the case of $\vartheta$, it is easy to see that no state gets eliminated, so the final tableau is the same as the initial one. 
\end{example}

The elimination phase is completed when no more applications of elimination rules are possible. Then we obtain the \keyterm{final tableau} for $\initf$, denoted by $\tabb{}$. It is declared \keyterm{open} if $\initf$ belongs to some state in it,  otherwise \keyterm{closed}. The procedure for deciding satisfiability of $\initf$ returns ``No'' if $\tabb{}$ is closed, ``Yes'' otherwise.

\begin{example}
\label{run-ex8}
(\textit{Continuation of Example \ref{run-ex7}})  
 At the end of the elimination phase, $\Delta_1$ and $\Delta_2$ are no longer in $\tabfin{\theta}$. Thus $\tabfin{\theta}$ is closed and the formula $\theta = \diams{1}(p \until q \vee \Box q) \wedge \diams{2}(\Diamond p \wedge \Box \neg q)$ is declared unsatisfiable. The final tableau for $\theta$ is given on Figure \ref{fig:fin-tab:theta}.

Respectively,  the final tableau for $\vartheta$ is open, hence $\vartheta$ is declared satisfiable. 
Indeed, a CGM can be extracted from the final tableau.  We will explain in Section \ref{subsec:completeness} how this can be done in a systematic way. 

\end{example}

\medskip
%% Final Tableau Exemple 1 : theta
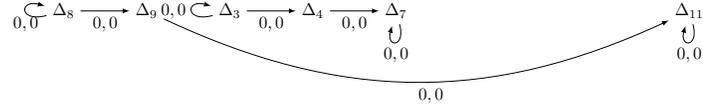
\begin{figure}
\begin{center}
\vspace{-1cm}
\noindent\scalebox{0.65}{ 
\begin{tikzpicture}[node distance=3cm]

%Définition des styles
  \tikzstyle{pre_etat}=[text badly centered]
  \tikzstyle{etat}=[text badly centered]
  \tikzstyle{dfleche}=[->,>=latex, double]
  \tikzstyle{fleche}=[->,>=latex]

%Description des noeuds
  \node(p1) {};

  \node[etat, text width=5cm] (D1) [left of=p1] {} ;
  \node[etat, text width=5cm] (D2) [right of=p1] {} ;

  \node[pre_etat, text width=2.5cm,xshift=0cm] (G1) [below right of=D1] {} ;
  \node[pre_etat, text width=2.5cm,xshift=0cm] (G4) [left of=G1] {} ;
  \node[pre_etat, text width=2.5cm,xshift=0cm] (G2) [left of=G4] {} ;
  \node[pre_etat, text width=2.5cm,xshift=0cm] (G3) [right of=G1] {} ;
  \node[pre_etat, text width=2.5cm,xshift=0cm] (G5) [right of=G3] {} ;
  \node[pre_etat, text width=2.5cm,xshift=0cm] (G6) [right of=G5] {} ;

  \node[etat,yshift=3cm] (D7) [below of=G3,node distance=2cm] {$\Delta_7$} ;
  \node[etat,yshift=3cm] (D10) [below of=G5,node distance=2cm] {} ;
  \node[etat,yshift=3cm] (D11) [below of=G6,node distance=2cm] {$\Delta_{11}$} ;
  \node[etat] (D4) [left of=D7,node distance=1.7cm] {$\Delta_4$} ;
  \node[etat] (D3) [left of=D4,node distance=1.7cm] {$\Delta_3$} ;
  \node[etat] (D9) [left of=D3,node distance=1.7cm] {$\Delta_9$} ;
  \node[etat] (D8) [left of=D9,node distance=1.7cm] {$\Delta_8$} ;
  \node[etat] (D6) [left of=D8,node distance=1.7cm] {} ;
  \node[etat] (D5) [left of=D6,node distance=1.7cm] {} ;

  \node(p2) [below left of=D5,node distance=1.5cm] {}; \node(p3) [below right of=D5,node distance=1.5cm] {};
  \node(p4) [below left of=D6,node distance=1.5cm] {}; \node(p5) [below right of=D6,node distance=1.5cm] {};
  \node(p6) [below left of=D10,node distance=1.5cm] {}; \node(p7) [below right of=D10,node distance=1.5cm] {};

% Description des liens entre états et pré-états
\path[fleche]
     (D3) edge [loop left] node [left] {$0,0$} (D3)
          edge node [below] {$0,0$} (D4)
     (D4) edge  node [below] {$0,0$} (D7)
     (D7) edge [loop below] node [below] {$0,0$} (D7)
    (D8) edge [loop left] node [below] {$0,0$} (D8)
	 edge node [below] {$0,0$} (D9)
    (D9) edge [bend right=25] node [below,xshift=0.3cm] {$0,0$} (D11)
    (D11) edge [loop below] node [below] {$0,0$} (D11);

\end{tikzpicture}
}
\caption{The final tableau for $\theta$} 
\vspace{-0.5cm}
\label{fig:fin-tab:theta}
\end{center}
\end{figure}

%%%%%%%%%%%%%%%%%%%%%%%%%%%%%%
\section{Termination and soundness}
\label{sec:soundness}
%%%%%%%%%%%%%%%%%%%%%%%%%%%%%%

The termination of the tableau procedure is straightforward, as there are only finitely many states and prestates that can be added in the construction phase.

\begin{theorem}\label{thm:sound}
The tableau method for  \ATLp is sound. 
\end{theorem}

Soundness of the tableau procedure with respect to unsatisfiability means that if a formula is satisfiable then its final tableau is open. 
To prove that, we essentially follow the same procedure as in the soundness proof for the tableau-based decision procedure for \atl in \cite{GorankoShkatov09ToCL}. 

The soundness proof consists of three main claims. 
First, we show that when a prestate $\Gamma$ is satisfiable then at least one of the states in $\mathbf{states}(\Gamma)$ is satisfiable. 
Then, we prove that when a state $\Delta$ is satisfiable then all the prestates in $\mathbf{prestates}(\Delta)$ are satisfiable. 
Finally, we show that no satisfiable states are eliminated in the elimination phase.
Below, we take the input formula of the tableau procedure to be $\initf$.

% *****
% 
% Construction Phase
% 
% *****

% *****
% *****  satisfiable prestate generates satisfiable states
% *****
The first step of the proof consists in showing that  \sr is sound:

\begin{lemma}
\label{lemme_sat1}
Let $\Gamma$ be a prestate of $\pretab$ and let $\m,s \models \Gamma$ for some CGM $\m$ and some $s \in \m$. 
Then, $\m,s\models \Delta$ holds for at least one $\Delta \in \mathbf{states}(\Gamma)$.
\end{lemma}

\begin{proof}
Straightforward from 
Proposition \ref{prop:FullExp}. \qed
\end{proof}

% *****
% *****  satisfiable state generates satisfiable prestates
% *****

The aim of the next two lemmas is to show that the \rnext creates only satisfiable prestates from satisfiable states.

We recall that we use  $\vectorsCoalatstate{A}{s}$ 
to denote the set of all A-actions that can be played by the 
coalition $A$ at state $s$, i.e. $\vectorsCoalatstate{A}{s} = \Pi_{\aga \in A} \vectorsCoalatstate{a}{s}$. 
We also use $\ensCoVect{A}{s}$ to denote the set of all $A$-co-actions available at state $s$ 
and $\covect{A}$ for an element of this set.
Let $\sigma_{A} \in \vectorsCoalatstate{A}{s}$. We say that an action profile $\moveAll$ extends an $A$-action $\sigma_A$, denoted by $\moveAll \sqsupseteq \sigma_A$, if $\moveAll(a) = \sigma_A(a)$ for every $a \in A$. We also use 
$\Out(s,\sigma_{A})$ to denote the set of all states $s'$ for which there exists an action profile
$\moveAll \in \vectorsatstate{s}$ that extends $\sigma_{A}$ and such that $\out(s, \moveAll) = s'$.
We define in a same way $\sqsupseteq$ and $\Out(s,\sigma^c_A)$ for an $A$-co-action $\sigma^c_A \in \ensCoVect{A}{s}$.

The following lemma states a semantic property, independent of the tableau construction.
\begin{lemma}
\label{mergeVectors}
 Let $\Theta = \{\diamsn{A_1}\varphi_1,\dots,\diamsn{A_m}\varphi_m,\crochetn{A'}\psi\}$
 be a set of formulae such that $A_i \cap A_j = \emptyset$ for every $1 \leqslant i,j \leqslant m$, $i\neq j$ and 
 $A_i \subseteq A'$ for every $1 \leqslant i \leqslant m$. 
 Let $\m,s \models \Theta$ for some GCM $\m$ and $s \in \m$. 
 Let $\sigma_{A_i} \in  \vectorsCoalatstate{A_i}{s}$ be an $A_i$-action witnessing the truth of $\diamsn{A_i}\varphi_i$ at $s$, for each $1 \leqslant i \leqslant m$, and let, finally, $\covect{A'} \in \ensCoVect{A'}{s}$ be an $A'$-co-action witnessing the truth of $\crochetn{A'}\psi$ at $s$. 
Then there exists $s' \in \Out(s,\sigma_{A_1})\cap\dots\cap \Out(s,\sigma_{A_m})\cap \Out(s,\covect{A'})$ such that $\m,s' \models \{\varphi_1,\dots,\varphi_m,\psi\}$.
\end{lemma}

\begin{proof}
Let $A = A_1 \cup \ldots \cup A_m$. Since $A_i \cap A_j = \emptyset$ for every $1 \leqslant i,j \leqslant m$, the actions $\sigma_{A_1},\ldots,\sigma_{A_m}$ can be combined to get
an $A$-action $\sigma_{A}$. This last can be arbitrarily extended to an $A'$-action $\sigma_{A'}$
because $A_i \subseteq A'$ for every $1 \leqslant i \leqslant m$. Finally, the so obtained $\sigma_{A'}$
can be completed by the $A'$-co-action $\covect{A'}$. The resulting action $\sigma_{\agents}$ leads
from $s$ to the desired $s'$.
\qed
\end{proof}

The next lemma states 
that satisfiability propagates from states to their successor prestates created via \rnext.
\begin{lemma}
\label{states_wrt_prestates}
 If $\Delta \in \pretab$ is a satisfiable state then all the prestates $\Gamma$ obtained by applying the \rnext\ are satisfiable. 
\end{lemma}

\begin{proof}
Follows by induction on the number of steps in the construction of the tableau, and from Lemma \ref{mergeVectors} and Remark \ref{rmk:rnext}. \qed 
\end{proof}

Thus, the \sr generates at least one satisfiable state from a satisfiable prestate and that the \rnext generates only satisfiable prestates from a satisfiable state.  
 Hence, we can conclude that the construction phase of the tableau procedure is sound.
 
\medskip 
% *****
% 
% Elimination Phase
% 
% *****
We now move to the elimination phase. 

\begin{lemma}
 Let $\Delta$ be a state in $\tabb{n}$. If $\Delta$ is satisfiable then \eun cannot eliminate $\Delta$ from $\tabb{n}$.
\end{lemma}

\begin{proof}
By Lemma \ref{mergeVectors}  a satisfiable state $\Delta$ generates only satisfiable successor prestates, and, by Lemma \ref{states_wrt_prestates}, 
each of these prestates in turn generates at least one satisfiable state. Therefore, if $\Delta$ is satisfiable 
then for every action profile $\moveAll \in \vectorsatstate{s}$ there is a state $\Delta'$ such that $\Delta \stackrel{\moveAll}{\longrightarrow} \Delta'$. 
Therefore the \eun cannot eliminate a satisfiable state. \qed
\end{proof}

It remains to be proved that a satisfiable state cannot be eliminated by \edeux, either.  
We recall that \edeux eliminates each state containing an eventuality that is not realized at that state. 
So we need to prove that if a state $\Delta$ is satisfiable, then every eventuality $\xi \in \Delta$ is realized at $\Delta$ throughout the elimination phase.

Note that the structure underlying a tableau can be seen as a non-deterministic CGS, where edges outgoing from a tableau state can lead to different successors even if they are labelled by the same action vector. The following two definitions will be used to extract  deterministic transitions from non-deterministic ones.

\begin{definition}[Outcome set of $\sigma_A$ at $s$]
Let $\cgs$ be a non-deterministic concurrent game structure,  let $s$ be a state
and let $\sigma_A \in \vectorsCoalatstate{A}{\sigma}$.
An \keyterm{outcome set of $\sigma_A$ at $s$}  
is a set of states $X$ such that for every $\moveAll \sqsupseteq \sigma_A$ 
there exists exactly one $s' \in X$ such that $s \stackrel{\moveAll}{\longrightarrow} s'$.  
\end{definition}

\begin{definition}[Outcome set of $\covect{A}$ at $s$]
Let $\cgs$ be a non-deterministic concurrent game structure,  let $s$ be a state
and let  $\covect{A} \in \Acovectsatstate{s}$.
 An \keyterm{outcome set of $\covect{A}$ at $s$} is a  set of states $X$ such that 
 for every $\sigma_A  \in \vectorsCoalatstate{A}{s}$
 there exists exactly one $s' \in X$ such that 
 $s \stackrel{\sigma^{c}_{A}(\sigma_A)}{\longrightarrow} s'$. 
 
\end{definition}

In particular, both definitions above can be applied to a tableau, where the states $s$  and $s'$ are taken to be tableau states $\Delta$ and $\Delta'$. 

\bigskip

Some notation. Consider a concurrent game structure  $\cgs$ which is  state-labelled by a set $\Theta$ of state formulae of \ATLp and suppose
that the elements of $\Theta$ are listed by any enumeration $E$ where enforceable next-time formulae appear before unavoidable next-time formulae, in particular, in the list $\mathbb{L}$ given in the definition of Rule \textbf{(Next)}. Then: 
\begin{enumerate}
  \item Whenever we write $\diamsn{A_p}\varphi_p \in \Theta$, we mean that $\diamsn{A_p}\varphi_p$ is the $p$-th enforceable next-time formula according to $E$.
  In particular, when $\cgs$ is a tableau, $E$ is usually assumed to be the listing of the successor formulae of $\Theta$ induced by the application of the \rnext to $\Theta$. 
  
  We use the notation $\crochetn{A'_q}\psi_q \in \Theta $ likewise. 
  
  \item Given $\diamsn{A_p}\varphi \in \Theta$, we denote by $\sigma_{A_p}[\diamsn{A_p}\varphi]$
a (somehow selected) $A$-action enforcing $\varphi$ in any associated successor state. 

In particular, when  $\cgs$  is a tableau, we denote by $\sigma_{A_p}[\diamsn{A_p}\varphi_p]$
the unique $A_p$-action $\sigma_{A_p} \in\vectorsCoalatstate{A_p}{\Theta}$  in the tableau such that $\sigma_{A_p}(a) = p$ for every $a \in A_p$.

   \item Likewise, given a formula $\crochetn{A'}\psi \in \Theta$, where $A'_q \not =\agents$, we denote by 
  $\covect{A'_q}[\crochetn{A'_q}\psi]$ a (somehow selected) $A'_q$-co-action $\sigma^c_{A'_q}$ enforcing
$\psi$ in any associated successor state.  

In particular, when  $\cgs$  is a tableau, we denote by $\covect{A'_q}[\crochetn{A'_q}\psi_q]$ 
the unique $A'_q$-co-action $\sigma^c_{A'_q} \in \ensCoVect{A'_q}{\Delta}$ in the tableau satisfying the following condition (with notation referring to the definition of \rnext): 
 $\co(\sigma^{c}_{A'_q}(\sigma_{A'_q})) = q$ and $\agents - A'_q\subseteq N(\sigma^{c}_{A'_q}(\sigma_{A'_q}))$ 
  for every $\sigma_{A'_q} \in \vectorsCoalatstate{A'_q}{\Delta}$.
\end{enumerate}

In order to prove that the rule \textbf{(ER2)} does not eliminate any satisfiable states, we need to show that if a tableau $\tabb n$ contains a state $\Delta$ that is satisfiable and contains an eventuality $\xi$ , then $\xi$ is realized at $\Delta$.
Thus we prove that $\tabb n$ ``contains'' a structure (more precisely, a tree) that ``witnesses'' the realization of $\xi$ at $\Delta$ in $\tabb n$. This tree will emulate a tree of runs effected by a strategy or a co-strategy that ``realizes'' an eventuality in a model. This simulation is done step-by-step, and each step, i.e. $A$-action (in the case of $\diaA\Phi$) or $A$-co-action (in the case of $\crochetA\Phi$) corresponds to a tableau action or co-action associated with a respective eventuality. The fact that this step-by-step simulation can be done is proved in the next two lemmas (together with their corollaries).

\begin{lemma}
\label{lem:outcomeSet}
 Let $\diamsn{A_p}\varphi_p \in \Delta \in \settab{n}$ and let 
 $\m,s \models \Delta$ for some CGM $\m$ and state $s \in \m$. 
 Let, furthermore, $\sigma_{A_p} \in \vectorsCoalatstate{A_p}{s}$ be an $A_p$-action witnessing the truth of 
 $\diamsn{A_p}\varphi_p$ at $s$. Then, there exists in $\tabb{n}$ an outcome set $X$ of 
 $\sigma_{A_p}[\diamsn{A_p}\varphi_p]$ such that for each $\Delta' \in X$ there exists 
 $s' \in \Out(s,\sigma_{A_p})$ such that $\m,s' \models \Delta'$.
\end{lemma}

\begin{proof}
We consider the following set of prestates (from the pretableau construction): 
\[Y = \{ \Gamma \in \mathbf{prestates}(\Delta) \mid \Delta \stackrel{\moveAll}{\longrightarrow} \Gamma \text{ for some }  \moveAll \sqsupseteq \sigma_{A_p}[\diamsn{A_p}\varphi_p]\}\]
For every $\Gamma \in Y$, it follows immediately from the \rnext that $\Gamma$ (which must contain $\varphi_p$) is
 either of the form 
 
 $\{\varphi_1,\dots,\varphi_m,\psi\}$, where
 $\{\diamsn{A_1}\varphi_1,\dots,\diamsn{A_m}\varphi_m,\crochetn{A'}\psi \} \subseteq \Delta$, 
 
 or of the form 
 
 $\{\varphi_1,\dots,\varphi_m\}$ where $\{\diamsn{A_1}\varphi_1,\dots,\diamsn{A_m}\varphi_m \} \subseteq \Delta$. 
 
 We can reduce the latter case to the former by adding the valid formula $\crochetn{\agents}\top$ (equivalent to $\diamsn{\emptyset}\top$).

Since $\m,s\models \Delta$,  
by Lemma \ref{mergeVectors}, 
there exists $s' \in \Out(s,\sigma_{A_p})$ with $\m,s' \models \Gamma$. 
Then $\Gamma$ can be extended to a fully expanded set
$\Delta'$ containing at least one successor formula 
($\diamsn{\agents}\top$, if nothing else) such that $\m,s' \models \Delta'$. 
This is done by choosing, for every $\beta$- or $\gamma$-formula to be processed in the procedure computing the family of full expansions, a disjunct, resp. a  $\gamma$-component, that is actually true in $\m$ at $s'$ (if there are several such options, the choice is arbitrary) and adding it to the current set. \qed
\end{proof}

\begin{corol}
\label{crl1}
 Let $\diamsn{A_p}\varphi_p \in \Delta$ for $\Delta \in \settab{n}$ and let $\m,s \models \Delta$ for some CGM $\m$ and state $s \in \m$. Let, furthermore, $\sigma_{A_p} \in  \vectorsCoalatstate{A_p}{s}$
 be an $A_p$-action witnessing the truth of $\diamsn{A_p}\varphi_p$ at $s$ and let 
 $\chi \in cl(\initf)$ be a $\beta$-formula (resp. a $\gamma$-formula) and $\psi$ be one of its $\beta$-components (resp. $\gamma$-components).   
 Then there exists in $\tabb{n}$ an outcome set $X_{\psi}$ of $\sigma_{A_p}[\diamsn{A_p}\varphi_p]$ 
 such that for every $\Delta' \in X_{\psi}$ there exists $s' \in \Out(s,\sigma_{A_p})$ such that 
 $\m,s' \models \Delta'$, and moreover, if $\m,s' \models\psi$, then $\psi \in \Delta'$.
\end{corol}
\begin{proof}
Construct $X_{\psi}$ just like $X$ was constructed in the proof of the preceding lemma, with a single modification: when dealing with the formula $\chi$, instead of choosing arbitrarily between the different options for $\psi$, choose $\psi$ which is true at $s'$. \qed
\end{proof}
\bigskip

Likewise, we obtain the following for unavoidable formulae:

\begin{lemma}
 Let $\crochetn{A'_q}\psi_q \in \Delta \in \settab{n}$ and let $\m,s \models \Delta$ for some CGM $\m$ and state $s \in \m$. Let, furthermore, $\covect{A'_q} \in \ensCoVect{A'_q}{s}$ be an $A'_q$-co-action witnessing the truth of $\crochetn{A'_q}\psi_q$ at $s$. Then, there exists in $\tabb{n}$ an outcome set $X$ of $\covect{A'_q}[\ \crochetn{A'_q}\psi_q]$ such that for each $\Delta' \in X$ there exists $s' \in \Out(s,\covect{A'_q})$ such that $\m,s' \models \Delta'$.
\end{lemma}

The proof is analogous to the proof of Lemma \ref{lem:outcomeSet}.

\begin{corol}
\label{crl2}
 Let $\crochetn{A'_q}\psi_q \in \Delta \in \settab{n}$ and let $\m,s \models \Delta$ for some CGM $\m$ and state $s \in \m$. Let, furthermore, $\sigma_{A'_q}^c \in \ensCoVect{A'_q}{s}$ be an $A'_q$-co-action witnessing the truth of $\crochetn{A'_q}\psi_q$ at $s$ and let $\chi \in cl(\initf)$ be a $\beta$-formula (resp. a $\gamma$-formula), whose associated $\beta_i$-component ($i \in \{1,2\})$
(resp. $i$-th $\gamma$-component  ($i \geqslant 1$)) is $\chi_i$. Then there exists in $\tabb{n}$ an outcome set $X_{\chi_i}$ of $\covect{A'_q}[\ \crochetn{A'_q}\psi_q]$ such that for every $\Delta' \in X_{\chi_i}$ there exists $s' \in \Out(s,\covect{A'_q})$ such that $\m,s' \models \Delta'$, and moreover, if $\m,s' \models\chi_i$, then $\chi_i \in \Delta'$.
\end{corol}

In what follows we make use of the notion of \keyterm{tree}. In our context, we use such a term as
a synonym of ``directed, connected, and acyclic graph, every node of which, except  the root, has exactly one incoming edge". We denote a tree as a pair $(R,\rightarrow)$, where $R$ is the set of nodes and
$\rightarrow$ is the parent-child relation (the edges).

The first kind of tree that we define is the so-called  \keyterm{realization witness tree}.  
Intuitively, such tree witnesses the satisfaction of a given potential eventuality $\xi$
at a state and simulates a tree of runs effected in a model by (co-)strategies. Our definition is
more general than the one in \cite{GorankoShkatov09ToCL}, 
as we want this notion to be applicable in a broader context, including tableaux, concurrent game models and concurrent game Hintikka structures (to be defined later). 

The two definitions below implicitly use the notion of descendant potential eventuality of degree $d$ and its associate notation (see Definition \ref{def:descendant}).
That notion was defined in the context of tableaux, however it is applicable to any CGS which is state-labelled by a set of state formulae. 
We recall that, given a potential eventuality $\xi=\diaA\Phi$ ($\crochetA\Phi$), by convention $\xi$ itself
is taken to be its (unique) descendant potential eventuality of degree 0 and that 
if $\xi^i$ is a descendant eventuality of degree $i$ of $\xi$ then a $\gamma$-component of $\xi^i$ will have
the form $\psi \wedge \ddiaA \Phi^{i+1}$ (respectively,  $\psi \wedge \dcrochetA\Phi^{i+1}$)
and $\diaA \Phi^{i+1}$ (respectively, $\crochetA \Phi^{i+1}$) will be a descendant potential eventuality of $\xi$ having degree $d = i+1$.

A piece of terminology that will be used often further: given sets $X,Y$ and a mapping $c: X \to Y$, we
sometimes say that \keyterm{the set  $X$ is $Y$-coloured by $c$} and that for any $x\in X$, the value $c(x)$ is \keyterm{the $Y$-colour of $x$ under the colouring $c$}. 

\begin{definition}[Realization witness tree for enforceable potential eventualities]
 \label{def:realTree1} 
Let $\cgs$ be any (non-deterministic) CGS with a state space $\stat$ which is state-labelled by some set of \ATLp formulae $\Gamma$, with a labelling function $c$. 
 Let $s\in \stat$ and let  $\xi \in s$ be a potential eventuality of the form $\diaA\Phi$.  
 A \keyterm{realization witness tree for $\xi$ at $s$} is a finite tree $\mr=(R,\rightarrow)$, where the set of nodes $R$ is $\stat$-coloured so that: 
 \begin{enumerate}
 \item the root of $\mr$ is coloured with $s$ and is of depth $0$;
 \item if an interior node $w$ of depth $i$ of $\mr$ is coloured with $s'$ where $c(s')=\Theta$, then there exists a successor $\xi^{i+1}$ of $\xi^i$ such that $\diaAn\xi^{i+1}\in \Theta$;  
 \item for every interior node $w \in \mr$ of depth $i$ coloured with $s'$ where $c(s') = \Theta$, the children of $w$ are coloured bijectively with 
 vertices from an outcome set of $\sigma_A[\diaAn\xi^i]$ at $s'$;
 \item if a leaf of depth $i$ of $\mr$ is coloured with $s'$ where $c(s') = \Theta$, then $\xi^i = \diaA\Phi \in \Theta$ is such that $\fonctA{\Phi}{\Theta} = true$.
\end{enumerate}
\end{definition}

\begin{definition}[Realization witness tree for unavoidable potential eventualities]

Let $\cgs$ be any (non-deterministic) CGS with a state space $\stat$ which is state-labelled by some set of \ATLp formulae $\Gamma$, with a labelling function $c$. 
 Let $s\in \stat$ and let  $\xi \in s$ be a potential eventuality of the form $\crochetA\Phi$.  
 A \keyterm{realization witness tree for $\xi$ at $s$} is a finite tree $\mr=(R,\rightarrow)$, where the set of nodes $R$ is $\stat$-coloured so that: 

\begin{enumerate}
 \item the root of $\mr$ is coloured with $s$ and is of depth $0$;
 \item if an interior node $w$ of depth $i$ of $\mr$ is coloured with $s'$ where $c(s')=\Theta$, then there exists a successor $\xi^{i+1}$ of $\xi^i $ such that $\crochetAn\xi^{i+1}\in \Theta$;
 \item for every interior node $w \in \mr$ of depth $i$ coloured with $s'$ where $c(s') = \Theta$, the children of $w$ are coloured bijectively with 
 vertices from an outcome set of $\covect{A}[\crochetAn\xi^{i+1}]$ at $s'$;
 \item if a leaf of depth $i$ of $\mr$ is coloured with $s'$ where $c(s') = \Theta$, then $\xi^i = \crochetA\Phi \in \Theta$ is such that $\fonctA{\Phi}{\Delta'} = true$.
\end{enumerate}

\end{definition}

We are going to apply the definitions above for the case when the CGS $\cgs$ is a tableau $\tabb{n}$, with states being (identified with) the sets of formulae in their labels.

\begin{lemma}
\label{lemma:realization}
 Let $\mr=(R,\rightarrow)$ be a realization witness tree for a potential eventuality $\xi$ at $\Delta \in \settab{n}$. 
 For every $\Delta'$ of depth $i$, colouring a node of $R$, $\xi^i$ is realized at $\Delta'$ in $\tabb{n}$.
 In particular for $i=0$, thus $\xi$ is realized at $\Delta$ in $\tabb{n}$. 
\end{lemma}

\begin{proof}
We prove this lemma by induction on the length of the longest path from a node coloured by $\Delta$ to a leaf of $\mr$.

\medskip

Base case: The length of the longest path from a node $w$ coloured by $\Delta$ to a leaf of $\mr$ is $0$. 
Then $w$ is a leaf and $\fonctA{\xi}{\Delta} = true$. 
Thus, by item 1 of Definition \ref{def:realEvent}, 
$\xi$ is realized at $\Delta \in \settab{n}$.

\medskip
Induction step: The length of the longest path from a node $w$ coloured by $\Delta$ to a leaf of $\mr$ is $l>0$. 
Then $w$ is an interior node of depth $i$, so $\diaA\xi^i \in \Delta$ (resp. $\crochetA\xi^i \in \Delta)$) and 
there exists a action (resp. a co-action) such that for all children $w'$ of $w$, where each $w'$ is coloured by 
$\Delta'$, $\xi^{i+1}_{\Delta} \in \Delta'$. Let $\mr'$ be a sub tree of $\mr$ whose root is $w'$. 
The length of the longest path from a node $w'$ coloured by $\Delta'$ to a leaf of $\mr'$ is at most $l-1$. 
Thus, by induction hypothesis, $\xi^{i+1}_{\Delta}$ is realized at $\Delta' \in \tabb{n}$ and 
$\xi^{j}$ is realized at $\Delta''$ in $\tabb{n}$.
Therefore $\xi^{i+1}_{\Delta}$ is realized at $\Delta'$ in $\tabb{n}$ and $w$ respects item 2 of Definition
\ref{def:realEvent}. 
We conclude that $\xi^i$ is realized at $\Delta \in \settab{n}$.  \qed
 
\end{proof}

We now prove the existence of a realization witness tree for any satisfiable state of a tableau containing a potential eventuality.

\begin{lemma}
\label{eventualities_realized}
 Let $\xi \in \Delta$ be a potential eventuality and  $\Delta \in \settab{n}$ 
 be satisfiable. Then there exists a realization witness tree $\mr = (R,\rightarrow)$ for $\xi$ at $\Delta \in \settab{n}$. 
 Moreover, every $\Delta'$, colouring a node of $R$, is satisfiable. 
\end{lemma}

\begin{proof}
We will only give the proof for potential eventualities of the type $\diaA\Phi$. The case of potential eventualities of type $\crochetA\Phi$ is similar.

When dealing with realization of potential eventualities, we have two cases:
\begin{enumerate}
	\item $\fonctA{\Phi}{\Delta} = true$. This case is straightforward, the realization witness tree consists of only the root, coloured with $\Delta$.
	\item $\fonctA{\Phi}{\Delta} = false$. This case means that there is a successor potential eventuality $\xi^1_\Delta$ such that $\diaAn\xi^1_\Delta \in \Delta$.
	
	As $\Delta$ is satisfiable, there exists a CGM $\m$ and a state $s \in \m$ such that $\m,s\models\Delta$, and in particular, $\m,s\models\diaAn\xi^1_\Delta$. Thus, there exists an $A$-action $\sigma_A \in \Acovectsatstate{s}$ such that $\m,s'\models\xi^1_\Delta$ for all
	$s' \in \Out(s,\sigma_A)$, that is an $A$-action witnessing the truth of $\diaAn\xi^1_\Delta$ at $s$.
	
	We know that $\Delta$ is satisfiable and that $\diaAn\xi^1_\Delta$ is an enforceable successor formula. 
    Let $p$ be the position of $\diaAn\xi^1_\Delta$ in the list made at the application of the \rnext on $\Delta$. 
    Note that $\xi$ is a $\gamma$-formula $\in cl(\initf)$, where at least one of its $\gamma$-components, obtained from a pair $\langle \psi,\Psi \rangle$, is such that $\fonctA{\Psi}{FE(\psi)}=true$. Let $\chi$ be such a $\gamma$-component. 
    So Lemma \ref{lem:outcomeSet} is applicable to $\Delta$, and according to that corollary, there exists an
	outcome set $X_\chi$ of $\sigma_A[\diaAn\xi^1_\Delta]$ at $\Delta$ such that, for every $\Delta' \in X_\chi$, there exists $s' \in \Out(s,\sigma)$ such that $\m,s'\models\Delta'$, and moreover, if $\m,s'\models\chi$, then $\chi \in \Delta'$. 
	We start building the realization witness tree $\mr$ with a simple tree whose root $r$ is coloured with $\Delta$ and whose leaves are coloured bijectively with a node from $X$. This first tree respects Items 1 to 3 of Definition \ref{def:realTree1}; some of the leaves respect Item 4 of this definition, but not all of them. The next part treats these leaves.
	
	Since, $\m,s \models \xi^1_\Delta$ with $\xi^1_\Delta = \diaA\Phi'$ for every $s' \in \Out(s,\sigma_A)$, it follows that for every such $s'$ there exists a perfect-recall $A$-strategy $F^{s'}_A$ such that for every $\lambda \in \Plays(s',F_A^{s'})$, $\m,\lambda\models \Phi'$. Then, playing $\sigma_A$ followed by playing $F^{s'}_A$ constitutes a perfect recall strategy $F_A$ witnessing the truth of $\xi$ at s.
	
	Then we continue the construction of $\mr$ as follows. Let $S'$ be the set of all sates $s''$ appearing as part of a play consistent with $F^{s'}_A$, containing a descendant eventuality $\xi^i$ of $\xi$ and satisfying the requirement that $\m,s''\not\models\chi$, for all $\gamma$-components $\chi$ obtained from a pair $\langle \psi,\Psi \rangle$ such that $\fonctA{\Psi}{FE(\psi)} = true$.
 For every $s'\in \Out(s,\sigma_A)$, we follow the perfect recall strategy $F^{s'}_A$, matching every state $s'' \in S'$ with a node $w''$ of $\mr$ and matching every $A'$-action of $F^{s'}_A$ at $s''$ with the tableau $\sigma_A[\diaAn\xi^i] \in \vectorsatstate{\Delta''}$ where $\xi^i$ is the descendant eventuality of $\xi$ in $w''$ and $\Delta''$ is the state colouring the node $w''$. We follow this way
 each $F^{s'}_A$ along each run until we reach a state $t$ where $\m,t \models \chi$ ($\chi$ is as described above). This means that we have reached a leaf of $\mr$; this leaf respects item 4 of Definition \ref{def:realTree1}. As $\m,s\models\xi$, such a state can be reached for each run, and we thus obtain a finite tree $\mr$.

	Thus, the so constructed realization witness tree conforms to Definition \ref{def:realTree1}. \qed
\end{enumerate}

\end{proof}

\begin{lemma}
\label{ERevent_OK}
 Let $\Delta$ be a state in $\tabb{n}$. If $\Delta$ is satisfiable then \edeux cannot eliminate $\Delta$ from $\tabb{n}$.
\end{lemma}

\begin{proof}
Let $\Delta \in \tabb{n}$ be a satisfiable state.\\
If $\Delta$ contains no eventuality, then \edeux is not applicable. \\
If $\Delta$ contains an eventuality $\xi$, then Lemma  \ref{eventualities_realized} ensures that there exists a realization witness tree for $\Delta$ and, by Lemma \ref{lemma:realization} we know that $\xi$ is realized at $\Delta$ in $\tabb{n}$. 
Therefore, \edeux cannot eliminate $\Delta$ for $\tabb{n}$. \qed
\end{proof}

\begin{theorem}[Soundness]
 If $\initf$ is satisfiable, then $\tabb{}$ is open.
\end{theorem}

\begin{proof}
Lemmas \ref{states_wrt_prestates}--\ref{ERevent_OK} ensure that if $\Delta$ is satisfiable, then $\Delta$ cannot be eliminated from $\tabb{n}$ due to \eun or \edeux. 
Moreover, Lemma \ref{lemme_sat1} ensures that if the input formula $\initf$ is satisfiable, then at 
least one state containing $\initf$ (created from the initial prestate) is satisfiable. Thus, this state cannot be eliminated and therefore the final tableau $\tabb{}$ is open. \qed
\end{proof}

%%%%%%%%%%%%%%%%%%%%%%%%%%%%%%
\section{Completeness, model synthesis and complexity}
\label{sec:compl}
%%%%%%%%%%%%%%%%%%%%%%%%%%%%%%

\subsection{Hintikka Structures}
\label{subsec:HS}

The tableau procedure actually attempts to build not a concurrent game model of the input formula but a state-labelled non-deterministic CGS, from which structures of a special kind can be extracted which essentially are partly defined concurrent game models. Following \cite{Pratt80,BPM83,GorankoShkatov09ToCL} we will call them
 \emph{Hintikka structures}. Here we will give the definition of a Hintikka structure for a given \ATLp formula $\initf$ and will show how to obtain a CGM for $\initf$ from a Hintikka structure for $\initf$. Later we will explain how to extract a Hintikka structure `satisfying' the input formula from its open final tableau. 

\begin{definition}
 \label{def:hintikka}
A \keyterm{Concurrent Game Hintikka Structure} (for short, CGHS)
 is a deterministic CGS 
 $\Hmodel = (\agents, \stat, \{\Actions_\aga\}_{\aga\in\agents}, \{\dd_{\aga}\}_{\aga\in\agents}, \out,H)$ which is 
 state-labelled by a given set $\Gamma$ of \ATLp-formulae with a state-labelling function $H$. Let $s \in \stat$ be a state of $\Hmodel$. An Hintikka structure $\Hmodel$ satisfies the following constraints: 
\begin{description}
 \item[H1] If $\varphi \in H(s)$ then  $\neg \varphi \not \in H(s)$;
 \item[H2] If an $\alpha$-formula belongs to $H(s)$, then its both $\alpha$-components do;
 \item[H3] If a $\beta$-formula belongs to $H(s)$, then one of its $\beta$-components does;
 \item[H4] If a $\gamma$-formula belongs to $H(s)$, then one of its $\gamma$-components does;
 \item[H5] If $\diaAn\psi \in  H(s)$, then there exists an A-action $\sigma_A \in \vectorsatstate{s}$ such that $\psi \in H(s')$ for all $s' \in \Out(s,\sigma_A)$. 
 Likewise, if $\crochetAn\psi \in H(s)$, then there exists an $A$-co-action $\sigma_A^c \in \ensCoVect{A}{s}$ such that, for all $\psi \in H(s')$ for all $s' \in \Out(s,\sigma^c_A)$.
 \item[H6] If a potential eventuality $\xi=\diaA \Phi$ (resp.  $\xi=\crochetA \Phi$) belongs to $H(s)$,
then there exists a realization witness tree, rooted at $s$ in $\Hmodel$ for $\xi=\diaA \Phi$ (resp.  $\xi=\crochetA \Phi$) at $s$.
\end{description}
\end{definition}

\begin{remark}
The condition H6 is well defined because a Hintikka structure is obtained by colouring via $H$  from a deterministic concurrent game structure, for which the notion of realization witness tree is defined.

\end{remark}

\begin{definition}
 Let $\Hmodel=(\agents, \stat, \{\Actions_\aga\}_{\aga\in\agents}, \{\dd_{\aga}\}_{\aga\in\agents}, \out, H)$ be a CGHS and $\initf$  be an \ATLp-formula. We say that 
 $\Hmodel$ 
 is a concurrent game Hintikka structure for $\initf$,
 if $\initf \in H(s)$ for some $s \in \stat$.
\end{definition}

We now show that from any CGHS for a given formula $\initf$ a CGM satisfying $\initf$ can be obtained. 

\begin{theorem}
\label{thm:FromCgsToModel}
Let $\Hmodel=(\agents, \stat, \{\Actions_\aga\}_{\aga\in\agents}, \{\dd_{\aga}\}_{\aga\in\agents}, \out, H)$  be a CGHS for a given \ATLp-formula $\initf$. 
Let further $\m=(\agents, \stat, \{\Actions_\aga\}_{\aga\in\agents}, \{\dd_{\aga}\}_{\aga\in\agents}, \out,\prop,\lab)$ be the CGM obtained from 
$\Hmodel$ by setting, for every
$s \in \stat$, $\lab(s) = H(s) \cap \prop$. Then, for every
$s \in \stat$ and  every  \ATLp formula $\varphi$, $\varphi \in H(s)$ implies $\m, s\models \varphi$.
In particular, $\m$ satisfies $\initf$.
\end{theorem}

\begin{proof}

Suppose $\varphi \in H(s)$. We will prove that $\m,s \models\varphi$ by induction on the structure of the state formula $\varphi$.

\textit{Base}.  If $\varphi \in \prop \cup \{\top\}$ belongs to $H(s)$, it is immediate that 
$\m, s \models \varphi$,  by definition of $\lab$ and H1.

\medskip
\textit{Inductive Step}. 
\begin{itemize}
\item $\varphi$ is $\psi_1  \wedge \psi_2$. By H2 we get that $\psi_1 \in H(s)$ and  $\psi_2 \in H(s)$ .
By inductive hypothesis $\m, s \models \psi_1$ and $\m, s \models \psi_2$. Therefore
$\m , s\models \varphi$.
\item $\varphi$ is $\psi_1  \vee \psi_2$. By H3 we get that either $\psi_1 \in H(s)$ or $\psi_2 \in H(s)$ .
By inductive hypothesis either $\m , s\models \psi_1$ or $\m, s \models \psi_2$. Therefore
$\m, s \models \varphi$.
\item $\varphi$ is $\diaAn\psi$ or $\crochetAn\psi$. An application of 
H5 and the inductive hypothesis to $\psi$ imply that $\m , s\models \varphi$.
\item $\varphi$ is $\diams{A} \Phi$ or $\varphi$ is $\crochet{A} \Phi$ , where $\Phi$ is a path formula 
whose main operator is different from $\rond$, that is $\varphi$ is a $\gamma$-formula.
Here we only present in detail the first case, the second one being quite similar.
We need to prove the existence of a (perfect recall) strategy $F_A$ such that, for each branch $\lambda$
in $\m$ stemming from $s$ and consistent with that strategy, 
 $\m, \lambda \models
\Phi$. This will imply that $\m , s\models \varphi$.
Since  $\varphi=\diaA \Phi \in H(s)$ by hypothesis, then H6 guarantees the existence of 
a realization witness tree  $T$ on $\Hmodel$ for $\varphi$.
By construction, $T$ provides a partial finite strategy $Fp_A$, defined only for
the finite set of histories occurring in $T$ and having length strictly
less than  the height of $T$. We want to show that  $Fp_A$ can be extended to
a strategy $F_A$, defined for all the histories in $\Hmodel$, and that $T$ can be extended to a possibly infinite tree $T'$ 
such that:
\begin{itemize}
\item Each node of $T'$ is also a node of $\Hmodel$ and each labelled edge of $T'$ is also a labelled edge of $\Hmodel$.
\item All paths in $T'$ are consistent with $F_A$, hence
$T'$ witnesses the truth of  $\diaA \Phi$ at state $s$ of $\m$ by instantiating the quantifier $\diaA$ to $F_A$.
\end{itemize}

Below, we show how to construct $F_A$ and $T'$. Let us consider any finite path in $T$ of the form $\lambda_{\leq n}$, where $\lambda_0=s$ and $\lambda_n$ is a leaf.
By construction of $T$, each node $\lambda_i$, for $1 \leq i \leq n$, 
is a node of $\Hmodel$ and each labelled edge of $T$ is a labelled edge of $\Hmodel$.
The descendant potential eventuality $\varphi^n$ of $\varphi$ belongs to the colour of $\lambda_n$ by construction of $T$. 
Since $\lambda_n$ is a node of $\Hmodel$ and $\varphi^n \in H(\lambda_n)$, by H4 some $\gamma$-component $\chi$ of $\varphi^n$ belongs to $H(\lambda_n)$. 
This formula $\chi$ is either of the form $\psi$ or of the form $\psi \wedge \diaAn\varphi^{n+1}$
(the second case occurs, for instance, when $\varphi$ has the form $\diaA \Box \theta$).

\ \ \ \ In the first case, any extension of the partial strategy $Fp_A$  and any extension of  $\lambda_{\leq n}$ to an infinite path will do.
 
\ \ \ In the second case, we apply H2 to get $\psi \in H(\lambda_n)$ and $\diaAn\varphi^{n+1} \in H(\lambda_n)$. By H5, there exists an A-action $\sigma_A \in \vectorsatstate{H(\lambda_n)}$
such that $\varphi^{n+1} \in H(s')$ for all $s' \in \Out(\lambda_n,\sigma_A)$.
Playing this A-action $\sigma_A$ after the partial strategy $Fp_A$ gives us a new partial strategy $Fp'_A$ defined for histories whose length is less than or equal to $n$. 
The set of successors of $\lambda_n$ for $T'$ is the set 
$\Out(\lambda_n,\sigma_A)$. For each $s' \in \Out(\lambda_n,\sigma_A)$, we can again apply H2, H4 and H5 to get a new partial strategy $Fp''_A$ defined for histories whose length is inferior or equal to $n+1$. For any $s' \in \Out(\lambda_n,\sigma_A)$, its successors are obtained by an application of $Fp''_A$. An infinite iteration of this procedure will give us the complete strategy $F_A$ and the way to extend the finite tree so as to get $T'$. 

\end{itemize}
\end{proof}

%%%%%%%%%%%%%%%
\subsection{Completeness and model synthesis}
\label{subsec:completeness}

\begin{theorem}\label{thm:complete}
The tableau method for  \ATLp is complete.
\end{theorem}

Completeness of the procedure means that an open tableau for $\initf$ implies existence of a CGM model for $\initf$. 
So, we start with an open tableau $\tabb{}$ for $\initf$ and we want to prove that $\initf$ is indeed satisfiable. The proof is constructive, as we will build from $\tabb{}$ a Hintikka structure $\hint$ that can be turned into a model for $\initf$. In order to construct that  Hintikka structure, first we will extract special $\tabb{}$-trees associated with potential eventualities, that can be seen as building modules to be used to construct the entire structure. Eventually, we show that the so constructed structure is a Hintikka structure for $\initf$. 

First, we need to define \emph{edge-labelling} of a tree.  

\begin{definition}
 Let $\treeW = (W,\rightsquigarrow)$ be a tree and $Y$ be a non-empty set. An \keyterm{edge-labelling of $\treeW$ by $Y$} is a mapping $l$ from the set of edges of 
 $\treeW$ to the set of non-empty subsets of $Y$. 
\end{definition}

\begin{definition}
Given a tableau $\tabb{}$, a tree $\treeW = (W,\rightsquigarrow)$ is a \keyterm{$\tabb{}$-tree} if the following conditions hold:
\begin{itemize}
 \item $\treeW$ is $\settab{}$-coloured, by some colouring mapping $c$. 
 \item $\treeW$ is edge-labelled by $\bigcup_{(\Delta \in \settab{})}\vectorsatstate{\Delta}$, by some edge-labelling mapping $l$;
 \item $l(w \rightsquigarrow w') \subseteq \vectorsatstate{\Delta}$ for every $w \in W$ with $c(w) = \Delta$;
 \item For every interior node $w \in W$ with $c(w) = \Delta$ and every successor $\Delta' \in \tabb{}$ of $\Delta$, 
 there exists exactly one $w' \in W$ such that $l(w \rightsquigarrow w') = \{\sigma \mid \Delta \stackrel{\sigma}{\longrightarrow} \Delta')\}$.
\end{itemize} 
\end{definition}

\begin{definition}
 Let $\Delta \in \settab{}$. A $\tabb{}$-tree $\treeW$ is \keyterm{rooted at} $\Delta$ if the root $r$ of $\treeW$ is coloured with $\Delta$. 
\end{definition}

For the purpose of our construction, we distinguish two kinds of $\tabb{}$-trees: \textit{simple} or \textit{realizing}. Their definitions are given below. Realizing $\tabb{}$-trees will deal especially with potential eventualities.

\begin{definition}
 A tree $\treeW = (W,\rightsquigarrow)$ is \keyterm{simple} if it has no interior nodes except the root. 
\end{definition}

Simple $\tabb{}$-trees can be seen as one-step modules.

\begin{definition}
 Let $\treeW = (W,\rightsquigarrow)$ be a $\tabb{}$-tree rooted at $\Delta$ and $\xi \in \Delta$ a potential eventuality. The tree $\treeW$ is a \keyterm{realizing $\tabb{}$-tree for $\xi$}, denoted $\treeW_{\xi}$,
if there exists a subtree $\mathcal{R}_{\xi}$ of $\treeW$ rooted at $\Delta$ such that $\mathcal{R}_{\xi}$ is a realization witness tree for $\xi$ rooted at $\Delta \in \tabb{}$. 
\end{definition}

\begin{lemma}
\label{lem:simpletree}
 Let $\Delta \in \settab{}$. Then, there exists a simple $\tabb{}$-tree rooted at $\Delta$.
\end{lemma}

\begin{proof}
 We construct a simple $\tabb{}$-tree $\treeW$ rooted at $\Delta$ as follows.
 The root of $\treeW$ is a node $r$ such that $c(r) = \Delta$. 
For every successor state $\Delta'$ of $\Delta \in \tabb{}$, let $Moves(\Delta, \Delta') = \{\sigma \mid \Delta \stackrel{\sigma}{\longrightarrow} \Delta'\}$. Note that, by construction of 
the tableau, the family $\{ Moves(\Delta, \Delta') \mid \Delta' \mbox { is a successor of } \Delta \}$ is a partition of the set $\vectorsatstate{\Delta}$ of all action profiles applied at $\Delta$.  Now, for each set $X$ of that family we select one successor $\Delta'$ of $\Delta$ such that $X=\dd(\Delta,\Delta')$ and add a successor $t$ to $\treeW$ such that $c(t) = \Delta'$ and $l(r \rightsquigarrow t) = \{\sigma \mid \Delta \stackrel{\sigma}{\longrightarrow} \Delta'\}$.  
\qed
\end{proof}

\begin{example}
\label{run-ex10}
(\textit{Continuation of Example \ref{run-ex7}})
 
 Consider the final tableau $\tabfin{\vartheta}$ 
 for the formula $\vartheta = \diams{1}(p\until q \lor \Box q) \land \crochet{2}(\Diamond p \land \Box \neg q)$. 
 
 We have seen in the example \ref{run-ex7} that 
 $\mathcal{S}^{\vartheta}  = \{\Delta_1,\dots,\Delta_8\}$.
 
We have listed possible simple $\mathcal{T}^\vartheta$-trees rooted at each $\Delta_i$ in the table on Figure \ref{fig:SimpleTrees}.

\begin{figure}
\begin{center}
\begin{tabular}{|c|c|}
%% Ligne 1
\hline
\begin{tikzpicture}[node distance=3cm]
 %Définition des styles
  \tikzstyle{etat}=[text badly centered]
  \tikzstyle{fleche}=[->,>=latex]
 %Description des noeuds
  \node[etat, text width=1cm] (D1) {$\Delta_1$};
  \node[etat, text width=1cm] (D3) [above right of=D1,yshift=-1cm] {$\Delta_3$} ;
  \node[etat, text width=1cm] (D5) [below right of=D1,yshift=1cm] {$\Delta_5$} ;
 %Description des liens
 \path[fleche]
    (D1) edge node [text width=0.5cm] {\tiny $0,0$ $0,1$} (D3)
	 edge node [text width=0.5cm] {\tiny $1,0$ $1,1$} (D5);
\end{tikzpicture}
 & 
 \begin{tikzpicture}[node distance=3cm]
 %Définition des styles
  \tikzstyle{etat}=[text badly centered]
  \tikzstyle{fleche}=[->,>=latex]
 %Description des noeuds
  \node[etat, text width=1cm] (D2) {$\Delta_2$};
  \node[etat, text width=1cm] (D3) [above right of=D2,yshift=-1cm] {$\Delta_3$} ;
  \node[etat, text width=1cm] (D7) [below right of=D2,yshift=1cm] {$\Delta_7$} ;
 %Description des liens
 \path[fleche]
    (D2) edge node [text width=0.5cm] {\tiny $0,0$ $0,1$} (D3)
	 edge node [text width=0.5cm] {\tiny $1,0$ $1,1$} (D7);
\end{tikzpicture}
\\\hline
%% Ligne 2
 \begin{tikzpicture}[node distance=2cm]
 %Définition des styles
  \tikzstyle{etat}=[text badly centered]
  \tikzstyle{fleche}=[->,>=latex]
 %Description des noeuds
  \node[etat, text width=1cm] (D3) {$\Delta_3$};
  \node[etat, text width=1cm] (D4) [right of=D3] {$\Delta_4$} ;
 %Description des liens
 \path[fleche]
    (D3) edge node [above] {\tiny $0,0$} (D4);
\end{tikzpicture} 
& 
\begin{tikzpicture}[node distance=2cm]
 %Définition des styles
  \tikzstyle{etat}=[text badly centered]
  \tikzstyle{fleche}=[->,>=latex]
 %Description des noeuds
  \node[etat, text width=1cm] (D4) {$\Delta_4$};
  \node[etat, text width=1cm] (D8) [right of=D4] {$\Delta_8$} ;
 %Description des liens
 \path[fleche]
    (D4) edge node [above] {\tiny $0,0$} (D8);
\end{tikzpicture} 
\\\hline
%% Ligne 3
\begin{tikzpicture}[node distance=2cm]
 %Définition des styles
  \tikzstyle{etat}=[text badly centered]
  \tikzstyle{fleche}=[->,>=latex]
 %Description des noeuds
  \node[etat, text width=1cm] (D5) {$\Delta_5$};
  \node[etat, text width=1cm] (D5b) [right of=D5] {$\Delta_5$} ;
 %Description des liens
 \path[fleche]
    (D5) edge node [above] {\tiny $0,0$} (D5b);
\end{tikzpicture} 
& 
\begin{tikzpicture}[node distance=2cm]
 %Définition des styles
  \tikzstyle{etat}=[text badly centered]
  \tikzstyle{fleche}=[->,>=latex]
 %Description des noeuds
  \node[etat, text width=1cm] (D6) {$\Delta_6$};
  \node[etat, text width=1cm] (D7) [right of=D6] {$\Delta_7$} ;
 %Description des liens
 \path[fleche]
    (D6) edge node [above] {\tiny $0,0$} (D7);
\end{tikzpicture} 
\\\hline
%% Ligne 4
\begin{tikzpicture}[node distance=2cm]
 %Définition des styles
  \tikzstyle{etat}=[text badly centered]
  \tikzstyle{fleche}=[->,>=latex]
 %Description des noeuds
  \node[etat, text width=1cm] (D7) {$\Delta_7$};
  \node[etat, text width=1cm] (D7b) [right of=D7] {$\Delta_7$} ;
 %Description des liens
 \path[fleche]
    (D7) edge node [above] {\tiny $0,0$} (D7b);
\end{tikzpicture} 
& 
\begin{tikzpicture}[node distance=2cm]
 %Définition des styles
  \tikzstyle{etat}=[text badly centered]
  \tikzstyle{fleche}=[->,>=latex]
 %Description des noeuds
  \node[etat, text width=1cm] (D8) {$\Delta_8$};
  \node[etat, text width=1cm] (D8b) [right of=D8] {$\Delta_8$} ;
 %Description des liens
 \path[fleche]
    (D8) edge node [above] {\tiny $0,0$} (D8b);
\end{tikzpicture} 
\\\hline
 \end{tabular}
 \caption{Simple $\mathcal{T}^\vartheta$-trees} 
\vspace{-0.5cm}
\label{fig:SimpleTrees}
 \end{center}
 \end{figure}
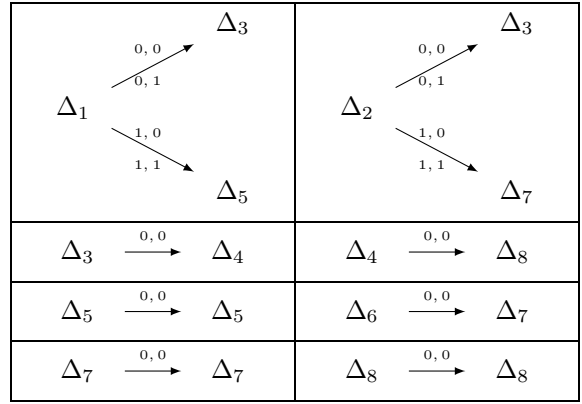
\end{example}

To show the existence of a realizing $\tabb{}$-tree for $\xi$ at $\Delta$, we first prove the existence of a realization witness tree $\mr_\xi$ for $\xi$ at $\Delta$.

\begin{lemma}
\label{lem:realization}
 Let $\tabb{}$ be a tableau for $\initf$ and $\xi$ be a potential eventuality realized at $\Delta \in \tabb{}$.  Then, there  exists a realization witness tree $\mr_\xi$ for $\xi$ at $\Delta$ in $\tabb{}$. 
\end{lemma}

\begin{proof}
We give detailed proof only for the case where $\xi = \diaA\Phi$; the other case is similar, just replace $\diaA$ by $\crochetA$ in the proof. 
Suppose that  $\xi$ is realized at $\Delta$ in $\tabb{}$. 
We define the \keyterm{rank of $\xi$ at $\Delta$ in $\tabb{}$}, denoted $rank(\xi, \Delta, \tabb{})$ to be the minimal length of a chain of descendant potential eventualities $\xi=\xi^0,...,\xi^n=\diaAn\Phi_n$ ensuring the realization of $\xi$, that is, $Real(\Phi_n, \Delta_j)= true$ for some state $\Delta_j$ descendant of $\Delta$ in $\tabb{}$. 
We prove the existence of a realization witness tree $\mr$ for $\xi$ at $\Delta$ in $\tabb{}$ by induction on rank $rank(\xi, \Delta, \tabb{})$.

Base: $rank(\xi, \Delta, \tabb{}) = 0$. 
Here $\xi$ is immediately realized and $\mr_\xi$ contains only the root coloured with $\Delta$.

Inductive step: $rank(\xi, \Delta, \tabb{}) = k$ where $k > 0$. Since $\xi$ is realized at $\Delta \in \tabb{}$ and $rank(\xi, \Delta, \tabb{}) > 0$, by Definition \ref{def:realEvent} we have that for every $\sigma \in D(\Delta,\diaAn\xi^1)$ there exists $\Delta' \in \tabb{}$ such that $\Delta \stackrel{\sigma}{\longrightarrow} \Delta'$ and $\xi^1$ is realized at $\Delta' \in \tabb{}$. We build a tree $T$ rooted at a node $r$ coloured with $\Delta$ where the children $v$ of $r$ are coloured bijectively with the set of $\Delta'$ obtained above. Then $rank(\xi^1, \Delta', \tabb{}) = k-1$ and we can apply the inductive hypothesis to obtain a realization witness tree $\mr_{\xi^1}$ for $\xi^1$ at $\Delta'$ in $\tabb{}$ for each $\Delta'$. Thus, replacing each node $v$ of $T$ by the corresponding $\mr_{\xi^1}$ gives us $\mr_\xi$.
\end{proof}

\begin{lemma}
\label{lem:H6}
 Let $\xi \in \Delta \in \settab{}$ be a potential eventuality. Then, there exists a finite realizing $\tabb{}$-tree for $\xi$ rooted at $\Delta$.
\end{lemma}

\begin{proof}
 Since $\tabb{}$ is open, $\xi$ is realized at $\Delta$ in $\tabb{}$. To construct the realizing $\tabb{}$-tree $\treeW_{\xi}$ for $\xi$ rooted at $\Delta$, we start from the realization witness tree $\mr_{\xi}$, whose existence is given by Lemma \ref{lem:realization} and provisionally we take $\treeW_\xi$ to be $\mr_\xi$. The problem with $\mr_{\xi}$ is that for some $\sigma \in \vectorsatstate{\Delta}$ at some node $w$ of $\mr_{\xi}$ , there is no edge $w \rightsquigarrow w'$ such that $l(w\rightsquigarrow w')\ni \sigma $.
Therefore, to extend $\treeW_{\xi}$ into a realizing $\tabb{}$-tree, for every such node $w$, we pick one of the successor states of $c(w)$ via $\sigma$, say $\Delta'$ and add a node $w'$ to $\treeW_{\xi}$ such that $c(w')=\Delta'$ and $l(w \rightsquigarrow w') \ni \sigma$. \qed
\end{proof}

\begin{example}
\label{run-ex11}
(\textit{Continuation of Example \ref{run-ex10}})  
 We now extract in Figure \ref{fig:RealizingTrees} a possible realizing tree from the open tableau for every state $\Delta_i$, using realization witness trees.
 
 \begin{figure}
 \begin{center}
\begin{tabular}{|c|c|c|c|c|}\hline
State & Eventualities & Rank & Realization Witness Tree & Realizing Tree \\\hline
$\Delta_1$ & $\xi_1=\diams{1}(p\until q \lor \Box q)$ & 1 & 
\begin{tikzpicture}[node distance=2cm]
 %Définition des styles
  \tikzstyle{etat}=[text badly centered]
  \tikzstyle{fleche}=[->,>=latex]
 %Description des noeuds
  \node[etat, text width=1cm] (D1) {$\Delta_1$};
  \node[etat, text width=1cm] (D4) [right of=D1] {$\Delta_4$} ;
 %Description des liens
 \path[fleche]
    (D1) edge node [text width=0.5cm] {\tiny $0,0$ $0,1$} (D4);
\end{tikzpicture}  
&
\begin{tikzpicture}[node distance=3cm]
 %Définition des styles
  \tikzstyle{etat}=[text badly centered]
  \tikzstyle{fleche}=[->,>=latex]
 %Description des noeuds
  \node[etat, text width=1cm] (D1) {$\Delta_1$};
  \node[etat, text width=1cm] (D4) [above right of=D1,yshift=-1cm] {$\Delta_4$} ;
  \node[etat, text width=1cm] (D5) [below right of=D1,yshift=1cm] {$\Delta_5$} ;
 %Description des liens
 \path[fleche]
    (D1) edge node [text width=0.5cm] {\tiny $0,0$ $0,1$} (D4)
	 edge node [text width=0.5cm] {\tiny $1,0$ $1,1$} (D5);
\end{tikzpicture}
\\\hline
$\Delta_1$ & $\xi_2=\crochet{2}(\Diamond p \land \Box \neg q)$ & 1 & 
\begin{tikzpicture}[node distance=2cm]
 %Définition des styles
  \tikzstyle{etat}=[text badly centered]
  \tikzstyle{fleche}=[->,>=latex]
 %Description des noeuds
  \node[etat, text width=1cm] (D1) {$\Delta_1$};
  \node[etat, text width=1cm] (D6) [right of=D1] {$\Delta_6$} ;
 %Description des liens
 \path[fleche]
    (D1) edge node [text width=0.5cm] {\tiny $1,0$ $1,1$} (D6);
\end{tikzpicture} 
 & 
\begin{tikzpicture}[node distance=3cm]
 %Définition des styles
  \tikzstyle{etat}=[text badly centered]
  \tikzstyle{fleche}=[->,>=latex]
 %Description des noeuds
  \node[etat, text width=1cm] (D1) {$\Delta_1$};
  \node[etat, text width=1cm] (D4) [above right of=D1,yshift=-1cm] {$\Delta_4$} ;
  \node[etat, text width=1cm] (D6) [below right of=D1,yshift=1cm] {$\Delta_6$} ;
 %Description des liens
 \path[fleche]
    (D1) edge node [text width=0.5cm] {\tiny $0,0$ $0,1$} (D4)
	 edge node [text width=0.5cm] {\tiny $1,0$ $1,1$} (D6);
\end{tikzpicture}
 \\\hline
$\Delta_2$ & $\xi_1=\diams{1}(p\until q \lor \Box q)$ & 1 & 
\begin{tikzpicture}[node distance=2cm]
 %Définition des styles
  \tikzstyle{etat}=[text badly centered]
  \tikzstyle{fleche}=[->,>=latex]
 %Description des noeuds
  \node[etat, text width=1cm] (D2) {$\Delta_2$};
  \node[etat, text width=1cm] (D4) [right of=D2] {$\Delta_4$} ;
 %Description des liens
 \path[fleche]
    (D2) edge node [text width=0.5cm] {\tiny $0,0$ $0,1$} (D4);
\end{tikzpicture}  
& 
\begin{tikzpicture}[node distance=3cm]
 %Définition des styles
  \tikzstyle{etat}=[text badly centered]
  \tikzstyle{fleche}=[->,>=latex]
 %Description des noeuds
  \node[etat, text width=1cm] (D1) {$\Delta_2$};
  \node[etat, text width=1cm] (D4) [above right of=D1,yshift=-1cm] {$\Delta_4$} ;
  \node[etat, text width=1cm] (D7) [below right of=D1,yshift=1cm] {$\Delta_7$} ;
 %Description des liens
 \path[fleche]
    (D1) edge node [text width=0.5cm] {\tiny $0,0$ $0,1$} (D4)
	 edge node [text width=0.5cm] {\tiny $1,0$ $1,1$} (D7);
\end{tikzpicture}
 \\\hline
$\Delta_2$ & $\xi_2=\crochet{2}(\Diamond p \land \Box \neg q)$ & 1 & 
\begin{tikzpicture}[node distance=2cm]
 %Définition des styles
  \tikzstyle{etat}=[text badly centered]
  \tikzstyle{fleche}=[->,>=latex]
 %Description des noeuds
  \node[etat, text width=1cm] (D2) {$\Delta_2$};
  \node[etat, text width=1cm] (D7) [right of=D2] {$\Delta_7$} ;
 %Description des liens
 \path[fleche]
    (D2) edge node [text width=0.5cm] {\tiny $1,0$ $1,1$} (D7);
\end{tikzpicture}  
& 
\begin{tikzpicture}[node distance=3cm]
 %Définition des styles
  \tikzstyle{etat}=[text badly centered]
  \tikzstyle{fleche}=[->,>=latex]
 %Description des noeuds
  \node[etat, text width=1cm] (D1) {$\Delta_2$};
  \node[etat, text width=1cm] (D4) [above right of=D1,yshift=-1cm] {$\Delta_4$} ;
  \node[etat, text width=1cm] (D7) [below right of=D1,yshift=1cm] {$\Delta_7$} ;
 %Description des liens
 \path[fleche]
    (D1) edge node [text width=0.5cm] {\tiny $0,0$ $0,1$} (D4)
	 edge node [text width=0.5cm] {\tiny $1,0$ $1,1$} (D7);
\end{tikzpicture}
 \\\hline
$\Delta_3$ & $\xi_3=\xi_1^1=\diams{1}p\until q$ & 1 & 
\begin{tikzpicture}[node distance=2cm]
 %Définition des styles
  \tikzstyle{etat}=[text badly centered]
  \tikzstyle{fleche}=[->,>=latex]
 %Description des noeuds
  \node[etat, text width=1cm] (D3) {$\Delta_3$};
  \node[etat, text width=1cm] (D4) [right of=D3] {$\Delta_4$} ;
 %Description des liens
 \path[fleche]
    (D3) edge node [above] {\tiny $0,0$} (D4);
\end{tikzpicture} 
&
\begin{tikzpicture}[node distance=2cm]
 %Définition des styles
  \tikzstyle{etat}=[text badly centered]
  \tikzstyle{fleche}=[->,>=latex]
 %Description des noeuds
  \node[etat, text width=1cm] (D3) {$\Delta_3$};
  \node[etat, text width=1cm] (D4) [right of=D3] {$\Delta_4$} ;
 %Description des liens
 \path[fleche]
    (D3) edge node [above] {\tiny $0,0$} (D4);
\end{tikzpicture} 
\\\hline
$\Delta_4$ & $\xi_3=\xi_1^1=\diams{1}p\until q$ & 0 & $\Delta_4$
& 
\begin{tikzpicture}[node distance=2cm]
 %Définition des styles
  \tikzstyle{etat}=[text badly centered]
  \tikzstyle{fleche}=[->,>=latex]
 %Description des noeuds
  \node[etat, text width=1cm] (D4) {$\Delta_4$};
  \node[etat, text width=1cm] (D8) [right of=D4] {$\Delta_8$} ;
 %Description des liens
 \path[fleche]
    (D4) edge node [above] {\tiny $0,0$} (D8);
\end{tikzpicture}\\\hline
$\Delta_5$ & $\xi_2=\crochet{2}(\Diamond p \land \Box \neg q)$ & 1 & 
\begin{tikzpicture}[node distance=2cm]
 %Définition des styles
  \tikzstyle{etat}=[text badly centered]
  \tikzstyle{fleche}=[->,>=latex]
 %Description des noeuds
  \node[etat, text width=1cm] (D5) {$\Delta_5$};
  \node[etat, text width=1cm] (D6) [right of=D5] {$\Delta_6$} ;
 %Description des liens
 \path[fleche]
    (D5) edge node [above] {\tiny $0,0$} (D6);
\end{tikzpicture} 
& 
\begin{tikzpicture}[node distance=2cm]
 %Définition des styles
  \tikzstyle{etat}=[text badly centered]
  \tikzstyle{fleche}=[->,>=latex]
 %Description des noeuds
  \node[etat, text width=1cm] (D5) {$\Delta_5$};
  \node[etat, text width=1cm] (D6) [right of=D5] {$\Delta_6$} ;
 %Description des liens
 \path[fleche]
    (D5) edge node [above] {\tiny $0,0$} (D6);
\end{tikzpicture}
\\\hline
$\Delta_6$ & $\xi_2=\crochet{2}(\Diamond p \land \Box \neg q)$ & 0 & $\Delta_6$
& 
\begin{tikzpicture}[node distance=2cm]
 %Définition des styles
  \tikzstyle{etat}=[text badly centered]
  \tikzstyle{fleche}=[->,>=latex]
 %Description des noeuds
  \node[etat, text width=1cm] (D6) {$\Delta_6$};
  \node[etat, text width=1cm] (D7) [right of=D6] {$\Delta_7$} ;
 %Description des liens
 \path[fleche]
    (D6) edge node [above] {\tiny $0,0$} (D7);
\end{tikzpicture}\\\hline
$\Delta_7$ & $\xi_4=\xi_2^1 = \crochet{2}\Box \neg q$ & 0 & $\Delta_7$ &
\begin{tikzpicture}[node distance=2cm]
 %Définition des styles
  \tikzstyle{etat}=[text badly centered]
  \tikzstyle{fleche}=[->,>=latex]
 %Description des noeuds
  \node[etat, text width=1cm] (D7) {$\Delta_7$};
  \node[etat, text width=1cm] (D7b) [right of=D7] {$\Delta_7$} ;
 %Description des liens
 \path[fleche]
    (D7) edge node [above] {\tiny $0,0$} (D7b);
\end{tikzpicture}\\\hline
 \end{tabular}
 \caption{Eventualities and realizing $\mathcal{T}^\vartheta$-trees} 
%\vspace{-0.5cm}
\label{fig:RealizingTrees}
 \end{center}
 \end{figure} 
\end{example}

We now construct a final structure, denoted by $\mathfrak{F}$, from simple and realizing $\tabb{}$-trees. This construction is made step-by-step. At the end of the construction, we prove that $\mathfrak{F}$ is indeed a Hintikka structure. 

\textbf{Step 1.} We define a grid $\mathcal{F}$ of size $m \times n$, where $m$ is the number of eventualities occurring in $\tabb{} $ and 
$n$ the number of states of $\tabb{}$. Each row of that grid is labelled by one of the potential eventualities and each column by a state of $\tabb{}$ previously ordered by name ($\Delta_i < \Delta_j$ if $i < j$). We denote by $\xi_i$ the eventuality associated to row $0 \leq i \leq m$, we denote by $\Delta_j$ the state associated to the column $0 \leq j \leq n$. The content $\grid(i,j)$ of each intersection between a row $i$ and a column $j$ of $\grid$ is as follows: if $\xi_i \in \Delta_j$, then $\grid(i,j)$ is the realizing $\tabb{}$-tree for $\xi_i$ rooted at $\Delta_j$, whose existence is ensured by Lemma \ref{lem:H6}; otherwise, $\grid(i,j)$ is the simple $\tabb{}$-tree rooted at $\Delta_j$, whose existence is ensured by Lemma \ref{lem:simpletree}.

\begin{example}
\label{run-ex12}
(\textit{Continuation of Example \ref{run-ex11}})  
The grid $\grid$ for our example  has a size $4 \times 8$ and is represented in Figure \ref{fig:grid}.

{
\newcommand{\mc}[3]{\multicolumn{#1}{#2}{#3}}

 \begin{figure} 
\begin{center}
\scalebox{0.8}{
\begin{tabular}{|c|c|c|c|c|c|}\hline
\mc{2}{|c|}{$\mathcal{F}$} & 0 & 1 & 2 & 3\\\hline
\mc{2}{|c|}{} & $\xi_1$ & $\xi_2$ & $\xi_3$ & $\xi_4$\\\hline
% % Ligne 1
0 & $\Delta_1$ & 
\begin{tikzpicture}[node distance=3cm]
 %Définition des styles
  \tikzstyle{etat}=[text badly centered]
  \tikzstyle{fleche}=[->,>=latex]
 %Description des noeuds
  \node[etat, text width=1cm] (D1) {$\Delta_1$};
  \node[etat, text width=1cm] (D4) [above right of=D1,yshift=-1cm] {$\Delta_4$} ;
  \node[etat, text width=1cm] (D5) [below right of=D1,yshift=1cm] {$\Delta_5$} ;
 %Description des liens
 \path[fleche]
    (D1) edge node [text width=0.5cm] {\tiny $0,0$ $0,1$} (D4)
	 edge node [text width=0.5cm] {\tiny $1,0$ $1,1$} (D5);
\end{tikzpicture}
 & 
\begin{tikzpicture}[node distance=3cm]
 %Définition des styles
  \tikzstyle{etat}=[text badly centered]
  \tikzstyle{fleche}=[->,>=latex]
 %Description des noeuds
  \node[etat, text width=1cm] (D1) {$\Delta_1$};
  \node[etat, text width=1cm] (D4) [above right of=D1,yshift=-1cm] {$\Delta_4$} ;
  \node[etat, text width=1cm] (D6) [below right of=D1,yshift=1cm] {$\Delta_6$} ;
 %Description des liens
 \path[fleche]
    (D1) edge node [text width=0.5cm] {\tiny $0,0$ $0,1$} (D4)
	 edge node [text width=0.5cm] {\tiny $1,0$ $1,1$} (D6);
\end{tikzpicture}
 & 
\begin{tikzpicture}[node distance=3cm]
 %Définition des styles
  \tikzstyle{etat}=[text badly centered]
  \tikzstyle{fleche}=[->,>=latex]
 %Description des noeuds
  \node[etat, text width=1cm] (D1) {$\Delta_1$};
  \node[etat, text width=1cm] (D3) [above right of=D1,yshift=-1cm] {$\Delta_3$} ;
  \node[etat, text width=1cm] (D5) [below right of=D1,yshift=1cm] {$\Delta_5$} ;
 %Description des liens
 \path[fleche]
    (D1) edge node [text width=0.5cm] {\tiny $0,0$ $0,1$} (D3)
	 edge node [text width=0.5cm] {\tiny $1,0$ $1,1$} (D5);
\end{tikzpicture}
 &
 \begin{tikzpicture}[node distance=3cm]
 %Définition des styles
  \tikzstyle{etat}=[text badly centered]
  \tikzstyle{fleche}=[->,>=latex]
 %Description des noeuds
  \node[etat, text width=1cm] (D1) {$\Delta_1$};
  \node[etat, text width=1cm] (D3) [above right of=D1,yshift=-1cm] {$\Delta_3$} ;
  \node[etat, text width=1cm] (D5) [below right of=D1,yshift=1cm] {$\Delta_5$} ;
 %Description des liens
 \path[fleche]
    (D1) edge node [text width=0.5cm] {\tiny $0,0$ $0,1$} (D3)
	 edge node [text width=0.5cm] {\tiny $1,0$ $1,1$} (D5);
\end{tikzpicture}
\\\hline
% % Ligne 2
1 & $\Delta_2$ & 
\begin{tikzpicture}[node distance=3cm]
 %Définition des styles
  \tikzstyle{etat}=[text badly centered]
  \tikzstyle{fleche}=[->,>=latex]
 %Description des noeuds
  \node[etat, text width=1cm] (D1) {$\Delta_2$};
  \node[etat, text width=1cm] (D4) [above right of=D1,yshift=-1cm] {$\Delta_4$} ;
  \node[etat, text width=1cm] (D7) [below right of=D1,yshift=1cm] {$\Delta_7$} ;
 %Description des liens
 \path[fleche]
    (D1) edge node [text width=0.5cm] {\tiny $0,0$ $0,1$} (D4)
	 edge node [text width=0.5cm] {\tiny $1,0$ $1,1$} (D7);
\end{tikzpicture}
 &
 \begin{tikzpicture}[node distance=3cm]
 %Définition des styles
  \tikzstyle{etat}=[text badly centered]
  \tikzstyle{fleche}=[->,>=latex]
 %Description des noeuds
  \node[etat, text width=1cm] (D1) {$\Delta_2$};
  \node[etat, text width=1cm] (D4) [above right of=D1,yshift=-1cm] {$\Delta_4$} ;
  \node[etat, text width=1cm] (D7) [below right of=D1,yshift=1cm] {$\Delta_7$} ;
 %Description des liens
 \path[fleche]
    (D1) edge node [text width=0.5cm] {\tiny $0,0$ $0,1$} (D4)
	 edge node [text width=0.5cm] {\tiny $1,0$ $1,1$} (D7);
\end{tikzpicture}
 & 
\begin{tikzpicture}[node distance=3cm]
 %Définition des styles
  \tikzstyle{etat}=[text badly centered]
  \tikzstyle{fleche}=[->,>=latex]
 %Description des noeuds
  \node[etat, text width=1cm] (D2) {$\Delta_2$};
  \node[etat, text width=1cm] (D3) [above right of=D2,yshift=-1cm] {$\Delta_3$} ;
  \node[etat, text width=1cm] (D7) [below right of=D2,yshift=1cm] {$\Delta_7$} ;
 %Description des liens
 \path[fleche]
    (D2) edge node [text width=0.5cm] {\tiny $0,0$ $0,1$} (D3)
	 edge node [text width=0.5cm] {\tiny $1,0$ $1,1$} (D7);
\end{tikzpicture} 
& 
\begin{tikzpicture}[node distance=3cm]
 %Définition des styles
  \tikzstyle{etat}=[text badly centered]
  \tikzstyle{fleche}=[->,>=latex]
 %Description des noeuds
  \node[etat, text width=1cm] (D2) {$\Delta_2$};
  \node[etat, text width=1cm] (D3) [above right of=D2,yshift=-1cm] {$\Delta_3$} ;
  \node[etat, text width=1cm] (D7) [below right of=D2,yshift=1cm] {$\Delta_7$} ;
 %Description des liens
 \path[fleche]
    (D2) edge node [text width=0.5cm] {\tiny $0,0$ $0,1$} (D3)
	 edge node [text width=0.5cm] {\tiny $1,0$ $1,1$} (D7);
\end{tikzpicture}
\\\hline
% % Ligne 3
2 & $\Delta_3$ & 
\begin{tikzpicture}[node distance=2cm]
 %Définition des styles
  \tikzstyle{etat}=[text badly centered]
  \tikzstyle{fleche}=[->,>=latex]
 %Description des noeuds
  \node[etat, text width=1cm] (D3) {$\Delta_3$};
  \node[etat, text width=1cm] (D4) [right of=D3] {$\Delta_4$} ;
 %Description des liens
 \path[fleche]
    (D3) edge node [above] {\tiny $0,0$} (D4);
\end{tikzpicture}  
& 
\begin{tikzpicture}[node distance=2cm]
 %Définition des styles
  \tikzstyle{etat}=[text badly centered]
  \tikzstyle{fleche}=[->,>=latex]
 %Description des noeuds
  \node[etat, text width=1cm] (D3) {$\Delta_3$};
  \node[etat, text width=1cm] (D4) [right of=D3] {$\Delta_4$} ;
 %Description des liens
 \path[fleche]
    (D3) edge node [above] {\tiny $0,0$} (D4);
\end{tikzpicture} 
 & \begin{tikzpicture}[node distance=2cm]
 %Définition des styles
  \tikzstyle{etat}=[text badly centered]
  \tikzstyle{fleche}=[->,>=latex]
 %Description des noeuds
  \node[etat, text width=1cm] (D3) {$\Delta_3$};
  \node[etat, text width=1cm] (D4) [right of=D3] {$\Delta_4$} ;
 %Description des liens
 \path[fleche]
    (D3) edge node [above] {\tiny $0,0$} (D4);
\end{tikzpicture}  
& 
\begin{tikzpicture}[node distance=2cm]
 %Définition des styles
  \tikzstyle{etat}=[text badly centered]
  \tikzstyle{fleche}=[->,>=latex]
 %Description des noeuds
  \node[etat, text width=1cm] (D3) {$\Delta_3$};
  \node[etat, text width=1cm] (D4) [right of=D3] {$\Delta_4$} ;
 %Description des liens
 \path[fleche]
    (D3) edge node [above] {\tiny $0,0$} (D4);
\end{tikzpicture} 
\\\hline
% % Ligne 4
3 & $\Delta_4$ & 
\begin{tikzpicture}[node distance=2cm]
 %Définition des styles
  \tikzstyle{etat}=[text badly centered]
  \tikzstyle{fleche}=[->,>=latex]
 %Description des noeuds
  \node[etat, text width=1cm] (D4) {$\Delta_4$};
  \node[etat, text width=1cm] (D8) [right of=D4] {$\Delta_8$} ;
 %Description des liens
 \path[fleche]
    (D4) edge node [above] {\tiny $0,0$} (D8);
\end{tikzpicture}
 & 
\begin{tikzpicture}[node distance=2cm]
 %Définition des styles
  \tikzstyle{etat}=[text badly centered]
  \tikzstyle{fleche}=[->,>=latex]
 %Description des noeuds
  \node[etat, text width=1cm] (D4) {$\Delta_4$};
  \node[etat, text width=1cm] (D8) [right of=D4] {$\Delta_8$} ;
 %Description des liens
 \path[fleche]
    (D4) edge node [above] {\tiny $0,0$} (D8);
\end{tikzpicture}
& 
\begin{tikzpicture}[node distance=2cm]
 %Définition des styles
  \tikzstyle{etat}=[text badly centered]
  \tikzstyle{fleche}=[->,>=latex]
 %Description des noeuds
  \node[etat, text width=1cm] (D4) {$\Delta_4$};
  \node[etat, text width=1cm] (D8) [right of=D4] {$\Delta_8$} ;
 %Description des liens
 \path[fleche]
    (D4) edge node [above] {\tiny $0,0$} (D8);
\end{tikzpicture}
 & 
\begin{tikzpicture}[node distance=2cm]
 %Définition des styles
  \tikzstyle{etat}=[text badly centered]
  \tikzstyle{fleche}=[->,>=latex]
 %Description des noeuds
  \node[etat, text width=1cm] (D4) {$\Delta_4$};
  \node[etat, text width=1cm] (D8) [right of=D4] {$\Delta_8$} ;
 %Description des liens
 \path[fleche]
    (D4) edge node [above] {\tiny $0,0$} (D8);
\end{tikzpicture}
\\\hline
% % Ligne 5
4 & $\Delta_5$ & 
\begin{tikzpicture}[node distance=2cm]
 %Définition des styles
  \tikzstyle{etat}=[text badly centered]
  \tikzstyle{fleche}=[->,>=latex]
 %Description des noeuds
  \node[etat, text width=1cm] (D5) {$\Delta_5$};
  \node[etat, text width=1cm] (D5b) [right of=D5] {$\Delta_5$} ;
 %Description des liens
 \path[fleche]
    (D5) edge node [above] {\tiny $0,0$} (D5b);
\end{tikzpicture} 
 & \begin{tikzpicture}[node distance=2cm]
 %Définition des styles
  \tikzstyle{etat}=[text badly centered]
  \tikzstyle{fleche}=[->,>=latex]
 %Description des noeuds
  \node[etat, text width=1cm] (D5) {$\Delta_5$};
  \node[etat, text width=1cm] (D6) [right of=D5] {$\Delta_6$} ;
 %Description des liens
 \path[fleche]
    (D5) edge node [above] {\tiny $0,0$} (D6);
\end{tikzpicture}  &
 \begin{tikzpicture}[node distance=2cm]
 %Définition des styles
  \tikzstyle{etat}=[text badly centered]
  \tikzstyle{fleche}=[->,>=latex]
 %Description des noeuds
  \node[etat, text width=1cm] (D5) {$\Delta_5$};
  \node[etat, text width=1cm] (D5b) [right of=D5] {$\Delta_5$} ;
 %Description des liens
 \path[fleche]
    (D5) edge node [above] {\tiny $0,0$} (D5b);
\end{tikzpicture} 
 & \begin{tikzpicture}[node distance=2cm]
 %Définition des styles
  \tikzstyle{etat}=[text badly centered]
  \tikzstyle{fleche}=[->,>=latex]
 %Description des noeuds
  \node[etat, text width=1cm] (D5) {$\Delta_5$};
  \node[etat, text width=1cm] (D5b) [right of=D5] {$\Delta_5$} ;
 %Description des liens
 \path[fleche]
    (D5) edge node [above] {\tiny $0,0$} (D5b);
\end{tikzpicture} 
\\\hline
% % Ligne 6
5 & $\Delta_6$
 & 
\begin{tikzpicture}[node distance=2cm]
 %Définition des styles
  \tikzstyle{etat}=[text badly centered]
  \tikzstyle{fleche}=[->,>=latex]
 %Description des noeuds
  \node[etat, text width=1cm] (D6) {$\Delta_6$};
  \node[etat, text width=1cm] (D7) [right of=D6] {$\Delta_7$} ;
 %Description des liens
 \path[fleche]
    (D6) edge node [above] {\tiny $0,0$} (D7);
\end{tikzpicture} 
& 
\begin{tikzpicture}[node distance=2cm]
 %Définition des styles
  \tikzstyle{etat}=[text badly centered]
  \tikzstyle{fleche}=[->,>=latex]
 %Description des noeuds
  \node[etat, text width=1cm] (D6) {$\Delta_6$};
  \node[etat, text width=1cm] (D7) [right of=D6] {$\Delta_7$} ;
 %Description des liens
 \path[fleche]
    (D6) edge node [above] {\tiny $0,0$} (D7);
\end{tikzpicture}
 & 
\begin{tikzpicture}[node distance=2cm]
 %Définition des styles
  \tikzstyle{etat}=[text badly centered]
  \tikzstyle{fleche}=[->,>=latex]
 %Description des noeuds
  \node[etat, text width=1cm] (D6) {$\Delta_6$};
  \node[etat, text width=1cm] (D7) [right of=D6] {$\Delta_7$} ;
 %Description des liens
 \path[fleche]
    (D6) edge node [above] {\tiny $0,0$} (D7);
\end{tikzpicture}
 & 
\begin{tikzpicture}[node distance=2cm]
 %Définition des styles
  \tikzstyle{etat}=[text badly centered]
  \tikzstyle{fleche}=[->,>=latex]
 %Description des noeuds
  \node[etat, text width=1cm] (D6) {$\Delta_6$};
  \node[etat, text width=1cm] (D7) [right of=D6] {$\Delta_7$} ;
 %Description des liens
 \path[fleche]
    (D6) edge node [above] {\tiny $0,0$} (D7);
\end{tikzpicture}
\\\hline
% % Ligne 7
6 & $\Delta_7$ & 
\begin{tikzpicture}[node distance=2cm]
 %Définition des styles
  \tikzstyle{etat}=[text badly centered]
  \tikzstyle{fleche}=[->,>=latex]
 %Description des noeuds
  \node[etat, text width=1cm] (D7) {$\Delta_7$};
  \node[etat, text width=1cm] (D7b) [right of=D7] {$\Delta_7$} ;
 %Description des liens
 \path[fleche]
    (D7) edge node [above] {\tiny $0,0$} (D7b);
\end{tikzpicture}
 & \begin{tikzpicture}[node distance=2cm]
 %Définition des styles
  \tikzstyle{etat}=[text badly centered]
  \tikzstyle{fleche}=[->,>=latex]
 %Description des noeuds
  \node[etat, text width=1cm] (D7) {$\Delta_7$};
  \node[etat, text width=1cm] (D7b) [right of=D7] {$\Delta_7$} ;
 %Description des liens
 \path[fleche]
    (D7) edge node [above] {\tiny $0,0$} (D7b);
\end{tikzpicture}
 & 
\begin{tikzpicture}[node distance=2cm]
 %Définition des styles
  \tikzstyle{etat}=[text badly centered]
  \tikzstyle{fleche}=[->,>=latex]
 %Description des noeuds
  \node[etat, text width=1cm] (D7) {$\Delta_7$};
  \node[etat, text width=1cm] (D7b) [right of=D7] {$\Delta_7$} ;
 %Description des liens
 \path[fleche]
    (D7) edge node [above] {\tiny $0,0$} (D7b);
\end{tikzpicture}
 & 
\begin{tikzpicture}[node distance=2cm]
 %Définition des styles
  \tikzstyle{etat}=[text badly centered]
  \tikzstyle{fleche}=[->,>=latex]
 %Description des noeuds
  \node[etat, text width=1cm] (D7) {$\Delta_7$};
  \node[etat, text width=1cm] (D7b) [right of=D7] {$\Delta_7$} ;
 %Description des liens
 \path[fleche]
    (D7) edge node [above] {\tiny $0,0$} (D7b);
\end{tikzpicture}
\\\hline
% % Ligne 8
7 & $\Delta_8$ & 
\begin{tikzpicture}[node distance=2cm]
 %Définition des styles
  \tikzstyle{etat}=[text badly centered]
  \tikzstyle{fleche}=[->,>=latex]
 %Description des noeuds
  \node[etat, text width=1cm] (D8) {$\Delta_8$};
  \node[etat, text width=1cm] (D8b) [right of=D8] {$\Delta_8$} ;
 %Description des liens
 \path[fleche]
    (D8) edge node [above] {\tiny $0,0$} (D8b);
\end{tikzpicture}
 & 
\begin{tikzpicture}[node distance=2cm]
 %Définition des styles
  \tikzstyle{etat}=[text badly centered]
  \tikzstyle{fleche}=[->,>=latex]
 %Description des noeuds
  \node[etat, text width=1cm] (D8) {$\Delta_8$};
  \node[etat, text width=1cm] (D8b) [right of=D8] {$\Delta_8$} ;
 %Description des liens
 \path[fleche]
    (D8) edge node [above] {\tiny $0,0$} (D8b);
\end{tikzpicture}
 & 
\begin{tikzpicture}[node distance=2cm]
 %Définition des styles
  \tikzstyle{etat}=[text badly centered]
  \tikzstyle{fleche}=[->,>=latex]
 %Description des noeuds
  \node[etat, text width=1cm] (D8) {$\Delta_8$};
  \node[etat, text width=1cm] (D8b) [right of=D8] {$\Delta_8$} ;
 %Description des liens
 \path[fleche]
    (D8) edge node [above] {\tiny $0,0$} (D8b);
\end{tikzpicture}
 & 
\begin{tikzpicture}[node distance=2cm]
 %Définition des styles
  \tikzstyle{etat}=[text badly centered]
  \tikzstyle{fleche}=[->,>=latex]
 %Description des noeuds
  \node[etat, text width=1cm] (D8) {$\Delta_8$};
  \node[etat, text width=1cm] (D8b) [right of=D8] {$\Delta_8$} ;
 %Description des liens
 \path[fleche]
    (D8) edge node [above] {\tiny $0,0$} (D8b);
\end{tikzpicture}
\\\hline
\end{tabular}
}
 \caption{The grid $\mathcal{F}$} 
%\vspace{-0.5cm}
\label{fig:grid}
 \end{center}
 \end{figure} 
 \bigskip
}%
\end{example}
\medskip

\textbf{Step 2.} We make a queue $\mathcal{Q}$ that will contain potential eventualities occurring in $\tabb{}$. The first element of $Q$ is either $\initf$, if $\initf$ is a potential eventuality, or the potential eventuality associated to the first column of the grid defined just above. Let $\xi_i$ be the first element of the queue, so that $\mathcal{Q}(0) = \xi_i$. Then we add to $\mathcal{Q}$ all the potential eventualities following the order of grid's rows and cycling if necessary, that is $\mathcal{Q}(k) = \xi_{((i+k)\bmod m)} $ for $k \in [1,m-1]$.

\textbf{Step 3.}
Let $\Delta$ be one of the states containing $\initf$.
Next, we take the element $\grid(\mathcal{Q}(0),\Delta)$ of the grid. 
The root of $\grid(\mathcal{Q}(0),\Delta)$ is then the root of $\mathfrak{F}$. 
Then we take one-by-one in order all the elements of the rest of the queue and do the following:

Let $\mathcal{Q}(i)$ be the current element of the queue to be treated.
For every dead-end state $w \in \mathfrak{F}$, that is a state without successors, such that $c(w) = \Delta_j$, we add the tree $\grid(\mathcal{Q}(i),\Delta_j)$ by identifying the dead-end state $w$ with the root of $\grid(\mathcal{Q}(i),\Delta_j)$;

\begin{example}
From the grid $\grid$, we can extract in four steps a partial structure $\mathfrak{F}$ realizing all the eventualities (see Figure \ref{fig:PartialStructure}).

\begin{figure}
\begin{center}
\scalebox{0.8}{
\begin{tabular}{|c|c|c|}\hline
Eventualities & Added grid elements & Partial structure $\mathfrak{F}$\\\hline
$\{\xi_1,\xi_2,\xi_3,\xi_4\}$ & $\mathcal{F}(\xi_1,\Delta_1)$ & 
\begin{tikzpicture}[node distance=3cm]
 %Définition des styles
  \tikzstyle{etat}=[text badly centered]
  \tikzstyle{fleche}=[->,>=latex]
 %Description des noeuds
  \node[etat, text width=1cm] (D1) {$\Delta_1$};
  \node[etat, text width=1cm] (D4) [above right of=D1,yshift=-1cm] {$\Delta_4$} ;
  \node[etat, text width=1cm] (D5) [below right of=D1,yshift=1cm] {$\Delta_5$} ;
 %Description des liens
 \path[fleche]
    (D1) edge node [text width=0.5cm] {\tiny $0,0$ $0,1$} (D4)
	 edge node [text width=0.5cm] {\tiny $1,0$ $1,1$} (D5);
\end{tikzpicture}
\\\hline
$\{\xi_2,\xi_3,\xi_4\}$  & $\mathcal{F}(\xi_2,\Delta_4); \mathcal{F}(\xi_2,\Delta_5)$ &
\begin{tikzpicture}[node distance=3cm]
 %Définition des styles
  \tikzstyle{etat}=[text badly centered]
  \tikzstyle{fleche}=[->,>=latex]
 %Description des noeuds
  \node[etat, text width=1cm] (D1) {$\Delta_1$};
  \node[etat, text width=1cm] (D4) [above right of=D1,yshift=-1cm] {$\Delta_4$} ;
  \node[etat, text width=1cm] (D5) [below right of=D1,yshift=1cm] {$\Delta_5$} ;
  \node[etat, text width=1cm] (D6) [right of=D5,node distance=2cm] {$\Delta_6$} ;
  \node[etat, text width=1cm] (D8) [right of=D4,node distance=2cm] {$\Delta_8$} ;
 %Description des liens
 \path[fleche]
    (D1) edge node [text width=0.5cm] {\tiny $0,0$ $0,1$} (D4)
	 edge node [text width=0.5cm] {\tiny $1,0$ $1,1$} (D5)
    (D4) edge node [above] {\tiny $0,0$} (D8)
    (D5) edge node [above] {\tiny $0,0$} (D6);
\end{tikzpicture}
\\\hline
$\{\xi_3,\xi_4\}$  & $\mathcal{F}(\xi_3,\Delta_8); \mathcal{F}(\xi_3,\Delta_6)$ & 
\begin{tikzpicture}[node distance=3cm]
 %Définition des styles
  \tikzstyle{etat}=[text badly centered]
  \tikzstyle{fleche}=[->,>=latex]
 %Description des noeuds
  \node[etat, text width=1cm] (D1) {$\Delta_1$};
  \node[etat, text width=1cm] (D4) [above right of=D1,yshift=-1cm] {$\Delta_4$} ;
  \node[etat, text width=1cm] (D5) [below right of=D1,yshift=1cm] {$\Delta_5$} ;
  \node[etat, text width=1cm] (D6) [right of=D5,node distance=2cm] {$\Delta_6$} ;
  \node[etat, text width=1cm] (D8) [right of=D4,node distance=2cm] {$\Delta_8$} ;
  \node[etat, text width=1cm] (D7) [right of=D6,node distance=2cm] {$\Delta_7$} ;
  \node[etat, text width=1cm] (D8b) [right of=D8,node distance=2cm] {$\Delta_8$} ;
 %Description des liens
 \path[fleche]
    (D1) edge node [text width=0.5cm] {\tiny $0,0$ $0,1$} (D4)
	 edge node [text width=0.5cm] {\tiny $1,0$ $1,1$} (D5)
    (D4) edge node [above] {\tiny $0,0$} (D8)
    (D5) edge node [above] {\tiny $0,0$} (D6)
    (D6) edge node [above] {\tiny $0,0$} (D7)
    (D8) edge node [above] {\tiny $0,0$} (D8b);
\end{tikzpicture}
\\\hline
$\{\xi_4\}$  & $\mathcal{F}(\xi_4,\Delta_8); \mathcal{F}(\xi_4,\Delta_8)$ & 
\begin{tikzpicture}[node distance=3cm]
 %Définition des styles
  \tikzstyle{etat}=[text badly centered]
  \tikzstyle{fleche}=[->,>=latex]
 %Description des noeuds
  \node[etat, text width=1cm] (D1) {$\Delta_1$};
  \node[etat, text width=1cm] (D4) [above right of=D1,yshift=-1cm] {$\Delta_4$} ;
  \node[etat, text width=1cm] (D5) [below right of=D1,yshift=1cm] {$\Delta_5$} ;
  \node[etat, text width=1cm] (D6) [right of=D5,node distance=2cm] {$\Delta_6$} ;
  \node[etat, text width=1cm] (D8) [right of=D4,node distance=2cm] {$\Delta_8$} ;
  \node[etat, text width=1cm] (D7) [right of=D6,node distance=2cm] {$\Delta_7$} ;
  \node[etat, text width=1cm] (D8b) [right of=D8,node distance=2cm] {$\Delta_8$} ;
    \node[etat, text width=1cm] (D7a) [right of=D7,node distance=2cm] {$\Delta_7$} ;
    \node[etat, text width=1cm] (D8d) [right of=D8b,node distance=2cm] {$\Delta_8$} ;
 %Description des liens
 \path[fleche]
    (D1) edge node [text width=0.5cm] {\tiny $0,0$ $0,1$} (D4)
	 edge node [text width=0.5cm] {\tiny $1,0$ $1,1$} (D5)
    (D4) edge node [above] {\tiny $0,0$} (D8)
    (D5) edge node [above] {\tiny $0,0$} (D6)
    (D6) edge node [above] {\tiny $0,0$} (D7)
    (D8) edge node [above] {\tiny $0,0$} (D8b)
    (D7) edge node [above] {\tiny $0,0$} (D7a)
    (D8b) edge node [above] {\tiny $0,0$} (D8d);
\end{tikzpicture}
\\\hline
\end{tabular}
}
 \caption{Eventualities and partial structures} 
\vspace{-0.5cm}
\label{fig:PartialStructure}
 \end{center}
 \end{figure} 

\end{example}

\textbf{Step 4.} Finally, we ensure that $\mathfrak{F}$ finite. While there is a dead-end state in $\mathfrak{F}$, say $w$ with $c(w)=\Delta_j$, we choose a component from the row $\grid(\Delta_j)$ as follows:
\begin{itemize}
 \item With priority we choose a component $\grid(i,\Delta_j), 0\leq i \leq m$ already occurring in $\mathfrak{F}$. Let $r$ be the root of the component $\grid(i,\Delta_j), 0\leq i \leq m$  inside $\mathfrak{F}$. Then we add an arrow $\rightsquigarrow$ between every predecessor $v$ of $w$ and the root $r$
 and labelled this arrow with $l(v \rightsquigarrow w)$. 
 Then we delete the node $w \in \mathfrak{F}$.
 \item Otherwise, if the chosen component $\grid(i,\Delta_j)$ is not already occurring in $\mathfrak{F}$ then we add the new component to $\mathfrak{F}$ as usual by merging the root of the component with the dead-end state $w$.
\end{itemize}

When there are no longer dead-ends in $\mathfrak{F}$, the structure is completed and we have obtained our final structure.

\begin{example}
The final structure $\mathfrak{F}$ for the formula $\vartheta = \diams{1}(p \until q \vee \Box q) \wedge \crochet{2}(\Diamond p \wedge \Box \neg q)$ is given in Figure \ref{fig:FinalStructure}.

\begin{figure}
 \begin{center}
	\begin{tikzpicture}[node distance=3cm]
	 %Définition des styles
	  \tikzstyle{etat}=[text badly centered]
	  \tikzstyle{fleche}=[->,>=latex]
	 %Description des noeuds
	  \node[etat, text width=1cm] (D1) {$\Delta_1$};
	  \node[etat, text width=1cm] (D4) [above right of=D1,yshift=-1cm] {$\Delta_4$} ;
	  \node[etat, text width=1cm] (D5) [below right of=D1,yshift=1cm] {$\Delta_5$} ;
	  \node[etat, text width=1cm] (D6) [right of=D5,node distance=2cm] {$\Delta_6$} ;
	  \node[etat, text width=1cm] (D8) [right of=D4,node distance=2cm] {$\Delta_8$} ;
	  \node[etat, text width=1cm] (D7) [right of=D6,node distance=2cm] {$\Delta_7$} ;
	  \node[etat, text width=1cm] (D8b) [right of=D8,node distance=2cm] {$\Delta_8$} ;
	 %Description des liens
	 \path[fleche]
	    (D1) edge node [text width=0.5cm] {\tiny $0,0$ $0,1$} (D4)
		 edge node [text width=0.5cm] {\tiny $1,0$ $1,1$} (D5)
	    (D4) edge node [above] {\tiny $0,0$} (D8)
	    (D5) edge node [above] {\tiny $0,0$} (D6)
	    (D6) edge node [above] {\tiny $0,0$} (D7)
	    (D8) edge node [above] {\tiny $0,0$} (D8b)
	    (D7) edge [loop right] node [above] {\tiny $0,0$} (D7)
	    (D8b) edge [loop right] node [above] {\tiny $0,0$} (D8b);
	\end{tikzpicture}
\end{center}	
\caption{The final structure $\mathfrak{F}$ for the formula $\vartheta = \diams{1}(p \until q \vee \Box q) \wedge \crochet{2}(\Diamond p \wedge \Box \neg q)$}
\label{fig:FinalStructure}
\end{figure}

By keeping only the propositional variables in the state labels we obtain the following concurrent game model satisfying $\vartheta$ in Figure \ref{fig:model}.

\begin{figure}
 \begin{center}
	\begin{tikzpicture}[node distance=3cm]
	 %Définition des styles
	  \tikzstyle{etat}=[circle,draw,text badly centered]
	  \tikzstyle{lalabel}=[node distance=0.7cm]
	  \tikzstyle{fleche}=[->,>=latex]
	 %Description des noeuds
	  \node[etat] (D1) {$S_1$};  \node [lalabel] (l1) [left of=D1] {$\{p\}$};
	  \node[etat] (D4) [above right of=D1,yshift=-1cm] {$S_2$} ;  \node [lalabel] (l2) [above of=D4] {$\{q\}$};
	  \node[etat] (D5) [below right of=D1,yshift=1cm] {$S_3$} ; \node [lalabel] (l3) [above of=D5] {$\emptyset$};
	  \node[etat] (D6) [right of=D5,node distance=2cm] {$S_4$} ; \node [lalabel] (l4) [above of=D6] {$\{p\}$};
	  \node[etat] (D8) [right of=D4,node distance=2cm] {$S_5$} ; \node [lalabel] (l5) [above of=D8] {$\{p,q\}$};
	  \node[etat] (D7) [right of=D6,node distance=2cm] {$S_6$} ; \node [lalabel] (l6) [above of=D7] {$\emptyset$};
	  \node[etat] (D8b) [right of=D8,node distance=2cm] {$S_7$} ; \node [lalabel] (l7) [above of=D8b] {$\{p,q\}$};
	 %Description des liens
	 \path[fleche]
	    (D1) edge node [text width=0.5cm] {\tiny $0,0$ $0,1$} (D4)
		 edge node [text width=0.5cm] {\tiny $1,0$ $1,1$} (D5)
	    (D4) edge node [above] {\tiny $0,0$} (D8)
	    (D5) edge node [above] {\tiny $0,0$} (D6)
	    (D6) edge node [above] {\tiny $0,0$} (D7)
	    (D8) edge node [above] {\tiny $0,0$} (D8b)
	    (D7) edge [loop right] node [above] {\tiny $0,0$} (D7)
	    (D8b) edge [loop right] node [above] {\tiny $0,0$} (D8b);
	\end{tikzpicture}
	\end{center}
\caption{Concurrent game model satisfying $\vartheta$}
\label{fig:model}
\end{figure}
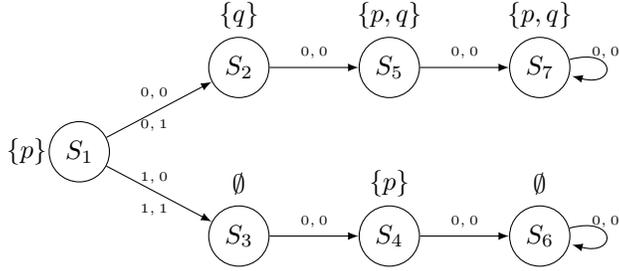
\end{example}

\begin{lemma}
\label{lem:H5}
 Let $\mathcal{T}$ be a $\tabb{}$-tree rooted at $\Delta=c(w)$. Then, the following holds:

\begin{enumerate}
 \item If $\diaAn\varphi \in \Delta$, then there exists an $A$-action $\sigma_A \in \vectorsCoalatstate{A}{\Delta}$ such that $\varphi \in c(w')=\Delta'$ where $l(w\rightsquigarrow w') \ni \sigma$ for every $\sigma \sqsupseteq \sigma_A$.
 \item If $\crochetAn\varphi \in \Delta$ then there exists a co-$A$-action $\sigma_A^c \in \ensCoVect{A}{\Delta}$ such that $\varphi \in c(w')=\Delta'$ where $l(w\rightsquigarrow w') \ni \sigma$ for every $\sigma \sqsupseteq \sigma_A^c(\sigma_A)$.
\end{enumerate}
\end{lemma}

\begin{proof}
 We recall that all successor formulae of $\Delta \in \settab{}$ are ordered at the application of the \rnext\ to $\Delta$.

(1) Suppose that $\diaAn\varphi \in \Delta$. Then the required $A$-action is $\sigma_A[\diaAn\varphi]$. Indeed, it immediately follows from the \rnext\ that for every $\sigma\sqsupseteq\sigma_A$ in the pretableau $\pretab$, if $\Delta\stackrel{\sigma}{\longrightarrow}\Gamma$, then $\varphi\in\Gamma$ and $\varphi \in \Delta'$ since $\Delta'$ is a full expansion of $\Gamma$. The statement (1) of the lemma follows.

(2) Suppose that $\crochetAn\varphi \in \Delta$. There are two cases to consider:

case 1: $A\neq\agents$. We consider an arbitrary $\sigma_A \in \vectorsCoalatstate{A}{\Delta}$. Then $\sigma_A$ can be extended to a action vector $\sigma'\sqsupseteq\sigma$. Let $N(\sigma_A)$ be the set $\{i\mid\sigma_A(i) \geq m\}$, where $m$ is the number of enforceable successor formulae in $\Delta$, and let $\co(\sigma_A) = \left(\sum_{i\in N(\sigma_A)}(\sigma_A(i)-m))\right)\bmod{l}$, where $l$ is the number of unavoidable successor formulae in $\Delta$. Now, we consider $\sigma'\sqsupseteq\sigma_A$ defined as follows: $\sigma'_b = ((q - \co(\sigma_A))\bmod{l}+m$ and $\sigma'_{a'}=m$ for any $a' \in \agents - (A \cup \{b\})$, where $b \in \agents - A$. Thus, we have $\agents - A \subseteq N(\sigma)$ and also $\co(\sigma') = (\co(\sigma_A) + (q - \co(\sigma_A)))\bmod l) + m = q$. Therefore, for this arbitrarily chosen $\sigma_A$ there exists at least one state, say $\Delta'$, such that $\Delta \stackrel{\sigma}{\longrightarrow}\Delta'$ and $\varphi\in\Delta'$. 

Case 2: $A=\agents$. Then, by virtue of (H2), $\diamsn{\emptyset}\neg\varphi\in\Delta$ and thus, by the \rnext, $\neg\varphi\in\Gamma$ for every successor $Gamma$ of $\Delta$. Then, $\neg\varphi\in\Delta'$ for every $\Delta'$ that is a successor of $\Delta$ in $\tabb{}$ and hence the colouring set of every leaf of $\mathcal{T}$. Then, the (unique) co-$\agents$-actions, which is an identity function, has the required properties.

The statement (2) of the lemma follows. \qed
\end{proof}

\begin{theorem}
 If $\tabb{}$ is open, then $\initf$ is satisfiable.
\end{theorem}

\begin{proof}
The structure $\mathfrak{F}$ constructed from $\tabb{}$ is a Hintikka structure. Indeed, H1-H4 of Definition \ref{def:hintikka} are satisfied since the nodes of $\mathfrak{F}$ are nodes of $\tabb{}$. H5 of the same definition essentially follows from Lemma \ref{lem:H5}. Whenever a node $w$ of $\mathfrak{F}$ contains a potential eventuality $\xi$, this means that this eventuality will stay in the queue (see construction of $\mathfrak{F}$ above) until realized. Moreover, if the $\tabb{}$-tree $\treeW$ chosen to complete $\mathfrak{F}$ from $w$ does not realize $\xi$, either $\xi$ or one of its descendants is present in each newly generated dead-end of $\mathfrak{F}$. So, when it is the turn to realize $\xi$ we add to each dead-end state the realizing $\tabb{}$-tree for $\xi$. This, together with Lemma \ref{lem:H5}, guarantees that there exists a realization witness tree for $\xi$ on $\mathfrak{F}$ at $w$. Thus, H6 of Definition \ref{def:hintikka} is satisfied, too. 

By construction, the structure $\mathfrak{F}$ is a concurrent game Hintikka structure
for $\initf$, thus  Theorem \ref{thm:FromCgsToModel} can be applied to obtain from it a model for $\initf$. Thus  $\initf$ is satisfiable.
 \qed
\end{proof}

%%%%%%%%%%%%%%%%%
\subsection{Complexity}
\label{sec:complexity}
%%%%%%%%%%%%%%%%%

\begin{theorem}
The tableau procedure for \ATLp runs in \twoexptime. 
\end{theorem}
\begin{proof}
The argument generally follows the calculations computing the complexity of the tableau method for \ATL in Section 4.7 of \cite{GorankoShkatov09ToCL}, with one essential difference: $\| cl(\initf)\|$ for any \ATL formula $\initf$ is linear in its length $|\initf |$,  
whereas $\| cl(\initf)\|$ 
for an \ATLp formula $\initf$ can be exponentially large in $|\initf |$, as shown after Lemma \ref{lem:closure}. This exponential blow-up, combined with the worst-case exponential in $\| cl(\initf)\|$ number of states in the tableau, accounts for the 
\twoexptime\ worst-case complexity of the tableau method for \ATLp, which is the expected optimal lower bound. It is also an upper bound for the tableau method, because no further exponential blow-ups occur in the prestate- and state-elimination phases. \qed
\end{proof}

There are various ways to restrict or parametrize the set of \ATLp formulae in order to avoid the exponential blow-up of their closure sets.  As suggested by the example after Lemma \ref{lem:closure}, the main cause for that blow-up of the number of $\gamma$-components of $\gamma$-formulae $\varphi = \diaA\Phi$ or $\varphi = \crochetA\Phi$ in \ATLp is the nesting of conjunctions and disjunctions in the path 
formula $\Phi$ which are not separated by 
%in the scope of 
temporal operators. Let us call the number of such nestings the \emph{superficial Boolean depth} of $\Phi$ and denote it by $\delta_{0}(\Phi)$. Then, let the \emph{nested  Boolean depth} of any \ATLp formula $\Psi$, denoted $\delta(\Psi)$, be the maximal superficial Boolean depth $\delta_{0}(\Phi)$ of a path sub-formula $\Phi$ of $\Psi$. For instance, 
$\delta(\ddiams{1}(p \until \lnot q)) = 0$, 
$\delta(\diams{1}(\Box p \lor ((q \land p)\until \lnot q)) = 1$, $\delta(\diams{1}(\Diamond q \land (\Box p \land (q \until \lnot q))) = 2$. 
Now, if this number for a formula $\initf $ is bounded above, the size of the closure 
$\|\initf \|$ becomes polynomially bounded in $|\initf |$ because the nesting of $\land$ and $\lor$ when they are separated by a temporal operator does not have a multiplicative effect on the number of $\gamma$-components. Consequently, the complexity of the tableau method is reduced to single exponential time, caused only by the maximal possible number of states in the tableau, just like in \ATL. Thus, we have the following. 

\begin{proposition}%% Final tableau
The tableau procedure for \ATLp applied to a class of \ATLp formulae of bounded nested Boolean depth  runs in \exptime. 
\end{proposition}

\begin{corollary}%% Final tableau
The tableau procedure for \ATLp applied to \ATL formulae runs in \exptime. 
\end{corollary}

%%%%%%%%%%%%%%%%%%%%%%%%%%%%%%%%%%%%%%%%
%%%%%%%%%%%%%%%%%%%%%%%%%%%%%%%%%%%%%%%%
\section{Concluding remarks}
\label{sec:concluding}
%%%%%%%%%%%%%%%%%%%%%%%%%%%%%%%%%%%%%%%%
%%%%%%%%%%%%%%%%%%%%%%%%%%%%%%%%%%%%%%%%

Here we have developed sound, complete and terminating tableau-based decision method for constructive satisfiability testing of \ATLp formulae and have argued for its practical usability and implementability. The method is amenable to further extension to the full \ATLs, but this is left to future work. 

Some comparison with the automata-based method for satisfiability testing in \ATLs, presented in \cite{Schewe08} are in order. The two methods appear to be quite different and, though eventually working in the same worst-case complexity, the double exponential blow-ups seem to occur in different ways, namely, in the automata-based method, one exponential blow-up occurs in converting the formula into an automaton,
 while the other is in the time complexity of checking non-emptiness of the resulting automaton. It would be instructive to  compare the practical implications and efficiency of both methods and we leave such systematic comparison to the future, when (hopefully) both methods are  implemented. For now, we only mention that the formula $\theta$ from our running example, the tableau for which is worked out
 explicitly and in detail in this paper, is translated with the method from \cite{Schewe08} into an automaton with $2^{12}$ alphabet symbols and over 100 states. Of course, this comparison cannot serve 
as an argument for general practical superiority in efficiency of the tableau-based method. Still, the technical details of both methods, illustrated in that example, indicate that, while the worst-case
 exponential blow-ups are bound to occur in both methods, they seem to be more controllable and avoidable in the tableau-based method, at the expense of its lesser automaticity and higher degree of user control. Thus, we would argue that both methods have generally incomparable pros and cons, and consequently are of independent interest, both theoretically and practically.      

\medskip
\textbf{Acknowledgements.} We thank the anonymous reviewers of  \cite{CDG-IJCAR2014} and of this paper for their helpful remarks and suggestions and for several corrections. 

% Bibliography
\bibliographystyle{splncs03}
\bibliography{VG-ATL}

\begin{thebibliography}{10}
\providecommand{\url}[1]{\texttt{#1}}
\providecommand{\urlprefix}{URL }

\bibitem{AHK02}
Alur, R., Henzinger, T.A., Kupferman, O.: Alternating-time temporal logic.
  Journal of the ACM  49(5),  672--713 (2002)

\bibitem{BPM83}
Ben-Ari, M., Pnueli, A., Manna, Z.: The temporal logic of branching time. Acta
  Informatica  20,  207--226 (1983)

\bibitem{CDG-IJCAR2014}
Cerrito, S., David, A., Goranko, V.: Optimal tableaux-based decision procedure
  for testing satisfiability in the alternating-time temporal logic {ATL+}. In:
  Proc. of IJCAR'2014. LNCS, vol. 8562, pp. 277--291. Springer (2014)

\bibitem{DBLP:conf/tableaux/David13}
David, A.: {TATL}: Implementation of {ATL} tableau-based decision procedure.
  In: Proc. of TABLEAUX'2013, Springer LNCS 8123. pp. 97--103 (2013)

\bibitem{vanDrimmelen03}
van Drimmelen, G.: Satisfiability in alternating-time temporal logic. In:
  Proceedings of the Eighteenth Annual IEEE Symposium on Logic in Computer
  Science (LICS 2003). pp. 208--217. IEEE Computer Society Press (June 2003)

\bibitem{EmHal85}
Emerson, E., Halpern, J.: Decision procedures and expressiveness in the
  temporal logic of branching time. J. of Computation and System Sciences
  30(1),  1--24 (1985)

\bibitem{Emerson90}
Emerson, E.A.: Temporal and modal logics. In: van Leeuwen, J. (ed.) Handbook of
  Theoretical Computer Science, vol.~B, pp. 995--1072. MIT Press (1990)

\bibitem{GorDrim06}
Goranko, V., van Drimmelen, G.: Complete axiomatization and decidablity of
  {A}lternating-time temporal logic. Theor. Comp. Sci.  353,  93--117 (2006)

\bibitem{GorankoShkatov09ToCL}
Goranko, V., Shkatov, D.: Tableau-based decision procedures for logics of
  strategic ability in multiagent systems. ACM Trans. Comput. Log.  11(1),
  1--49 (2009)

\bibitem{Johannsen&Lange2003}
Johannsen, J., Lange, M.: {CTL}+ is {C}omplete for {D}ouble {E}xponential
  {T}ime. In: Proc. of ICALP'03. LNCS, vol. 2719, pp. 767--775. Springer (2003)

\bibitem{Pratt80}
Pratt, V.R.: A near optimal method for reasoning about action. Journal of
  Computer and System Sciences  20,  231--254 (1980)

\bibitem{Schewe08}
Schewe, S.: {ATL}* satisfiability is 2{EXPTIME}-complete. In: Proc. of ICALP
  (Part 2). LNCS, vol. 5126, pp. 373--385. Springer (2008)

\bibitem{WLWW06}
Walther, D., Lutz, C., Wolter, F., Wooldridge, M.: {A}{T}{L} satisfiability is
  indeed {E}xp{T}ime-complete. Journal of Logic and Computation  16(6),
  765--787 (2006)

\bibitem{Wolper85}
Wolper, P.: The tableau method for temporal logic: an overview. Logique et
  Analyse  28(110--111),  119--136 (1985)

\end{thebibliography}

\end{document}